\documentclass[11pt,reqno]{article}

\usepackage{adjustbox}
\usepackage{amsmath}
\usepackage{amssymb}
\usepackage{amsthm}
\usepackage{authblk}
\usepackage{bm}
\usepackage[bf,footnotesize]{caption}
\usepackage[shortlabels]{enumitem}
\usepackage{eucal}
\usepackage{fullpage}
\usepackage{graphicx}
\usepackage{lineno}
\usepackage[sc]{mathpazo}
\usepackage{mathrsfs}
\usepackage{mathtools}
\usepackage{multirow}
\usepackage[numbers,sort&compress,square]{natbib}
\usepackage{pdflscape}
\usepackage{rotating}
\usepackage{setspace}
\usepackage{subfig}
\usepackage{xcolor}
\usepackage{xr}
\usepackage[all,cmtip]{xy}

\newcommand*\patchAmsMathEnvironmentForLineno[1]{%
\expandafter\let\csname old#1\expandafter\endcsname\csname #1\endcsname
\expandafter\let\csname oldend#1\expandafter\endcsname\csname end#1\endcsname
\renewenvironment{#1}%
{\linenomath\csname old#1\endcsname}%
{\csname oldend#1\endcsname\endlinenomath}}%
\newcommand*\patchBothAmsMathEnvironmentsForLineno[1]{%
\patchAmsMathEnvironmentForLineno{#1}%
\patchAmsMathEnvironmentForLineno{#1*}}%
\AtBeginDocument{%
\patchBothAmsMathEnvironmentsForLineno{equation}%
\patchBothAmsMathEnvironmentsForLineno{align}%
\patchBothAmsMathEnvironmentsForLineno{flalign}%
\patchBothAmsMathEnvironmentsForLineno{alignat}%
\patchBothAmsMathEnvironmentsForLineno{gather}%
\patchBothAmsMathEnvironmentsForLineno{multline}%
}

\theoremstyle{definition}
\newtheorem{corollary}{Corollary}
\newtheorem{definition}{Definition}
\newtheorem{example}{Example}
\newtheorem{lemma}{Lemma}
\newtheorem{proposition}{Proposition}
\newtheorem{remark}{Remark}
\newtheorem{theorem}{Theorem}

\newtheorem*{theorem*}{Theorem}

\newcommand{\eq}[1]{\textbf{Eq.~\ref{eq:#1}}}
\newcommand{\fig}[1]{\textbf{Fig.~\ref{fig:#1}}}

\DeclareMathOperator*{\argmax}{argmax}
\DeclareMathOperator*{\argmin}{argmin}

\title{\textbf{Expectation-enforcing strategies for repeated games}}
\author{Nikos Dimou$^{1}$ and Alex McAvoy$^{2,3,4,}$\footnote{Please direct correspondence to A.M. (\texttt{amcavoy@unc.edu}).}\\
\footnotesize{$^{1}$Department of Statistics \& Operations Research, University of North Carolina at Chapel Hill, Chapel Hill, NC \\
$^{2}$School of Data and Information Sciences, University of North Carolina at Chapel Hill, Chapel Hill, NC \\
$^{3}$Department of Mathematics, University of North Carolina at Chapel Hill, Chapel Hill, NC} \\
$^{4}$NSF-Simons National Institute for Theory and Mathematics in Biology, Chicago, IL}

\date{}

\begin{document}

\allowdisplaybreaks

\maketitle

\begin{abstract}
Originating in evolutionary game theory, the class of ``zero-determinant'' strategies enables a player to unilaterally enforce linear payoff relationships in simple repeated games. An upshot of this kind of payoff constraint is that it can shape the incentives for the opponent in a predetermined way. An example is when a player ensures that the agents get equal payoffs. While extensively studied in infinite-horizon games, extensions to discounted games, non-affine constraints, richer strategic environments, and behaviors with long memory remain incompletely understood. In this paper, we provide necessary and sufficient conditions for a player to enforce arbitrary constraints (affine or not), in expectation, in discounted games. These conditions characterize precisely which payoff relationships are enforceable using strategies of arbitrary complexity. Our main result establishes that any such enforceable relationship can actually be implemented using a simple two-point reactive learning strategy, which conditions on the opponent's most recent action and the player's own previous mixed action, using information from only one round into the past. For additive payoff constraints, we show that enforcement is possible using even simpler (reactive) strategies that depend solely on the opponent's last move. In other words, this tractable class is universal within expectation-enforcing strategies. As examples, we apply these results to characterize extortionate, generous, and fair strategies in the repeated prisoner's dilemma, nonlinear donation game, asymmetric donation game, and the hawk-dove game, identifying precisely when each class of strategy is enforceable and with what minimum discount factor. We conclude with an example of enforcing a non-affine constraint that induces inequalities on expected payoffs rather than knife-edge conditions.
\end{abstract}

\section{Introduction}
A foundational question in repeated games is how much an individual can shape long-run outcomes without imposing equilibrium discipline on others. Classical work, notably the Folk Theorem \citep{fudenberg:E:1986}, characterizes the set of payoffs sustainable in equilibrium, provided the horizon of the game is sufficiently long. In contrast, far less is known about the unilateral problem: how much can a player control the space of possible outcomes, even when the opponent behaves arbitrarily? The space of feasible payoffs against a fixed strategy provides outcome-relevant information not always apparent from the behavioral mechanics of a strategy itself. Importantly, this perspective gives insights into the incentives an opponent faces, whether it be in an adaptive (evolutionary) setting or one of reinforcement learning. In fact, despite the rich space of subgame perfect equilibria in repeated games, reinforcement learning illustrates that gradual modifications of one's strategy based on self-interested objectives (i.e., caring about maximizing one's own payoff) leads to inefficient equilibria \citep{foerster:AAMAS:2018,mcavoy:Nexus:2022}, which raises a question about how to productively shape incentives.

Although we are ultimately interested in repeated games in a general setting, some of the motivation of our approach comes from evolutionary game theory, where there is a broad body of work on cooperation in simple repeated games like the prisoner's dilemma. The mechanism by which cooperation emerges in the repeated prisoner's dilemma (``direct reciprocity'' \citep{nowak:Science:2006}) is predicated on the fact that defection now can be punished in the future and cooperation now can be reciprocated in the future \citep{trivers:TQRB:1971}. The underlying stage game has two actions, ``cooperate'' ($C$) and ``defect'' ($D$), and two agents facing one another receive payoffs according to the matrix
\begin{align}
\bordermatrix{%
& C & D \cr
C &\ R,\ R & \ S,\ T \cr
D &\ T,\ S & \ P,\ P \cr
} , \label{eq:payoff_matrix}
\end{align}
where $T>R>P>S$. This payoff ranking ensures that defection is the dominant action, despite mutual cooperation yielding the socially optimal outcome.

For this kind of game, \citet{press:PNAS:2012} described a new class of ``zero-determinant'' (ZD) strategies that permit a striking level of control over expected linear payoff outcomes. In the infinitely repeated prisoner's dilemma, ZD strategies allow player $X$ to choose constants $\alpha$, $\beta$, and $\gamma$ and unilaterally enforce the equation $\alpha\pi_{X}+\beta\pi_{Y}+\gamma =0$, where $\pi_{X}$ and $\pi_{Y}$ are the long-term average payoffs of $X$ and $Y$, respectively. This equation is ``enforced'' by $X$ in the sense that it holds regardless of the opponent's behavior in the repeated game. Of course, not all linear relationships are feasible, but several surprising classes of relationships are enforceable in the repeated prisoner's dilemma, including those in the family $\kappa -\pi_{X}=\chi\left(\kappa -\pi_{Y}\right)$ for $\chi\geqslant 1$ and $\kappa\in\left[P,R\right]$. Of particular interest are strategies enforcing $\pi_{X}-P=\chi\left(\pi_{Y}-P\right)$ for $\chi >1$, which allow $X$ to effectively extort the opponent and claim an unfair share of payoffs beyond those of mutual defection. A concrete example of a ZD strategy is tit-for-tat (TFT) in the (infinite-horizon) repeated prisoner's dilemma \citep{axelrod:Science:1981,axelrod:BB:1984}, which enforces equal payoffs ($\chi =1$).

Since their discovery, ZD strategies have been investigated in numerous social dilemmas, including repeated public goods \citep{pan:SR:2015} and (nonlinear) donation games \citep{mcavoy:PNAS:2016}, and their properties have been studied in the contexts of discounted games \citep{hilbe:GEB:2015, ichinose2018zero}, games with continuous action spaces \citep{mcavoy:PNAS:2016}, alternating games \citep{mcavoy:TPB:2017}, evolutionary environments \citep{adami:NC:2013,hilbe:PLOSONE:2013,stewart:PNAS:2013}, multiplayer games \citep{hilbe:PNAS:2014}, human experiments \citep{wang:NC:2016,milinski:NC:2016}, and stochastic, multi-state games \citep{mcavoy:PNAS:2025}. However, the study of ZD strategies remains incomplete. Most work investigates linear payoff relationships in games where both players have access to only two actions. It also is not understood exactly what kinds of payoff constraints can be enforced, even in simple games.

Furthermore, ZD strategies are typically derived under the assumption that one is searching within the class of ``memory-one'' strategies, which specify the player's reaction depending on the outcome of the previous round only. This class of strategies strikes a balance between mathematical tractability and behavioral complexity, and many famous strategies (including tit-for-tat and win-stay, lose-shift \citep{nowak:Nature:1993}) condition on only the previous round's outcome. It has been shown that longer finite memory, which is computationally intractable \citep{hilbe2017memory}, is not always needed against reasonable opponents, and strategies of reduced memory suffice to produce optimal payoff outcomes in some specific game structures \citep{lesigang2025can,barlo2009repeated}. Even though longer-memory ZD strategies have been developed \citep{ueda2021memory}, any advantage they have over memory-one strategies remains unclear.

An additional constraint is imposed by discounting, which weights payoffs from future rounds less than current payoffs through a discount factor $\lambda\in\left[0,1\right)$. Discounting reflects time preferences or uncertainty about future interactions, and it restricts the set of enforceable payoff relationships: as $\lambda$ decreases (placing less weight on the future), fewer relationships can be enforced \citep{hilbe:GEB:2015}. Given the restricted range of feasible ZD strategies in discounted games and the limited understanding of longer-memory properties, the following questions naturally arise:
\begin{itemize}
\item How much control does a player have against an adaptive opponent?
\item Can a player enforce additional payoff relationships by extending memory?
\end{itemize}

This paper addresses both questions and provides a definitive answer to what is possible and what is implementable by a single player in repeated games. More precisely, we adopt the framework of ``autocratic'' strategies, which generalize ZD strategies and allow for non-affine constraints. We provide necessary and sufficient conditions of which payoff relationships can be enforced using strategies of arbitrary complexity, thus establishing a complete characterization on the level of control that a player is able to exert against any adaptive opponent. Our main contribution gives a (perhaps surprising) answer to the second question: extending memory beyond a simple reactive learning structure provides no additional power. In fact, any enforceable payoff relationship can be implemented using a two-point reactive learning strategy. Reactive learning strategies formally generalize memory-one strategies by conditioning on the opponent's most recent action and the player's own previous mixed action (rather than realized action), allowing the player to track their own randomization history. A two-point reactive learning strategy further simplifies this by mixing between just two (fixed) possible responses based on the opponent's last action. In addition, for additive constraint functions, we show that enforcement is possible using even simpler reactive strategies that depend solely on the opponent's last action. These results demonstrate that reactive learning strategies are universal within the class of all strategies that endow a player with the ability to control expected payoff outcomes.

For any candidate payoff constraint, we give a concrete ``next-round correction'' condition that is both necessary and sufficient for enforcement.
Furthermore, we provide an explicit formula for the minimum discount factor required to enforce a given payoff relationship. This result resolves a number of open problems. First, the problem of identifying the minimum discount factor in games with a more complex structure than that of the repeated prisoner's dilemma has appeared previously in several works, e.g., in \citep{govaert2019zero}. Second, it answers an open question raised by \citet{hilbe:PLOSONE:2013} regarding the existence of autocratic strategies in discounted games, and extends prior work \citep{hilbe:GEB:2015} by providing exact thresholds for strategies of arbitrary memory.

In contrast to the approach that resulted in the discovery of ZD strategies \citep{press:PNAS:2012}, we do not rely on the standard method of studying the transition probabilities of memory-one strategies adopted by the focal player, as this technique turns out to be prohibitive when dealing with strategies of broad cognitive complexity. Instead, we explore the dynamics of the repeated game by identifying a mechanism that allows one to control the stochastic path of the game and correct past ``suboptimal'' behaviors or errors. Although we focus on discounted games, which are more realistic, we also extend our techniques to the infinite-horizon regime.

The main theoretical results are applied to several classical games. In the repeated prisoner's dilemma, we provide exact conditions for extortionate, generous, and fair strategies, computing the minimum discount factor required for each class. Our results verify and extend prior work \citep{hilbe:GEB:2015,mamiya2020zero} by providing constructive formulas. The generality of our framework also enables important non-existence results. For instance, we prove that symmetric relationships (such as $\pi_{X}=\pi_{Y}$ in the infinitely repeated prisoner's dilemma) are either enforceable for all time horizons (including one-shot games) or else only in games with infinite horizons, resolving the question of when fair strategies exist. This shows that tit-for-tat's ability to enforce equal payoffs is fundamentally limited to the infinite-horizon setting. Beyond symmetric, two-action games, we characterize autocratic strategies in a multi-action nonlinear donation game, an asymmetric donation game, and the hawk-dove game. The nonlinear donation game demonstrates that our framework extends naturally to multi-action games. The asymmetric donation game illustrates the distinction between equality (enforcing $\pi_{X}=\pi_{Y}$) and fairness (proportional sharing based on action costs), showing that only the former can sometimes be enforced unilaterally. All of the above deal with enforcing affine constraints on payoffs, and we conclude by showing that non-affine constraints can generate feasible regions of full dimension (rather than lines).

A final important algorithmic consequence of our characterization is that verifying whether a given payoff relationship is enforceable, as well as computing the minimum discount factor and constructing an associated autocratic strategy, can be accomplished in polynomial time using linear programming. This stands in stark contrast to the computational intractability of analyzing general behavioral strategies \citep{hauert1997effects,hilbe2017memory}. Our results thus provide both theoretical closure on the memory question of enforcing payoff constraints, as well as practical tools for computing autocratic strategies.

The remainder of this paper is organized as follows. In Section~\ref{sec:background}, we review the framework of repeated games and introduce the notion of autocratic strategies, establishing key background results including the pointwise and generalized ``next-round correction'' conditions. In Section~\ref{sec:autocratic}, we present our main theoretical contributions: we prove that any enforceable payoff relationship can be implemented using a two-point reactive learning strategy (Theorem~\ref{th:mainresultDiscounted}), characterize the minimum discount factor required for enforcement (Proposition~\ref{prop:lambda_min}), and show that additive payoff constraints admit even simpler reactive implementations (Theorem~\ref{thm:additive}). We also extend our results to the undiscounted setting ($\lambda\rightarrow 1$). In Section~\ref{sec:properties}, we establish computational and structural properties of autocratic strategies, including polynomial-time algorithms for verification and construction. Finally, in Section~\ref{sec:applications}, we apply our framework to four examples of social dilemmas \citep{dawes:ARP:1980}: the repeated prisoner's dilemma, a nonlinear donation game, an asymmetric donation game, and the hawk-dove game, providing complete characterizations of all linear payoff relationships that can be unilaterally enforced, as well as the minimum average time horizon needed to do so. For the repeated prisoner's dilemma, we also show how non-affine constraints can lead to interesting regions of full dimension (i.e., not lines).

\section{Background and auxiliary results}\label{sec:background}

\subsection{Repeated games and discounting}
We consider repeated games between two players, $X$ and $Y$, with finite action spaces $S_{X}$ and $S_{Y}$, respectively. The players receive short-term payoffs via functions $u_{X}:S_{X}\times S_{Y}\rightarrow\mathbb{R}$ and $u_{Y}:S_{X}\times S_{Y}\rightarrow\mathbb{R}$. In each round $t\in\left\{0,1,2,\dots\right\}$, players simultaneously choose actions $\left(s_{X}^{t},s_{Y}^{t}\right)\in S_{X}\times S_{Y}$ from a distribution and receive stage-game payoffs $u_{X}\left(s_{X}^{t},s_{Y}^{t}\right)$ and $u_{Y}\left(s_{X}^{t},s_{Y}^{t}\right)$. The game continues with probability $\lambda\in\left[0,1\right)$ after each round. The discounted payoff for player $X$ over a realization of action outcomes is
\begin{align*}
\pi_{X} = \left(1-\lambda\right)\sum_{t=0}^{\infty}\lambda^{t} u_{X}\left(s_{X}^{t},s_{Y}^{t}\right) ,
\end{align*}
with a similar expression for player $Y$, with $u_{Y}$ replacing $u_{X}$. The factor $1-\lambda$ normalizes payoffs to be comparable across different continuation probabilities.

\subsection{Behavioral strategies and histories}
Let $\mathcal{H}^{T}=\left(S_{X}\times S_{Y}\right)^{T}$ denote the set of all possible histories of length $T$, representing the sequence of action pairs played from round $0$ through round $T-1$. We write $\mathcal{H} = \bigcup_{T\geqslant 0}\mathcal{H}^{T}$ for the set of all finite histories, where $\mathcal{H}^{0} = \left\{\varnothing\right\}$ represents the null history at the start of the game.

Let $\Delta\left(S_{X}\right)$ and $\Delta\left(S_{Y}\right)$ denote the corresponding spaces of mixed actions (distributions over pure actions). A behavioral strategy for player $X$ is a map $\sigma_{X}:\mathcal{H}\rightarrow\Delta\left(S_{X}\right)$ that specifies a mixed action for each possible history. Given strategies $\sigma_{X}$ and $\sigma_{Y}$, we write $\mathbb{E}_{\sigma_{X},\sigma_{Y}}\left[\cdot\right]$ for the expected value with respect to the probability distribution over infinite sequences of realized play induced by $\sigma_{X}$ and $\sigma_{Y}$. Notably, the mean long-term payoffs for $X$ and $Y$, respectively, are
\begin{align*}
\pi_{X} &\coloneqq \mathbb{E}_{\sigma_{X},\sigma_{Y}}\left[\left(1-\lambda\right)\sum_{t=0}^{\infty}\lambda^{t}u_{X}\left(s_{X}^{t},s_{Y}^{t}\right)\right] ; \\
\pi_{Y} &\coloneqq \mathbb{E}_{\sigma_{X},\sigma_{Y}}\left[\left(1-\lambda\right)\sum_{t=0}^{\infty}\lambda^{t}u_{Y}\left(s_{X}^{t},s_{Y}^{t}\right)\right] .
\end{align*}

A behavioral strategy is a ``memory-one'' strategy if its response depends only on the most recent round of play. Formally, a memory-one strategy for player $X$ consists of \emph{(i)} an initial mixed action $\sigma_{X}^{0}\in\Delta\left(S_{X}\right)$ and \emph{(ii)} a response rule $\sigma_{X}:S_{X}\times S_{Y}\rightarrow\Delta\left(S_{X}\right)$ that specifies a mixed action $\sigma_{X}\left[s_{X},s_{Y}\right]\in\Delta\left(S_{X}\right)$ for each action pair $\left(s_{X},s_{Y}\right)$ in the previous round. Importantly, the response depends only on the realized actions $\left(s_{X},s_{Y}\right)$, not on how those actions were generated through randomization. Let $\textbf{Mem}_{X}^{1}$ denote the set of all memory-one strategies for player $X$.

\subsection{Reactive learning strategies}\label{sec:RL_strategies}
Reactive learning strategies, introduced by \citet{mcavoy:PRSA:2019}, formally generalize memory-one strategies by allowing a player to condition on their own mixed action from the previous round.

\begin{definition}
A reactive learning strategy for $X$ consists of an initial action $\sigma_{X}^{0}\in\Delta\left(S_{X}\right)$ and a response rule $\sigma_{X}^{\ast}:\Delta\left(S_{X}\right)\times S_{Y}\rightarrow\Delta\left(S_{X}\right)$.
\end{definition}

Let $\textbf{RL}_{X}$ denote the set of all reactive learning strategies for $X$. Every memory-one strategy naturally induces a reactive learning strategy via the canonical embedding
\begin{align*}
^{\ast} &: \textbf{Mem}_{X}^{1} \longrightarrow \textbf{RL}_{X} \nonumber \\
&: \left(\sigma_{X}^{0},\sigma_{X}\right) \longmapsto \left(\sigma_{X}^{0},\sigma_{X}^{\ast}\right) ,
\end{align*}
where the initial actions are the same, and $\sigma_{X}^{\ast}\left[\tau_{X},s_{Y}\right]\left(\cdot\right) \coloneqq \mathbb{E}_{s_{X}\sim\tau_{X}}\left[\sigma_{X}\left[s_{X},s_{Y}\right]\left(\cdot\right)\right]$ over $S_{X}$. Conversely, restricting a reactive learning strategy to Dirac measures recovers a memory-one response function, $\sigma_{X}\left[s_{X},s_{Y}\right]\coloneqq\sigma_{X}^{\ast}\left[\delta_{s_{X}},s_{Y}\right]$, where $\delta_{s_{X}}$ is the point mass at $s_{X}$.

\begin{definition}
A sequence $\left\{\tau_{X}^{t}\right\}_{t=0}^{\infty}\subseteq\Delta\left(S_{X}\right)$ is a ``chain'' of mixed actions derived from $\left(\sigma_{X}^{0},\sigma_{X}^{\ast}\right)$ if $\tau_{X}^{0}=\sigma_{X}^{0}$ and, for every $t\geqslant 1$, there exists an action $s_{Y}^{t-1}\in S_{Y}$ such that $\tau_{X}^{t}=\sigma_{X}^{\ast}\left[\tau_{X}^{t-1},s_{Y}^{t-1}\right]$. We can then assemble chains into the set of reachable mixed actions,
\begin{align*}
\mathcal{M}_{X} \coloneqq \bigcap \left\{\mathcal{M}\subseteq\Delta\left(S_{X}\right)\mid \sigma_{X}^{0}\in\mathcal{M}\textrm{ and } \sigma_{X}^{\ast}\left[\tau_{X},s_{Y}\right]\in\mathcal{M}\text{ for all }\tau_{X}\in\mathcal{M}\textrm{ and }s_{Y}\in S_{Y}\right\}.
\end{align*}
\end{definition}
That is, $\mathcal{M}_{X}$ is the smallest subset of $\Delta\left(S_{X}\right)$ containing the initial action $\sigma_{X}^{0}$ and closed under the response map $\tau_{X}\mapsto\sigma_{X}^{\ast}\left[\tau_{X},s_{Y}\right]$ for every $s_{Y}\in S_{Y}$.

We have framed reactive learning strategies as generalizing memory-one strategies by replacing $S_{X}$ with $\Delta\left(S_{X}\right)$ in the response function. This construction might appear to give strategies outside of the standard class of behavioral strategies, but, in fact, a reactive learning strategy is a particular kind of behavioral strategy. To see why, let $\mathcal{H}_{Y}\coloneqq\bigcup_{t\geqslant 0}S_{Y}^{t}$ be the opponent-action history space (where we interpret $S_{Y}^{0}=\left\{\varnothing\right\}$ as the ``empty'' history). Consider the map taking a reactive learning strategy, $\left(\sigma_{X}^{0},\sigma_{X}^{\ast}\right)$, to a behavioral strategy, $\mathcal{U}\left(\sigma_{X}^{0},\sigma_{X}^{\ast}\right):\mathcal{H}_{Y}\rightarrow\Delta\left(S_{X}\right)$,
defined by $\mathcal{U}\left(\sigma_{X}^{0},\sigma_{X}^{\ast}\right)\left[\varnothing\right] =\sigma_{X}^{0}$ and $\mathcal{U}\left(\sigma_{X}^{0},\sigma_{X}^{\ast}\right)\left[s_{Y}^{0},\dots ,s_{Y}^{t}\right] =\tau_{X}^{t+1}$, where $\tau_{X}^{i+1}=\sigma_{X}^{\ast}\left[\tau_{X}^{i},s_{Y}^{i}\right]$ for $i=0,\dots ,t$ and $\tau_{X}^{0}=\sigma_{X}^{0}$. The map $\mathcal{U}$ is not surjective because elements in the image of $\mathcal{U}$ satisfy the specific property that if two distinct histories prescribe the same randomization for player $X$ at time $t$, then they must prescribe the same randomization for $X$ at all $T\geqslant t$ when $Y$ uses the same actions in both sequences thereafter.

We can re-frame this problem slightly. We say that $\sigma_{X}$ is ``right-invariant'' if, whenever $\sigma_{X}\left[\alpha\right] =\sigma_{X}\left[\beta\right]$ for $\alpha ,\beta\in\mathcal{H}_{Y}$, we have $\sigma_{X}\left[\alpha\gamma\right] =\sigma_{X}\left[\beta\gamma\right]$ for all $\gamma\in\mathcal{H}_{Y}$. One can check that if $\sigma_{X}:\mathcal{H}_{Y}\rightarrow\Delta\left(S_{X}\right)$ is in the image of $\mathcal{U}$, then $\sigma_{X}$ is right-invariant. Conversely, if $\sigma_{X}:\mathcal{H}_{Y}\rightarrow\Delta\left(S_{X}\right)$ is right-invariant, then $\sigma_{X}$ lifts to a canonical element in the domain of $\mathcal{U}$. By construction, this map provides more than an abstract identification. It gives an equivalence in behavior against any kind of opponent. Thus, reactive learning strategies are simply behavioral strategies with structure that allows histories to be ``rolled up'' into richer information that gets carried from round to round. Structurally, therefore, reactive learning strategies have a memory length of one round, despite being distinct from the standard notion of ``memory-one'' strategy in repeated games \cite{barlo2009repeated}.

\subsection{Feasible and enforceable payoffs}
The feasible region is the set of all payoff pairs $\left(\pi_{Y},\pi_{X}\right)$ that can arise from some pair of strategies (note the unusual ordering, which we fix throughout). For a generic two-player game, this forms a convex subset of $\mathbb{R}^{2}$ of full dimension. When a player commits to a strategy, $\sigma_{X}$, the achievable payoffs as $\sigma_{Y}$ varies form a convex subset of the feasible region. For a generic strategy $\sigma_{X}$, characterizing the geometry of achievable payoff pairs $\left(\pi_{Y},\pi_{X}\right)$ as $\sigma_{Y}$ ranges over all opponent strategies requires detailed knowledge of the repeated game dynamics. There is currently no simple way to determine this payoff region, apart from in special cases, such as two-action games with infinite horizon, where this region is the convex hull of at most $11$ points \citep{mcavoy:PRSA:2019}.

\subsection{Autocratic strategies}
The following definition formalizes the notion of unilateral enforcement of expectations:
\begin{definition}\label{def:autocratic}
Let $\lambda\in [0,1)$ and $\varphi:S_{X}\times S_{Y}\rightarrow\mathbb{R}$ be a fixed function on the space of joint actions. A strategy $\sigma_{X}:\mathcal{H}\rightarrow\Delta\left(S_{X}\right)$ is $\left(\varphi,\lambda\right)$-autocratic if, for all opponent strategies $\sigma_{Y}:\mathcal{H}\rightarrow\Delta\left(S_{Y}\right)$,
\begin{align*}
\mathbb{E}_{\sigma_{X},\sigma_{Y}}\left[\left(1-\lambda\right)\sum_{t=0}^{\infty}\lambda^{t}\varphi\left(s_{X}^{t},s_{Y}^{t}\right)\right] = 0 .
\end{align*}
\end{definition}

Intuitively, an autocratic strategy unilaterally enforces a constraint on the expected value of $\varphi$, regardless of the opponent's behavior. When $\varphi\left(s_{X},s_{Y}\right) = \alpha u_{X}\left(s_{X},s_{Y}\right) + \beta u_{Y}\left(s_{X},s_{Y}\right) + \gamma$, this corresponds to enforcing the linear payoff relationship $\alpha\pi_{X} + \beta\pi_{Y} + \gamma = 0$, which is exactly the motivation behind ZD strategies \citep{press:PNAS:2012} (see \fig{averagingStrategies}). We say that $\varphi$ is affine if it can be expressed as $\alpha u_{X}\left(s_{X},s_{Y}\right) + \beta u_{Y}\left(s_{X},s_{Y}\right) + \gamma$ for some $\alpha ,\beta ,\gamma\in\mathbb{R}$. General functions $\varphi$ need not be expressible in this way, and in such a case we call the constraint non-affine. We avoid the term ``nonlinear constraint'' because, whereas pinning an affine $\varphi$ to zero enforces a linear constraint on expected payoffs, doing so for a non-affine constraint generally does not enforce a nonlinear constraint on expected payoffs. Indeed, relationships such as $\pi_{X}=\pi_{Y}^{2}$ cannot generally be enforced unilaterally.

\begin{figure}
    \centering
    \includegraphics[width=\textwidth]{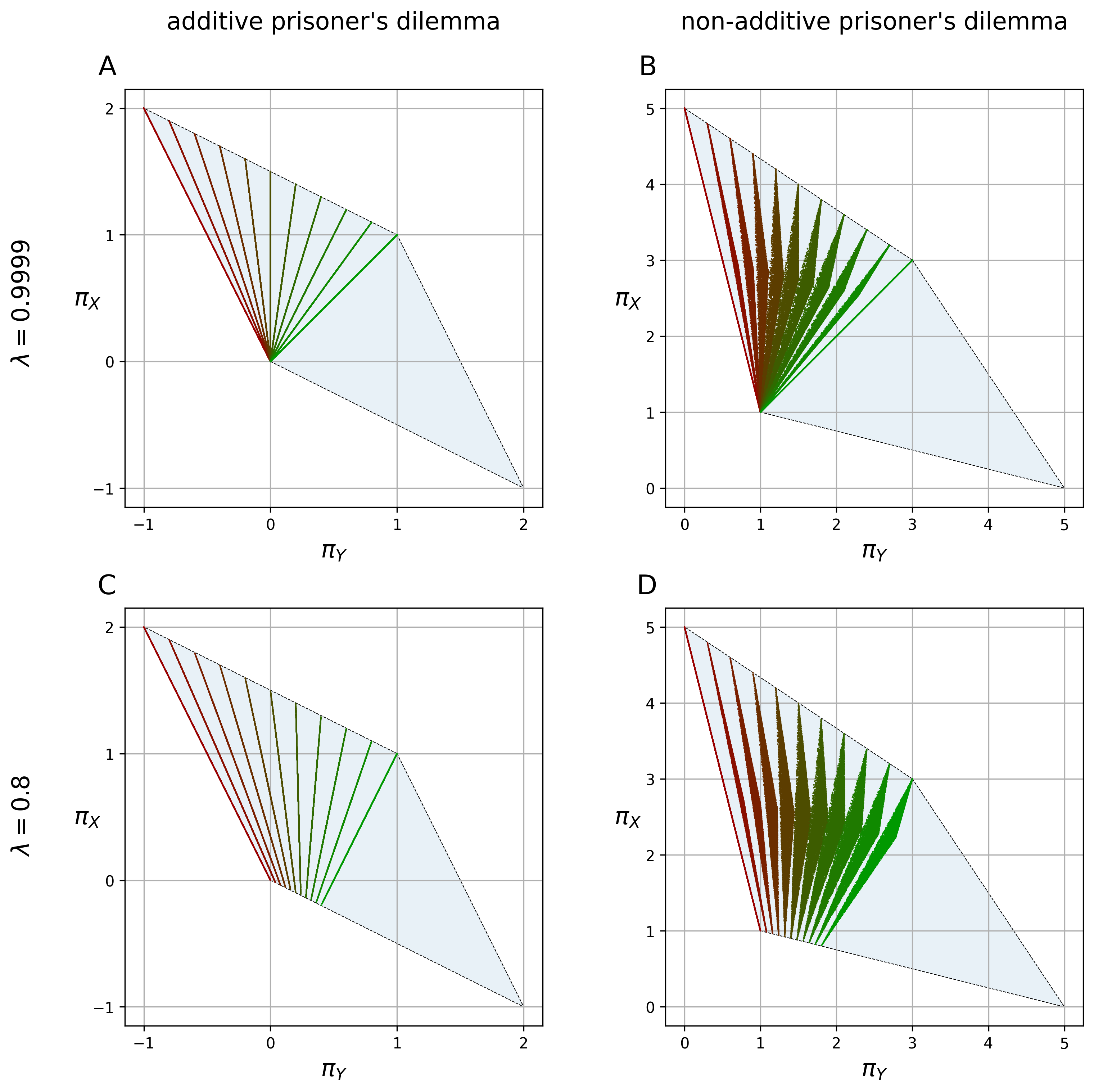}
    \caption{Payoff regions enforced when $X$ plays weighted averages of ALLD (red) and TFT (green) in repeated prisoner's dilemmas. Each colored region shows the payoff region obtained from the strategy $\sigma_{X}\coloneqq\left(1-p\right)\textrm{ALLD}+p\textrm{TFT}$ played against $10^{4}$ randomly-chosen opposing strategies, for $p\in\left\{k/10\right\}_{k=0}^{10}$. (A,C)~Additive prisoner's dilemma with $\left(R,S,T,P\right)=\left(1,-1,2,0\right)$, where all $p\in\left[0,1\right]$ give a line. (B,D)~Non-additive prisoner's dilemma with $\left(R,S,T,P\right)=\left(3,0,5,1\right)$: the strategy enforces a two-dimensional convex region, in general. While line-enforcing strategies naturally arise in additive games through simple mixtures of well-known strategies, non-additive games require more sophisticated constructions. Panels A and B use $\lambda =0.9999$ (a game with $10{,}000$ rounds, on average, approximating an undiscounted game), and panels C and D use $\lambda =0.8$ (a game with $5$ rounds, on average).\label{fig:averagingStrategies}}
\end{figure}

\begin{definition}\label{def:enforceable}
For $\varphi:S_{X}\times S_{Y}\rightarrow\mathbb{R}$ and $\lambda\in\left[0,1\right)$, we say that $\varphi\equiv 0$ is $\lambda$-enforceable if there exists a $\left(\varphi,\lambda\right)$-autocratic strategy $\sigma_{X}:\mathcal{H}\rightarrow\Delta\left(S_{X}\right)$. We say $\varphi\equiv 0$ is enforceable if it is $\lambda$-enforceable for some $\lambda\in\left[0,1\right)$.
\end{definition}

\subsection{The pointwise next-round correction condition}
The following result from \citet{mcavoy:PNAS:2016} provides a sufficient condition for a memory-one strategy to enforce a linear payoff relationship.

\begin{theorem}[\citet{mcavoy:PNAS:2016}]\label{thm:mcavoy-hauert}
Suppose that $\left(\sigma_{X}^{0},\sigma_{X}\left[s_{X},s_{Y}\right]\right)$ is a memory-one strategy for $X$. If there exists a function $\psi:S_{X}\rightarrow\mathbb{R}$ such that
\begin{align}\label{eq:mainEquation2016}
\alpha u_{X}\left(s_{X},s_{Y}\right) + \beta u_{Y}\left(s_{X},s_{Y}\right) + \gamma = \psi\left(s_{X}\right) - \lambda\mathbb{E}_{s_{X}'\sim\sigma_{X}\left[s_{X},s_{Y}\right]}\left[\psi\left(s_{X}'\right)\right] - \left(1-\lambda\right)\mathbb{E}_{s_{X}'\sim\sigma_{X}^{0}}\left[\psi\left(s_{X}'\right)\right]
\end{align}
holds for every $s_{X}\in S_{X}$ and $s_{Y}\in S_{Y}$, then $\left(\sigma_{X}^{0},\sigma_{X}\left[s_{X},s_{Y}\right]\right)$ enforces the linear payoff relationship
\begin{align}\label{eq:linearRelationship2016}
\alpha\pi_{X}+\beta\pi_{Y}+\gamma = 0
\end{align}
against any behavioral strategy of player $Y$, including those with infinite memory.
\end{theorem}

We call \eq{mainEquation2016} the pointwise next-round correction condition, and we call $\psi$ the enforcement potential. This condition provides a local, action-by-action characterization that guarantees the global payoff constraint \eq{linearRelationship2016}.

\begin{example}[Tit-for-tat in the undiscounted prisoner's dilemma]\label{ex:TFT}
Consider the prisoner's dilemma with payoff matrix \eq{payoff_matrix} in the undiscounted setting ($\lambda \rightarrow 1$). The well-known strategy tit-for-tat (TFT) plays $C$ initially and then copies the opponent's previous action. That is, $\sigma_{X}^{0}=C$ and $\sigma_{X}\left[s_{X},s_{Y}\right] = s_{Y}$ for all $s_{X},s_{Y} \in \left\{C,D\right\}$. It is known that TFT enforces the fair relationship $\pi_{X} = \pi_{Y}$, or equivalently, $\varphi \equiv 0$ where $\varphi\left(s_{X},s_{Y}\right) =u_{Y}\left(s_{X},s_{Y}\right) -u_{X}\left(s_{X},s_{Y}\right)$ \citep{press:PNAS:2012}. To verify this using the pointwise next-round correction condition (\eq{mainEquation2016}), we seek an enforcement potential $\psi: \left\{C,D\right\} \rightarrow \mathbb{R}$ such that
\begin{align*}
\varphi\left(s_{X},s_{Y}\right) = \psi\left(s_{X}\right) - \psi\left(\sigma_{X}\left[s_{X},s_{Y}\right]\right)
\end{align*}
for all $s_{X},s_{Y}\in\left\{C,D\right\}$. This equation simplifies to $\varphi\left(s_{X},s_{Y}\right) =\psi\left(s_{X}\right) - \psi\left(s_{Y}\right)$ since TFT satisfies $\sigma_{X}\left[s_{X},s_{Y}\right] =s_{Y}$. Taking $\psi\left(C\right) = 0$ and solving, we obtain $\psi\left(D\right) = S - T$. By Theorem \ref{thm:mcavoy-hauert} (extended to $\lambda = 1$), TFT enforces $\pi_{X} = \pi_{Y}$ against any opponent strategy. This example illustrates how the enforcement potential $\psi$ captures the ``correction'' that each action provides toward achieving the target payoff relationship.
\end{example}

A natural question is whether the pointwise next-round correction condition is also necessary for enforcing linear payoff relationships. The next example shows that it is not, in general:
\begin{example}\label{ex:PNRCC-not-necessary}
Consider a two-player, two-action game with payoff matrix
\begin{align*}
\bordermatrix{%
& C & D \cr
C &\ -1 & \ -2 \cr
D &\ +1 & \ +2 \cr
} .
\end{align*}
By playing $C$ and $D$ with equal probability in each round, player $X$ can ensure $\pi_{X}=0$ regardless of $Y$'s strategy. Specifically, the constant mixed strategy with $\sigma_{X}^{0}\left(C\right)=\sigma_{X}^{0}\left(D\right)=1/2$ and $\sigma_{X}\left[s_{X},s_{Y}\right]=\sigma_{X}^{0}$ for all $s_{X},s_{Y}\in\left\{C,D\right\}$ enforces $\pi_{X}=0$ for all $\lambda\in\left[0,1\right)$.

However, this strategy does not satisfy the pointwise next-round correction condition (\eq{mainEquation2016}) for every action pair $\left(s_{X},s_{Y}\right)$ because $u_{X}\left(s_{X},s_{Y}\right)$ depends on $s_{Y}$. To see this more explicitly, suppose there exists $\psi:S_{X}\rightarrow\mathbb{R}$ such that \eq{mainEquation2016} holds with $\alpha =1$ and $\beta =\gamma =0$. By translating, we may assume $\psi\left(C\right) =0$. The pointwise next-round correction condition then requires:
\begin{align*}
-1 &= - \lambda\sigma_{X}\left[C,C\right]\left(D\right)\psi\left(D\right) - \left(1-\lambda\right)\sigma_{X}^{0}\left(D\right)\psi\left(D\right); \\
-2 &= - \lambda\sigma_{X}\left[C,D\right]\left(D\right)\psi\left(D\right) - \left(1-\lambda\right)\sigma_{X}^{0}\left(D\right)\psi\left(D\right); \\
+1 &= \psi\left(D\right) - \lambda\sigma_{X}\left[D,C\right]\left(D\right)\psi\left(D\right) - \left(1-\lambda\right)\sigma_{X}^{0}\left(D\right)\psi\left(D\right); \\
+2 &= \psi\left(D\right) - \lambda\sigma_{X}\left[D,D\right]\left(D\right)\psi\left(D\right) - \left(1-\lambda\right)\sigma_{X}^{0}\left(D\right)\psi\left(D\right).
\end{align*}
This system has no solution for $\psi\left(D\right)$, which shows that the pointwise next-round correction condition (\eq{mainEquation2016}) is sufficient but not necessary for enforcing linear payoff relationships.
\end{example}

\subsection{The generalized next-round correction condition}
We now extend the pointwise condition to reactive learning strategies. Suppose $\left(\sigma_{X}^{0},\sigma_{X}\left[s_{X},s_{Y}\right]\right)$ satisfies the pointwise next-round correction condition (\eq{mainEquation2016}) for some enforcement potential $\psi:S_{X}\rightarrow\mathbb{R}$. For $\tau_{X}\in\Delta\left(S_{X}\right)$, define the map $\Psi\left(\tau_{X}\right) \coloneqq \mathbb{E}_{s_{X}\sim\tau_{X}}\left[\psi\left(s_{X}\right)\right]$ (which we also refer to as an enforcement potential), and recall the induced reactive learning strategy, $\sigma_{X}^{\ast}\left[\tau_{X},s_{Y}\right]\left(\cdot\right) \coloneqq \mathbb{E}_{s_{X}\sim\tau_{X}}\left[\sigma_{X}\left[s_{X},s_{Y}\right]\left(\cdot\right)\right]$.

Let $\varphi\left(\tau_{X},s_{Y}\right)\coloneqq\mathbb{E}_{s_{X}\sim\tau_{X}}\left[\varphi\left(s_{X},s_{Y}\right)\right]$ be the linear extension of $\varphi$ to mixed actions in the first coordinate. Taking expectations of \eq{mainEquation2016} with respect to $\tau_{X}\in\mathcal{M}_{X}$ yields the generalized next-round correction condition,
\begin{align}\label{eq:mainIntegral}
\varphi\left(\tau_{X},s_{Y}\right) = \Psi\left(\tau_{X}\right) - \lambda\Psi\left(\sigma_{X}^{\ast}\left[\tau_{X},s_{Y}\right]\right) - \left(1-\lambda\right)\Psi\left(\sigma_{X}^{0}\right)
\end{align}
for every $\tau_{X}\in\mathcal{M}_{X}$ and $s_{Y}\in S_{Y}$, where $\varphi\left(s_{X},s_{Y}\right) = \alpha u_{X}\left(s_{X},s_{Y}\right) + \beta u_{Y}\left(s_{X},s_{Y}\right) + \gamma$. Unlike the pointwise condition, we will show that the generalized next-round correction condition is both necessary and sufficient for $\left(\sigma_{X}^{0},\sigma_{X}^{\ast}\right)$ to be autocratic, for general functions $\varphi$. We first need a technical lemma, which shows that certain kinds of autocratic strategies can be characterized by enforcement against deterministic sequences:
\begin{lemma}\label{lem:exogenous}
Consider the opponent-action history space $\mathcal{H}_{Y}\coloneqq\bigcup_{t\geqslant 0}S_{Y}^{t}$, where we interpret $S_{Y}^{0}=\left\{\varnothing\right\}$ as the ``empty'' history. A behavioral strategy $\sigma_{X}:\mathcal{H}_{Y}\rightarrow\Delta\left(S_{X}\right)$ is $\left(\varphi,\lambda\right)$-autocratic if and only if it enforces $\varphi\equiv 0$ with discount factor $\lambda$ against all exogenous, deterministic sequences $\left\{s_{Y}^{t}\right\}_{t=0}^{\infty}\subseteq S_{Y}$.
\end{lemma}
\begin{proof}
The ``only if'' direction follows immediately since exogenous sequences are themselves behavioral strategies. For the converse, suppose $\sigma_{X}:\mathcal{H}_{Y}\rightarrow\Delta\left(S_{X}\right)$ enforces $\varphi\equiv 0$ against all exogenous sequences. Let $\sigma_{Y}:\mathcal{H}\rightarrow\Delta\left(S_{Y}\right)$ be a behavioral strategy. Conditioned on $Y$'s action sequence being $s_{Y}^{0},\dots ,s_{Y}^{t}$ through time $t$, $X$'s own past actions do not influence the time-$t$ distribution of $s_{X}^{t}$. Therefore, the expectation of $\varphi\left(s_{X}^{t},s_{Y}^{t}\right)$ at time $t\geqslant 1$ is $\varphi\left(\sigma_{X}\left[s_{Y}^{0},\dots ,s_{Y}^{t-1}\right] ,s_{Y}^{t}\right)$ and
\begin{align*}
\mathbb{E}_{\sigma_{X},\sigma_{Y}}\left[\sum_{t=0}^{\infty}\lambda^{t}\varphi\left(s_{X}^{t},s_{Y}^{t}\right)\right] &= \mathbb{E}_{\left\{s_{Y}^{t}\right\}_{t=0}^{\infty}\sim\left(\sigma_{X},\sigma_{Y}\right)}\left[\varphi\left(\sigma_{X}\left[\varnothing\right] ,s_{Y}^{0}\right) +\sum_{t=1}^{\infty}\lambda^{t}\varphi\left(\sigma_{X}\left[s_{Y}^{0},\dots ,s_{Y}^{t-1}\right] ,s_{Y}^{t}\right)\right] \nonumber \\
&= \mathbb{E}_{\left\{s_{Y}^{t}\right\}_{t=0}^{\infty}\sim\left(\sigma_{X},\sigma_{Y}\right)}\left[0\right] \nonumber \\
&= 0
\end{align*}
by the hypothesis.
\end{proof}

\begin{remark}
In Lemma~\ref{lem:exogenous}, the hypothesis that $\sigma_{X}$ is a map from $\mathcal{H}_{Y}$ to $\Delta\left(S_{X}\right)$ instead of from $\mathcal{H}$ to $\Delta\left(S_{X}\right)$ is important. For a strategy $\sigma_{X}:\mathcal{H}\rightarrow\Delta\left(S_{X}\right)$, we can already see why at time $t=1$ by looking at the expectations
\begin{align*}
\mathbb{E}_{\sigma_{X},\sigma_{Y}}\left[\varphi\left(s_{X}^{0},s_{Y}^{0}\right)\right] &= \mathbb{E}_{s_{Y}^{0}\sim\sigma_{Y}\left[\varnothing\right]}\left[\varphi\left(\sigma_{X}\left[\varnothing\right] ,s_{Y}^{0}\right)\right] ; \\
\mathbb{E}_{\sigma_{X},\sigma_{Y}}\left[\varphi\left(s_{X}^{1},s_{Y}^{1}\right)\right] &= \sum_{\left(s_{X}^{0},s_{Y}^{0}\right)\in S_{X}\times S_{Y}} \sigma_{X}\left[\varnothing\right]\left(s_{X}^{0}\right)\sigma_{Y}\left[\varnothing\right]\left(s_{Y}^{0}\right) \nonumber \\
&\quad \times\sum_{\left(s_{X}^{1},s_{Y}^{1}\right)\in S_{X}\times S_{Y}} \sigma_{X}\left[s_{X}^{0},s_{Y}^{0}\right]\left(s_{X}^{1}\right)\sigma_{Y}\left[s_{X}^{0},s_{Y}^{0}\right]\left(s_{Y}^{1}\right) \varphi\left(s_{X}^{1},s_{Y}^{1}\right) .
\end{align*}
Conditioned on the action sequence $s_{Y}^{0},s_{Y}^{1}$ for $Y$ at time $1$, we therefore cannot generically write the expectation of $\varphi\left(s_{X}^{1},s_{Y}^{1}\right)$ as $\sum_{s_{X}^{0}\in S_{X}}\sigma_{X}\left[\varnothing\right]\left(s_{X}^{0}\right)\varphi\left(\sigma_{X}\left[s_{X}^{0},s_{Y}^{0}\right] ,s_{Y}^{1}\right)$. The issue is that conditioning on $s_{Y}^{1}$ might preclude some values of $s_{X}^{0}$ altogether, and this issue propagates to larger $t\geqslant 1$. In fact, Lemma~\ref{lem:exogenous} is not true for general $\sigma_{X}:\mathcal{H}\rightarrow\Delta\left(S_{X}\right)$. Indeed, let 
\begin{align*}
\varphi &= \bordermatrix{%
& L & R \cr
U &\ 1 & \ -1 \cr
D &\ -1 & \ 1 \cr
}
\end{align*}
and suppose that $X$ initially plays $U$ and $D$ each with probability $1/2$ and then repeats that action forever. Against any exogenous opponent, $\varphi$ averages to zero over the game. But consider the behavioral opponent who randomizes uniformly between $L$ and $R$ in the first round and then plays $L$ if $X$ played $U$ and $R$ if $X$ played $D$. The expected value of $\varphi$ over the game is
\begin{align*}
\mathbb{E}_{\sigma_{X},\sigma_{Y}}\left[\left(1-\lambda\right)\sum_{t=0}^{\infty}\lambda^{t}\varphi\left(s_{X}^{t},s_{Y}^{t}\right)\right] &= \left(1-\lambda\right)\frac{1}{4}\left(1-1-1+1\right) + \lambda = \lambda ,
\end{align*}
which is nonzero for $\lambda\neq 0$.

Our use of this lemma below applies only to strategies of the form $\sigma_{X}:\mathcal{H}_{Y}\rightarrow\Delta\left(S_{X}\right)$, so the failure of the lemma to extend to all $\sigma_{X}:\mathcal{H}\rightarrow\Delta\left(S_{X}\right)$ does not present any issues.
\end{remark}

The generalized next-round correction condition provides a complete characterization of autocratic reactive learning strategies. Since \eq{mainEquation2016} implies \eq{mainIntegral} by linearity of expectation, the following strengthens the main result of \citet{mcavoy:PNAS:2016}.

\begin{proposition}\label{prop:gnrcc_iff}
Consider a function $\varphi:S_{X}\times S_{Y}\rightarrow\mathbb{R}$ and fix $\lambda\in\left[0,1\right)$. If the generalized next-round correction condition holds for the reactive learning strategy $\left(\sigma_{X}^{0},\sigma_{X}^{\ast}\right)$, and furthermore $\lim_{t\rightarrow\infty}\lambda^{t}\Psi\left(\tau_{X}^{t}\right) =0$ for every chain $\left\{\tau_{X}^{t}\right\}_{t=0}^{\infty}\subseteq\mathcal{M}_{X}$ derived from $\left(\sigma_{X}^{0},\sigma_{X}^{\ast}\right)$, then $\left(\sigma_{X}^{0},\sigma_{X}^{\ast}\right)$ is a $\left(\varphi,\lambda\right)$-autocratic strategy.
\end{proposition}

\begin{proof}
By Lemma~\ref{lem:exogenous} and the preceding discussion in Section~\ref{sec:RL_strategies} linking reactive learning strategies to behavioral strategies $\sigma_{X}:\mathcal{H}_{Y}\rightarrow\Delta\left(S_{X}\right)$, it suffices to show that $\left(\sigma_{X}^{0},\sigma_{X}^{\ast}\right)$ enforces $\varphi\equiv 0$ with discount factor $\lambda$ against all exogenous sequences. Consider an arbitrary sequence $\left\{s_{Y}^{t}\right\}_{t=0}^{\infty}\subseteq S_{Y}$. The strategy $\left(\sigma_{X}^{0},\sigma_{X}^{\ast}\right)$ and this sequence generate a chain of mixed actions $\left\{\tau_{X}^{t}\right\}_{t=0}^{\infty}$ with $\tau_{X}^{0}=\sigma_{X}^{0}$ and $\tau_{X}^{t+1}=\sigma_{X}^{\ast}\left[\tau_{X}^{t},s_{Y}^{t}\right]$ for all $t\geqslant 0$. Let $\mu_{X}^{t}\in\Delta\left(\left(\mathcal{M}_{X}\times S_{Y}\right)^{t+1}\right)$ be the probability distributions defined inductively by $\mu_{X}^{0}$ being the Dirac delta centered at $\left(\sigma_{X}^{0},s_{Y}^{0}\right)$ and, for $t\geqslant 1$,
\begin{align*}
\mu_{X}^{t}\left(\left(\widetilde{\tau}_{X}^{0},\widetilde{s}_{Y}^{0}\right),\ldots ,\left(\widetilde{\tau}_{X}^{t},\widetilde{s}_{Y}^{t}\right)\right) &\coloneqq \mu_{X}^{t-1}\left(\left(\widetilde{\tau}_{X}^{0},\widetilde{s}_{Y}^{0}\right),\ldots , \left(\widetilde{\tau}_{X}^{t-1},\widetilde{s}_{Y}^{t-1}\right)\right) \times\delta_{\sigma_{X}^{\ast}\left[\widetilde{\tau}_{X}^{t-1},\widetilde{s}_{Y}^{t-1}\right],\widetilde{\tau}_{X}^{t}}\times\delta_{s_{Y}^{t}}.
\end{align*}
From this we obtain the marginalized distributions $\left\{\nu_{X}^{t}\right\}_{t=0}^{\infty}$ over $\mathcal{M}_{X}\times S_{Y}$ defined by
\begin{align*}
\nu_{X}^{t}\left(A\times B\right)\coloneqq\mu_{X}^{t}\left(\left(\mathcal{M}_{X}\times S_{Y}\right)^{t}\times\left(A\times B\right)\right) .
\end{align*}
The generalized next-round correction condition (\eq{mainIntegral}) gives
\begin{align}\label{eq:expectedconcrete}
\mathbb{E}_{\left(\tau_{X},s_{Y}\right)\sim\nu_{X}^{t}}\left[\varphi\left(\tau_{X},s_{Y}\right)\right] = \mathbb{E}_{\left(\tau_{X},s_{Y}\right)\sim\nu_{X}^{t}}\left[\Psi\left(\tau_{X}\right)\right] - \lambda\mathbb{E}_{\left(\tau_{X},s_{Y}\right)\sim\nu_{X}^{t}}\left[\Psi\left(\sigma_{X}^{\ast}\left[\tau_{X},s_{Y}\right]\right)\right] - \left(1-\lambda\right)\Psi\left(\sigma_{X}^{0}\right).
\end{align}
Since the exogenous sequence for $Y$ is deterministic, by induction we find that $\nu_{X}^{t}$ is a delta mass on $\left(\tau_{X}^{t},s_{Y}^{t}\right)$, where  $\tau_{X}^{0}=\sigma_{X}^{0}$ and $\tau_{X}^{t}=\sigma_{X}^{\ast}\left[\tau_{X}^{t-1},s_{Y}^{t-1}\right]$ for every $t\geqslant 1$. Thus,
\begin{align}\label{eq:easytransition}
\mathbb{E}_{\left(\tau_{X},s_{Y}\right)\sim\nu_{X}^{t}}\left[\Psi\left(\sigma_{X}^{\ast}\left[\tau_{X},s_{Y}\right]\right)\right] = \mathbb{E}_{\left(\tau_{X},s_{Y}\right)\sim\nu_{X}^{t+1}}\left[\Psi\left(\tau_{X}\right)\right]
\end{align}
for every $t\geqslant 0$. From \eq{expectedconcrete} and \eq{easytransition} we deduce the recursion
\begin{align}\label{eq:expectedconcrete2}
\mathbb{E}_{\left(\tau_{X},s_{Y}\right)\sim\nu_{X}^{t}}\left[\varphi\left(\tau_{X},s_{Y}\right)\right] = \mathbb{E}_{\left(\tau_{X},s_{Y}\right)\sim\nu_{X}^{t}}\left[\Psi\left(\tau_{X}\right)\right] - \lambda\mathbb{E}_{\left(\tau_{X},s_{Y}\right)\sim\nu_{X}^{t+1}}\left[\Psi\left(\tau_{X}\right)\right] - \left(1-\lambda\right)\Psi\left(\sigma_{X}^{0}\right) ,
\end{align}
which gives a telescoping sum. The result follows by multiplying \eq{expectedconcrete2} by $\lambda^{t}$ and summing over $t\geqslant 0$, since the boundary term vanishes by the hypothesis that $\lim_{t\rightarrow\infty}\lambda^{t}\Psi\left(\tau_{X}^{t}\right) =0$.
\end{proof}

We next prove that the generalized next-round correction condition is also a necessary condition, and we consider the uniqueness of the enforcement potential $\Psi$.

\begin{proposition}\label{prop:uniquenessOfPsi}
If $\left(\sigma_{X}^{0},\sigma_{X}^{\ast}\right)$ is a reactive learning strategy that is $\left(\varphi ,\lambda\right)$-autocratic, then there exists $\mathcal{M}_{X}\subseteq\Delta\left(S_{X}\right)$ and an enforcement potential $\Psi :\mathcal{M}_{X}\rightarrow\mathbb{R}$ such that
\begin{enumerate}

\item[\emph{(a)}] $\lim_{t\rightarrow\infty}\lambda^{t}\Psi\left(\tau_{X}^{t}\right) =0$ for every chain, $\left\{\tau_{X}^{t}\right\}_{t=0}^{\infty}\subseteq\mathcal{M}_{X}$, derived from $\left(\sigma_{X}^{0},\sigma_{X}^{\ast}\right)$;

\item[\emph{(b)}] the generalized next-round correction condition holds for all $\tau_{X}\in\mathcal{M}_{X}$ and $s_{Y}\in S_{Y}$;

\item[\emph{(c)}] if $\widetilde{\mathcal{M}}_{X}\subseteq\Delta\left(S_{X}\right)$ and $\widetilde{\Psi}:\widetilde{\mathcal{M}}_{X}\rightarrow\mathbb{R}$ also satisfy $\lim_{t\rightarrow\infty}\lambda^{t}\widetilde{\Psi}\left(\widetilde{\tau}_{X}^{t}\right) =0$ for chains $\left\{\widetilde{\tau}_{X}^{t}\right\}_{t=0}^{\infty}\subseteq\widetilde{\mathcal{M}}_{X}$, and \eq{mainIntegral} holds for all $\tau_{X}\in\widetilde{\mathcal{M}}_{X}$ and $s_{Y}\in S_{Y}$, then $\mathcal{M}_{X}\subseteq\widetilde{\mathcal{M}}_{X}$ and $\widetilde{\Psi}\vert_{\mathcal{M}_{X}}-\Psi$ is constant on $\mathcal{M}_{X}$.

\end{enumerate}
\end{proposition}
\begin{proof}
If $\lambda =0$, then the result is immediate, so we assume that $\lambda >0$. Let $\mathcal{M}_{X}$ be the set of all truncated chains of mixed actions derived from $\sigma_{X}^{\ast}$. In particular, $\mathcal{M}_{X}$ is the set of all actions
$\tau_{X}\in\Delta\left(S_{X}\right)$ for which there exist sequences $\left\{\tau_{X}^{0},\dots ,\tau_{X}^{T}\right\}\subseteq\Delta\left(S_{X}\right)$ and $\left\{s_{Y}^{0},\dots ,s_{Y}^{T-1}\right\}\subseteq S_{Y}$ with $\tau_{X}^{0}=\sigma_{X}^{0}$, $\tau_{X}^{T}=\tau_{X}$, and $\tau_{X}^{t}=\sigma_{X}^{\ast}\left[\tau_{X}^{t-1},s_{Y}^{t-1}\right]$ for $1\leqslant t\leqslant T$. For $\tau_{X}\in\mathcal{M}_{X}$, with sequences $\left\{\tau_{X}^{0},\dots ,\tau_{X}^{T}\right\}\subseteq\Delta\left(S_{X}\right)$ and $\left\{s_{Y}^{0},\dots ,s_{Y}^{T-1}\right\}\subseteq S_{Y}$ chosen as above, let
\begin{align}\label{eq:psi0}
\Psi\left(\tau_{X}\right) &\coloneqq -\lambda^{-T} \sum_{t=0}^{T-1} \lambda^{t} \varphi\left(\tau_{X}^{t},s_{Y}^{t}\right) .
\end{align}
(We define $\Psi\left(\sigma_{X}^{0}\right)\coloneqq 0$.) 
To see that \eq{psi0} gives a well-defined value of $\Psi$ at $\tau_{X}$, suppose that $\left\{\widetilde{\tau}_{X}^{0},\dots ,\widetilde{\tau}_{X}^{\widetilde{T}}\right\}\subseteq\Delta\left(S_{X}\right)$ and $\left\{\widetilde{s}_{Y}^{0},\dots ,\widetilde{s}_{Y}^{\widetilde{T}-1}\right\}\subseteq S_{Y}$ are also sequences with $\widetilde{\tau}_{X}^{0}=\sigma_{X}^{0}$, $\widetilde{\tau}_{X}^{\widetilde{T}}=\tau_{X}$, and $\widetilde{\tau}_{X}^{t}=\sigma_{X}^{\ast}\left[\widetilde{\tau}_{X}^{t-1},\widetilde{s}_{Y}^{t-1}\right]$ for $1\leqslant t\leqslant\widetilde{T}$. For any $s_{Y}^{\ast}\in S_{Y}$, we can extend the sequences $\left\{\tau_{X}^{t}\right\}_{t=0}^{T}$ and $\left\{\widetilde{\tau}_{X}^{t}\right\}_{t=0}^{\widetilde{T}}$ as follows: for each $t\geqslant 1$, let $\tau_{X}^{T+t}\coloneqq\sigma_{X}^{\ast}\left[\tau_{X}^{T+t-1},s_{Y}^{\ast}\right]$ and $\widetilde{\tau}_{X}^{\widetilde{T}+t}\coloneqq\sigma_{X}^{\ast}\left[\widetilde{\tau}_{X}^{\widetilde{T}+t-1},s_{Y}^{\ast}\right]$. Since $\tau_{X}^{T}=\widetilde{\tau}_{X}^{\widetilde{T}}=\tau_{X}$, we have $\tau_{X}^{T+t}=\widetilde{\tau}_{X}^{\widetilde{T}+t}$ for every $t\geqslant 0$. Therefore, by the hypothesis,
\begin{align}\label{eq:wellDefined}
0 &= \sum_{t=0}^{T-1}\lambda^{t}\varphi\left(\tau_{X}^{t},s_{Y}^{t}\right) + \sum_{t=T}^{\infty}\lambda^{t}\varphi\left(\tau_{X}^{t},s_{Y}^{\ast}\right) \nonumber \\
&= \sum_{t=0}^{\widetilde{T}-1}\lambda^{t}\varphi\left(\widetilde{\tau}_{X}^{t},\widetilde{s}_{Y}^{t}\right) + \sum_{t=\widetilde{T}}^{\infty}\lambda^{t}\varphi\left(\widetilde{\tau}_{X}^{t},s_{Y}^{\ast}\right) \nonumber \\
&= \sum_{t=0}^{\widetilde{T}-1}\lambda^{t}\varphi\left(\widetilde{\tau}_{X}^{t},\widetilde{s}_{Y}^{t}\right) + \lambda^{\widetilde{T}-T}\sum_{t=T}^{\infty}\lambda^{t}\varphi\left(\tau_{X}^{t},s_{Y}^{\ast}\right) .
\end{align}
It follows immediately that $\Psi$ is well-defined since, from \eq{wellDefined}, we obtain
\begin{align*}
-\lambda^{-T}\sum_{t=0}^{T-1}\lambda^{t}\varphi\left(\tau_{X}^{t},s_{Y}^{t}\right) &= -\lambda^{-\widetilde{T}}\sum_{t=0}^{\widetilde{T}-1}\lambda^{t}\varphi\left(\widetilde{\tau}_{X}^{t},\widetilde{s}_{Y}^{t}\right) .
\end{align*}
That $\lim_{t\rightarrow\infty}\lambda^{t}\Psi\left(\tau_{X}^{t}\right) =0$ for every chain of mixed actions $\left\{\tau_{X}^{t}\right\}_{t=0}^{\infty}$ derived from $\left(\sigma_{X}^{0},\sigma_{X}^{\ast}\right)$ similarly follows from the fact that $\left(\sigma_{X}^{0},\sigma_{X}^{\ast}\right)$ is $\left(\varphi ,\lambda\right)$-autocratic, which establishes part \emph{(a)}.

For part \emph{(b)}, consider $\tau_{X}\in\mathcal{M}_{X}$. We may assume, without a loss of generality, that $\tau_{X}$ is not the initial action $\sigma_{X}^{0}$, since $\Psi\left(\sigma_{X}^{0}\right) =0$. Then, as above, there exist $T\geqslant 1$ and sequences $\left\{\tau_{X}^{0},\dots ,\tau_{X}^{T}\right\}\subseteq\Delta\left(S_{X}\right)$ and $\left\{s_{Y}^{0},\dots ,s_{Y}^{T-1}\right\}\subseteq S_{Y}$ with $\tau_{X}^{0}=\sigma_{X}^{0}$, $\tau_{X}^{T}=\tau_{X}$, and $\tau_{X}^{t}=\sigma_{X}^{\ast}\left[\tau_{X}^{t-1},s_{Y}^{t-1}\right]$ for $1\leqslant t\leqslant T$. Let $s_{Y}\in S_{Y}$, and set $\tau_{X}^{T+1}\coloneqq\sigma_{X}^{\ast}\left[\tau_{X},s_{Y}\right]$. By the definition of $\Psi$, we have 
\begin{align*}
\Psi\left(\tau_{X}\right) -\lambda\Psi\left(\sigma_{X}^{\ast}\left[\tau_{X},s_{Y}\right]\right) &=
\Psi\left(\tau_{X}^T\right) -\lambda\Psi\left(\tau_{X}^{T+1}\right) \nonumber \\ &=-\lambda^{-T}\sum_{t=0}^{T-1} \lambda^{t}\varphi\left(\tau_{X}^{t},s_{Y}^{t}\right) +\lambda\lambda^{-\left(T+1\right)}\sum_{t=0}^{T} \lambda^{t} \varphi\left(\tau_{X}^{t},s_{Y}^{t}\right) \nonumber \\
&= \varphi\left(\tau_{X},s_{Y}\right) .
\end{align*}
The generalized next-round correction condition follows.

Since $\mathcal{M}_{X}$ is defined as the set of all chains of mixed actions derived from $\left(\sigma_{X}^{0},\sigma_{X}^{\ast}\right)$, it follows that any other such $\widetilde{\mathcal{M}}_{X}$ must contain $\mathcal{M}_{X}$. With $\widehat{\Psi}\coloneqq\widetilde{\Psi}\vert_{\mathcal{M}_{X}}-\widetilde{\Psi}\left(\sigma_{X}^{0}\right)$, we have $\widehat{\Psi}\left(\sigma_{X}^{0}\right) =0$ and
\begin{align*}
\widehat{\Psi}\left(\sigma_{X}^{\ast}\left[\tau_{X},s_{Y}\right]\right) &= - \lambda^{-1}\varphi\left(\tau_{X},s_{Y}\right) + \lambda^{-1}\widehat{\Psi}\left(\tau_{X}\right)
\end{align*}
for every $\tau_{X}\in\mathcal{M}_{X}$ and $s_{Y}\in S_{Y}$. Using this recurrence, it follows by induction that $\widehat{\Psi}=\Psi$, and thus $\widetilde{\Psi}\vert_{\mathcal{M}_{X}}-\Psi=\widetilde{\Psi}\left(\sigma_{X}^{0}\right)$ on $\mathcal{M}_{X}$, which completes the proof of part \emph{(c)}.
\end{proof}

\section{Autocratic strategies with short memory}\label{sec:autocratic}
For most of this section, we assume $\lambda\in\left[0,1\right)$; that is, the game terminates with positive probability $1-\lambda >0$ after each round. Only in Section~\ref{sec:undiscounted}, which covers the undiscounted, infinite-horizon case, do we consider $\lambda\rightarrow 1$.

A basic question for implementing autocratic strategies in practice is: how much memory is required? While Definition~\ref{def:autocratic} allows strategies with arbitrary memory, we show that strategies with shorter memory suffice. In this section, we prove that every autocratic strategy, regardless of its complexity, can be replaced by a two-point reactive learning strategy, which is one that mixes between just two fixed distributions (mixed actions). We first show that autocratic strategies need only condition on the opponent's history (Proposition~\ref{prop:reactive_behavioral}). We then use the fact that such strategies correspond to reactive learning strategies with a right-invariance property, which allows longer action histories to be ``rolled up'' into information that can be carried from round to round. Finally, we construct explicit two-point strategies that enforce any enforceable payoff relationship (Theorem~\ref{th:mainresultDiscounted}).

\begin{proposition}\label{prop:reactive_behavioral}
Suppose that $\sigma_{X}:\mathcal{H}\rightarrow\Delta\left(S_{X}\right)$ is $\left(\varphi ,\lambda\right)$-autocratic. Then, there exists an opponent-conditioned strategy, $\widetilde{\sigma}_{X}:\mathcal{H}_{Y}\rightarrow\Delta\left(S_{X}\right)$, that is also $\left(\varphi ,\lambda\right)$-autocratic. In particular, it always suffices for $X$ to condition on only the observed history of $Y$ alone.
\end{proposition}
\begin{proof}
Suppose that $\sigma_{X}:\mathcal{H}\rightarrow\Delta\left(S_{X}\right)$ is an autocratic strategy for $X$. If the opponent plays an exogenous sequence of pure actions, $\left\{s_{Y}^{t}\right\}_{t=0}^{\infty}$, which we denote by $\sigma_{Y}$, then
\begin{align*}
0 &= \mathbb{E}_{\sigma_{X},\sigma_{Y}}\left[\sum_{t=0}^{\infty} \lambda^{t} \varphi\left(s_{X}^{t},s_{Y}^{t}\right) \right] \nonumber \\
&= \sum_{t=0}^{\infty} \lambda^{t} \mathbb{E}_{\sigma_{X},\sigma_{Y}}\left[\varphi\left(\sigma_{X}\left[\left(s_{X}^{0},s_{Y}^{0}\right) ,\dots ,\left(s_{X}^{t-1},s_{Y}^{t-1}\right)\right] ,s_{Y}^{t}\right) \right] \nonumber \\
&= \sum_{t=0}^{\infty} \lambda^{t} \varphi\left(\widetilde{\sigma}_{X}\left[s_{Y}^{0},\dots ,s_{Y}^{t-1}\right] ,s_{Y}^{t}\right) ,
\end{align*}
where $\widetilde{\sigma}_{X}:\mathcal{H}_{Y}\rightarrow\Delta\left(S_{X}\right)$ is defined by $\widetilde{\sigma}_{X}\left[\varnothing\right]\coloneqq\sigma_{X}\left[\varnothing\right]$ and
\begin{align*}
\widetilde{\sigma}_{X}\left[s_{Y}^{0},\dots ,s_{Y}^{t-1}\right]\left(s_{X}^{t}\right) &\coloneqq \mathbb{E}_{\left(s_{X}^{0},\dots ,s_{X}^{t-1}\right)\in S_{X}^{t}}\left[ \sigma_{X}\left[\left(s_{X}^{0},s_{Y}^{0}\right) ,\dots ,\left(s_{X}^{t-1},s_{Y}^{t-1}\right)\right]\left(s_{X}^{t}\right)\right] .
\end{align*}
We simply do not track dependence on $X$'s actions, which we can do because $X$'s mixed actions determine her own history and we are taking expectations in the evaluation of autocratic strategies. We note that this strategy, while defined using exogenous sequences of opponent (pure) actions, can be used against any opponent, and it remains autocratic by Lemma~\ref{lem:exogenous}, as desired.
\end{proof}

\subsection{Enforcing payoff constraints with short memory}
The key to constructing simple autocratic strategies is identifying when two mixed actions can serve as the building blocks for enforcement. Lemma~\ref{lem:inequalities} ultimately provides necessary and sufficient conditions: two mixed actions $\tau_{X}^{+}$ and $\tau_{X}^{-}$ can enforce a constraint if their maximum and minimum payoffs satisfy certain inequalities that balance the discounting and allow for stable enforcement across all opponent responses.

Intuitively, these inequalities ensure that player $X$ can always find an appropriate mixture of $\tau_{X}^{+}$ and $\tau_{X}^{-}$ in response to any opponent action $s_{Y}$ such that the expected payoff remains on target. The ``$+$'' and ``$-$'' superscripts suggest their roles: $\tau_{X}^{+}$ typically yields higher values of $\varphi$ while $\tau_{X}^{-}$ yields lower values, and the strategy adjusts the mixture to maintain the target.

\begin{lemma}\label{lem:inequalities}
Suppose that $\lambda\in\left(0,1\right)$ and that $\tau_{X}^{\pm}\in\Delta\left(S_{X}\right)$ satisfy the inequalities
\begin{subequations}\label{eq:lemma_inequality}
\begin{align}
\min_{s_{Y}\in S_{Y}}\varphi\left(\tau_{X}^{+},s_{Y}\right) &\geqslant \left(1-\lambda\right)\max_{s_{Y}\in S_{Y}}\varphi\left(\tau_{X}^{+},s_{Y}\right) +\lambda\max_{s_{Y}\in S_{Y}}\varphi\left(\tau_{X}^{-},s_{Y}\right) ; \label{eq:lemma_inequality1} \\
\max_{s_{Y}\in S_{Y}}\varphi\left(\tau_{X}^{-},s_{Y}\right) &\leqslant \lambda\min_{s_{Y}\in S_{Y}}\varphi\left(\tau_{X}^{+},s_{Y}\right) +\left(1-\lambda\right)\min_{s_{Y}\in S_{Y}}\varphi\left(\tau_{X}^{-},s_{Y}\right) . \label{eq:lemma_inequality2}
\end{align}
\end{subequations}
Then, there exist $p_{0}\in\left[0,1\right]$ and a function $p^{\ast}:\left[0,1\right]\times S_{Y}\rightarrow\left[0,1\right]$ such that, with the reactive learning strategy
\begin{align*}
\sigma_{X}^{0} &= p_{0}\tau_{X}^{+}+\left(1-p_{0}\right)\tau_{X}^{-} ; \\
\sigma_{X}^{\ast}\left[p\tau_{X}^{+}+\left(1-p\right)\tau_{X}^{-},s_{Y}\right] &= p^{\ast}\left[p,s_{Y}\right]\tau_{X}^{+}+\left(1-p^{\ast}\left[p,s_{Y}\right]\right)\tau_{X}^{-} ,
\end{align*}
$X$ can enforce $\varphi\equiv K$ for any $K\in\left[\max_{s_{Y}\in S_{Y}}\varphi\left(\tau_{X}^{-},s_{Y}\right) ,\min_{s_{Y}\in S_{Y}}\varphi\left(\tau_{X}^{+},s_{Y}\right)\right]$.
\end{lemma}

The proof strategy is as follows. By Proposition \ref{prop:gnrcc_iff}, satisfying the generalized next-round correction condition (\eq{mainIntegral}) with an appropriate enforcement potential $\Psi$ is sufficient to guarantee that a reactive learning strategy is autocratic. For a two-point strategy that mixes between $\tau_{X}^{+}$ and $\tau_{X}^{-}$, the map $\Psi$ is completely determined by just two values: $\psi\left(\tau_{X}^{+}\right)$ and $\psi\left(\tau_{X}^{-}\right)$. The strategy begins with an initial mixture $\sigma_{X}^{0}=p_{0}\tau_{X}^{+}+\left(1-p_{0}\right)\tau_{X}^{-}$ for some $p_{0}\in\left[0,1\right]$, and responds to each opponent action $s_{Y}\in S_{Y}$ by playing $\sigma_{X}^{\ast}\left[p\tau_{X}^{+} + \left(1-p\right)\tau_{X}^{-},s_{Y}\right] = p^{\ast}\left[p,s_{Y}\right]\tau_{X}^{+} + \left(1-p^{\ast}\left[p,s_{Y}\right]\right)\tau_{X}^{-}$. The challenge is to choose $\psi\left(\tau_{X}^{+}\right)$, $\psi\left(\tau_{X}^{-}\right)$, and $p_{0}$ such that \eq{mainIntegral} holds and all transition probabilities $p^{\ast}\left[p,s_{Y}\right]$ remain in $\left[0,1\right]$ for every $p \in \left[0,1\right]$ and $s_{Y} \in S_{Y}$. We construct these values explicitly below.

\begin{proof}
For two-point strategies, finding an enforcement potential $\psi$ satisfying the generalized next-round correction condition amounts to finding two constants, $\psi\left(\tau_{X}^{+}\right)$ and $\psi\left(\tau_{X}^{-}\right)$. Let
\begin{subequations}\label{eq:psi_twopoint}
\begin{align}
\psi\left(\tau_{X}^{+}\right) &\coloneqq \frac{1}{1-\lambda}\min_{s_{Y}\in S_{Y}}\varphi\left(\tau_{X}^{+},s_{Y}\right) ; \\
\psi\left(\tau_{X}^{-}\right) &\coloneqq \frac{1}{1-\lambda}\max_{s_{Y}\in S_{Y}}\varphi\left(\tau_{X}^{-},s_{Y}\right) .
\end{align}
\end{subequations}
From \eq{lemma_inequality}, we must have $\psi\left(\tau_{X}^{+}\right)\geqslant\psi\left(\tau_{X}^{-}\right)$, which may be seen by writing the left-hand side of \eq{lemma_inequality1} as $\lambda\min_{s_{Y}\in S_{Y}}\varphi\left(\tau_{X}^{+},s_{Y}\right)+(1-\lambda)\min_{s_{Y}\in S_{Y}}\varphi\left(\tau_{X}^{+},s_{Y}\right)$. If $\psi\left(\tau_{X}^{+}\right) =\psi\left(\tau_{X}^{-}\right)$, then \eq{lemma_inequality} implies that $\varphi\left(\tau_{X}^{+},s_{Y}\right) =\varphi\left(\tau_{X}^{-},s_{Y}\right) =\left(1-\lambda\right)\psi\left(\tau_{X}^{+}\right)$ for every $s_{Y}\in S_{Y}$, in which case $\varphi\equiv\left(1-\lambda\right)\psi\left(\tau_{X}^{+}\right) =\left(1-\lambda\right)\psi\left(\tau_{X}^{-}\right)$ can be enforced by an unconditional strategy. Therefore, we may assume that $\psi\left(\tau_{X}^{+}\right) -\psi\left(\tau_{X}^{-}\right) >0$.

Fix $K\in\left[\max_{s_{Y}\in S_{Y}}\varphi\left(\tau_{X}^{-},s_{Y}\right) ,\min_{s_{Y}\in S_{Y}}\varphi\left(\tau_{X}^{+},s_{Y}\right)\right]$ and consider the response function
\begin{align}\label{eq:pfunction}
p^{\ast}\left[p,s_{Y}\right] = \frac{K-p\varphi\left(\tau_{X}^{+},s_{Y}\right) -\left(1-p\right)\varphi\left(\tau_{X}^{-},s_{Y}\right)}{\lambda\left(\psi\left(\tau_{X}^{+}\right) -\psi\left(\tau_{X}^{-}\right)\right)} + \frac{p-\left(1-\lambda\right) p_{0}}{\lambda} .
\end{align}
To have $p^{\ast}\left[p,s_{Y}\right]\in\left[0,1\right]$ for all $p\in\left[0,1\right]$ and $s_{Y}\in S_{Y}$, necessary and sufficient conditions are
\begin{align*}
\frac{K -\min_{s_{Y}\in S_{Y}}\varphi\left(\tau_{X}^{-},s_{Y}\right)}{\left(1-\lambda\right)\left(\psi\left(\tau_{X}^{+}\right) -\psi\left(\tau_{X}^{-}\right)\right)} -\frac{\lambda}{1-\lambda} &\leqslant p_{0} \leqslant \frac{K -\max_{s_{Y}\in S_{Y}}\varphi\left(\tau_{X}^{-},s_{Y}\right)}{\left(1-\lambda\right)\left(\psi\left(\tau_{X}^{+}\right) -\psi\left(\tau_{X}^{-}\right)\right)} ; \\
1-\frac{\min_{s_{Y}\in S_{Y}}\varphi\left(\tau_{X}^{+},s_{Y}\right) -K}{\left(1-\lambda\right)\left(\psi\left(\tau_{X}^{+}\right) -\psi\left(\tau_{X}^{-}\right)\right)} &\leqslant p_{0} \leqslant \frac{1}{1-\lambda}-\frac{\max_{s_{Y}\in S_{Y}}\varphi\left(\tau_{X}^{+},s_{Y}\right) -K}{\left(1-\lambda\right)\left(\psi\left(\tau_{X}^{+}\right) -\psi\left(\tau_{X}^{-}\right)\right)} .
\end{align*}
The fact that each of these two intervals is non-empty follows from \eq{lemma_inequality}. Moreover, by the values given in \eq{psi_twopoint}, the two intervals intersect at a unique initial probability, which is
\begin{align}\label{eq:p_0discounted}
p_{0} &= \frac{K -\max_{s_{Y}\in S_{Y}}\varphi\left(\tau_{X}^{-},s_{Y}\right)}{\left(1-\lambda\right)\left(\psi\left(\tau_{X}^{+}\right) -\psi\left(\tau_{X}^{-}\right)\right)} = \frac{K -\max_{s_{Y}\in S_{Y}}\varphi\left(\tau_{X}^{-},s_{Y}\right)}{\min_{s_{Y}\in S_{Y}}\varphi\left(\tau_{X}^{+},s_{Y}\right) -\max_{s_{Y}\in S_{Y}}\varphi\left(\tau_{X}^{-},s_{Y}\right)} .
\end{align}
With $p_{0}$ and $p^{\ast}$ well-defined and taking values in $\left[0,1\right]$, the generalized next-round correction condition is satisfied for $\varphi -K$, and thus $\left(p_{0},p^{\ast}\right)$ allows $X$ to enforce $\varphi\equiv K$.
\end{proof}

\begin{remark}
In the statement of Lemma~\ref{lem:inequalities}, we implicitly assume that the mixing probability, $p$, can be determined by the value of $p\tau_{X}^{+}+\left(1-p\right)\tau_{X}^{-}$. If $p,q\in\left[0,1\right]$ are mixing probabilities satisfying $p\tau_{X}^{+}+\left(1-p\right)\tau_{X}^{-}=q\tau_{X}^{+}+\left(1-q\right)\tau_{X}^{-}$, then $\left(p-q\right)\tau_{X}^{+}=\left(p-q\right)\tau_{X}^{-}$. If $\tau_{X}^{+}\neq\tau_{X}^{-}$, then $p=q$. If $\tau_{X}^{+}=\tau_{X}^{-}$, then the lemma holds trivially.
\end{remark}

We are especially interested in enforcing relationships of the form $\varphi\equiv 0$, and there is no loss of generality in setting $K=0$ since we can absorb this constant into $\varphi$, if necessary.

The following result is an immediate consequence of Lemma~\ref{lem:inequalities}:
\begin{corollary}\label{cor:corollary_zero}
If $\max_{s_{Y}\in S_{Y}}\varphi\left(\tau_{X}^{-},s_{Y}\right)\leqslant 0\leqslant\min_{s_{Y}\in S_{Y}}\varphi\left(\tau_{X}^{+},s_{Y}\right)$ and the hypotheses of Lemma~\ref{lem:inequalities} hold for $\tau_{X}^{+}$ and $\tau_{X}^{-}$, then $X$ can enforce $\varphi\equiv 0$ (using a two-point reactive learning strategy).
\end{corollary}
Motivated by this result and our focus on enforcing $\varphi\equiv 0$, we define the sets
\begin{subequations}\label{eq:SeparationSets}
\begin{align}
\Phi_{X}^{+} &\coloneqq \left\{\tau_{X}\in\Delta\left(S_{X}\right)\mid\min_{s_{Y}\in S_{Y}}\varphi\left(\tau_{X},s_{Y}\right)\geqslant 0\right\} ; \\
\Phi_{X}^{-} &\coloneqq \left\{\tau_{X}\in\Delta\left(S_{X}\right)\mid\max_{s_{Y}\in S_{Y}}\varphi\left(\tau_{X},s_{Y}\right)\leqslant 0\right\} .
\end{align}
\end{subequations}
In general, either or both of these sets can be empty.

We are now in a position to state and prove our main theoretical result:

\begin{theorem}\label{th:mainresultDiscounted}
Suppose that $\sigma_{X}:\mathcal{H}\rightarrow\Delta\left(S_{X}\right)$ is a $\left(\varphi ,\lambda\right)$-autocratic strategy of arbitrary memory. Then, there exists a two-point reactive learning strategy that is also $\left(\varphi ,\lambda\right)$-autocratic.
\end{theorem}

\begin{proof}
If $\lambda =0$, then conditioning is irrelevant and only the initial mixed action matters, so the result is trivial. Therefore, we assume that $\lambda >0$ going forward. Our approach is based on Lemma~\ref{lem:inequalities}. From the hypothesis that there exists a $\left(\varphi ,\lambda\right)$-autocratic strategy, we will construct mixed actions $\tau_{X}^{\pm}\in\Delta\left(S_{X}\right)$ satisfying \eq{lemma_inequality}.

By Proposition~\ref{prop:reactive_behavioral}, we may assume that $\sigma_{X}$ is a map from $\mathcal{H}_{Y}=\bigcup_{t\geqslant 0}S_{Y}^{t}$ to $\Delta\left(S_{X}\right)$, i.e., an opponent-conditioned behavioral strategy. Consider the map defined by
\begin{align*}
\Theta &: \mathcal{H}_{Y} \longrightarrow \mathbb{R} \nonumber \\
&: \left(s_{Y}^{0},\dots ,s_{Y}^{T-1}\right) \longmapsto -\lambda^{-T}\sum_{t=0}^{T-1}\lambda^{t}\varphi\left(\sigma_{X}\left[s_{Y}^{0},\dots ,s_{Y}^{t-1}\right] ,s_{Y}^{t}\right) ,
\end{align*}
where $\Theta\left(\varnothing\right)\coloneqq 0$. From the definition of $\Theta$, we see that for all $h\in\mathcal{H}_{Y}$ and $s_{Y}\in S_{Y}$,
\begin{align*}
\varphi\left(\sigma_{X}\left[h\right] ,s_{Y}\right) &= \Theta\left(h\right) -\lambda \Theta\left(h,s_{Y}\right) .
\end{align*}
From this equation, we also see that for every $h\in\mathcal{H}_{Y}$,
\begin{align}
\max_{s_{Y}\in S_{Y}} &\varphi\left(\sigma_{X}\left[h\right] ,s_{Y}\right) +\lambda\inf_{h'\in\mathcal{H}_{Y}\setminus\left\{\varnothing\right\}}\Theta\left(h'\right) \nonumber \\
&\leqslant \Theta\left(h\right)\leqslant \min_{s_{Y}\in S_{Y}}\varphi\left(\sigma_{X}\left[h\right] ,s_{Y}\right) +\lambda\sup_{h'\in\mathcal{H}_{Y}\setminus\left\{\varnothing\right\}}\Theta\left(h'\right) . \label{eq:theta_inequality}
\end{align}
To see that these bounds are finite, we note that for any $h\in\mathcal{H}_{Y}$ and $s_{Y}\in S_{Y}$,
\begin{align*}
\Theta\left(h\right) &= \sum_{t=T}^{\infty}\lambda^{t-T}\varphi\left(\sigma_{X}\left[h,\underbrace{s_{Y},\dots ,s_{Y}}_{\textrm{$t-T$ times}}\right] ,s_{Y}\right) ,
\end{align*}
since $\sigma_{X}$ is $\left(\varphi ,\lambda\right)$-autocratic. This equation gives $\left\|\Theta\right\|_{\infty}\leqslant\frac{1}{1-\lambda}\left\|\varphi\right\|_{\infty}<\infty$, since $\varphi$ is bounded.

Fix two sequences of histories, $\left\{h_{n}^{+}\right\}_{n=0}^{\infty}$ and $\left\{h_{n}^{-}\right\}_{n=0}^{\infty}$, such that $\left\{\Theta\left(h_{n}^{+}\right)\right\}_{n=0}^{\infty}$ converges monotonically to $\sup_{h\in\mathcal{H}_{Y}\setminus\left\{\varnothing\right\}}\Theta\left(h\right)$ and $\left\{\Theta\left(h_{n}^{-}\right)\right\}_{n=0}^{\infty}$ converges monotonically to $\inf_{h\in\mathcal{H}_{Y}\setminus\left\{\varnothing\right\}}\Theta\left(h\right)$. For every fixed $\varepsilon >0$, there must then exist $N_{\varepsilon}\geqslant 0$ such that, whenever $n\geqslant N_{\varepsilon}$,
\begin{subequations}
\begin{align}
\Theta\left(h_{n}^{+}\right) &> \sup_{h\in\mathcal{H}_{Y}\setminus\left\{\varnothing\right\}}\Theta\left(h\right) -\frac{1}{2\lambda}\varepsilon ; \label{eq:extrema_epsilon_a} \\
\Theta\left(h_{n}^{-}\right) &< \inf_{h\in\mathcal{H}_{Y}\setminus\left\{\varnothing\right\}}\Theta\left(h\right) +\frac{1}{2\lambda}\varepsilon . \label{eq:extrema_epsilon_b}
\end{align}\label{eq:extrema_epsilon}
\end{subequations}
We note that if $\sup_{h\in\mathcal{H}_{Y}\setminus\left\{\varnothing\right\}}\Theta\left(h\right) <0$, then \eq{extrema_epsilon_a} holds for $h_{n}^{+}=\varnothing$ for all $n\geqslant 0$. Similarly, if $\inf_{h\in\mathcal{H}_{Y}\setminus\left\{\varnothing\right\}}\Theta\left(h\right) >0$, then \eq{extrema_epsilon_b} holds when $h_{n}^{-}=\varnothing$ for all $n\geqslant 0$. Consider now the pair,
\begin{align*}
\left(\tau_{X,n}^{+},\tau_{X,n}^{-}\right) &\coloneqq 
\begin{cases}
\left(\sigma_{X}\left[\varnothing\right] ,\sigma_{X}\left[h_{n}^{-}\right]\right) & \sup_{h\in\mathcal{H}_{Y}\setminus\left\{\varnothing\right\}}\Theta\left(h\right) < 0,\\[1em]
\left(\sigma_{X}\left[h_{n}^{+}\right] ,\sigma_{X}\left[\varnothing\right]\right) & \inf_{h\in\mathcal{H}_{Y}\setminus\left\{\varnothing\right\}}\Theta\left(h\right) > 0,\\[1em]
\left(\sigma_{X}\left[h_{n}^{+}\right],\sigma_{X}\left[h_{n}^{-}\right]\right) & \textrm{otherwise}.
\end{cases}
\end{align*}
Since $\Delta\left(S_{X}\right)$ is compact, by passing to subsequences of histories if necessary, we may assume that $\left\{\tau_{X,n}^{+}\right\}_{n=0}^{\infty}$ and $\left\{\tau_{X,n}^{-}\right\}_{n=0}^{\infty}$ are convergent sequences in $\Delta\left(S_{X}\right)$, with limits $\tau_{X}^{+}$ and $\tau_{X}^{-}$, respectively. To complete the proof, using Lemma~\ref{lem:inequalities} (Corollary~\ref{cor:corollary_zero}), we show that \eq{lemma_inequality} holds and that $\tau_{X}^{+}\in\Phi_{X}^{+}$ and $\tau_{X}^{-}\in\Phi_{X}^{-}$. To see that \eq{lemma_inequality} holds, we fix $\varepsilon >0$ and let $n\geqslant N_{\varepsilon}$ and note that
\begin{align}
\left(1-\lambda\right) &\left(\sup_{h\in\mathcal{H}_{Y}\setminus\left\{\varnothing\right\}}\Theta\left(h\right) - \inf_{h\in\mathcal{H}_{Y}\setminus\left\{\varnothing\right\}}\Theta\left(h\right)\right) \nonumber \\
&< \Theta\left(h_{n}^{+}\right) -\Theta\left(h_{n}^{-}\right) + \frac{1}{\lambda}\varepsilon - \lambda\left(\sup_{h\in\mathcal{H}_{Y}\setminus\left\{\varnothing\right\}}\Theta\left(h\right) - \inf_{h\in\mathcal{H}_{Y}\setminus\left\{\varnothing\right\}}\Theta\left(h\right)\right) \quad \left(\eq{extrema_epsilon}\right) \nonumber \\
&\leqslant \min_{s_{Y}\in S_{Y}}\varphi\left(\tau_{X,n}^{+},s_{Y}\right) -\max_{s_{Y}\in S_{Y}}\varphi\left(\tau_{X,n}^{-},s_{Y}\right) + \frac{1}{\lambda}\varepsilon . \quad \left(\eq{theta_inequality}\right) \label{eq:sup-inf_inequality}
\end{align}
From this inequality, we can conclude that
\begin{align*}
\left(1-\lambda\right)\max_{s_{Y}\in S_{Y}} &\varphi\left(\tau_{X,n}^{+},s_{Y}\right) +\lambda\max_{s_{Y}\in S_{Y}} \varphi\left(\tau_{X,n}^{-},s_{Y}\right) \nonumber \\
&\leqslant \left(1-\lambda\right)\min_{s_{Y}\in S_{Y}}\varphi\left(\tau_{X,n}^{+},s_{Y}\right) + \lambda\max_{s_{Y}\in S_{Y}} \varphi\left(\tau_{X,n}^{-},s_{Y}\right) \nonumber \\
&\quad +\left(1-\lambda\right)\lambda\left(\sup_{h\in\mathcal{H}_{Y}\setminus\left\{\varnothing\right\}}\Theta\left(h\right) - \inf_{h\in\mathcal{H}_{Y}\setminus\left\{\varnothing\right\}}\Theta\left(h\right)\right) \nonumber \quad \left(\eq{theta_inequality}\right) \\
&< \left(1-\lambda\right)\min_{s_{Y}\in S_{Y}}\varphi\left(\tau_{X,n}^{+},s_{Y}\right) + \lambda\max_{s_{Y}\in S_{Y}} \varphi\left(\tau_{X,n}^{-},s_{Y}\right) \nonumber \\
&\quad +\lambda\left(\min_{s_{Y}\in S_{Y}}\varphi\left(\tau_{X,n}^{+},s_{Y}\right) -\max_{s_{Y}\in S_{Y}}\varphi\left(\tau_{X,n}^{-},s_{Y}\right) + \frac{1}{\lambda}\varepsilon\right) \quad \left(\eq{sup-inf_inequality}\right) \nonumber \\
&= \min_{s_{Y}\in S_{Y}}\varphi\left(\tau_{X,n}^{+},s_{Y}\right) +\varepsilon ; \\
\lambda\min_{s_{Y}\in S_{Y}} &\varphi\left(\tau_{X,n}^{+},s_{Y}\right) +\left(1-\lambda\right)\min_{s_{Y}\in S_{Y}} \varphi\left(\tau_{X,n}^{-},s_{Y}\right) \nonumber \\
&\geqslant \lambda\min_{s_{Y}\in S_{Y}}\varphi\left(\tau_{X,n}^{+},s_{Y}\right) + \left(1-\lambda\right)\max_{s_{Y}\in S_{Y}} \varphi\left(\tau_{X,n}^{-},s_{Y}\right) \nonumber \\
&\quad -\left(1-\lambda\right)\lambda\left(\sup_{h\in\mathcal{H}_{Y}\setminus\left\{\varnothing\right\}}\Theta\left(h\right) - \inf_{h\in\mathcal{H}_{Y}\setminus\left\{\varnothing\right\}}\Theta\left(h\right)\right) \quad \left(\eq{theta_inequality}\right) \nonumber \\
&> \lambda\min_{s_{Y}\in S_{Y}}\varphi\left(\tau_{X,n}^{+},s_{Y}\right) + \left(1-\lambda\right)\max_{s_{Y}\in S_{Y}} \varphi\left(\tau_{X,n}^{-},s_{Y}\right) \nonumber \\
&\quad -\lambda\left(\min_{s_{Y}\in S_{Y}}\varphi\left(\tau_{X,n}^{+},s_{Y}\right) -\max_{s_{Y}\in S_{Y}}\varphi\left(\tau_{X,n}^{-},s_{Y}\right) + \frac{1}{\lambda}\varepsilon\right) \quad \left(\eq{sup-inf_inequality}\right) \nonumber \\
&= \max_{s_{Y}\in S_{Y}} \varphi\left(\tau_{X,n}^{-},s_{Y}\right) - \varepsilon .
\end{align*}
Since $\varepsilon >0$ was arbitrary, it follows that \eq{lemma_inequality} holds for the pair $\left(\tau_{X}^{+},\tau_{X}^{-}\right)$.

What remains to be shown is that $\tau_{X}^{+}\in\Phi_{X}^{+}$ and $\tau_{X}^{-}\in\Phi_{X}^{-}$, which we establish in three cases:
\begin{enumerate}

\item[\emph{(i)}] If $\sup_{h\in\mathcal{H}_{Y}\setminus\left\{\varnothing\right\}}\Theta\left(h\right) <0$, then, by \eq{theta_inequality},
\begin{align*}
\min_{s_{Y}\in S_{Y}}\varphi\left(\sigma_{X}\left[\varnothing\right] ,s_{Y}\right)\geqslant -\lambda\sup_{h\in\mathcal{H}_{Y}\setminus\left\{\varnothing\right\}}\Theta\left(h\right) > 0 ,
\end{align*}
and thus $\min_{s_{Y}\in S_{Y}}\varphi\left(\tau_{X}^{+},s_{Y}\right) >0$. For $n\geqslant N_{\varepsilon}$, \eq{theta_inequality} and \eq{extrema_epsilon} give
\begin{align}
\max_{s_{Y}\in S_{Y}}\varphi\left(\tau_{X,n}^{-},s_{Y}\right)
&< \left(1-\lambda\right)\inf_{h\in\mathcal{H}_{Y}\setminus\left\{\varnothing\right\}}\Theta\left(h\right) +\frac{1}{2\lambda}\varepsilon , \label{eq:case_i}
\end{align}
and thus $\max_{s_{Y}\in S_{Y}}\varphi\left(\tau_{X}^{-},s_{Y}\right) < 0$ in the limit, since $\varepsilon >0$ was arbitrary.

\item[\emph{(ii)}] If $\inf_{h\in\mathcal{H}_{Y}\setminus\left\{\varnothing\right\}}\Theta\left(h\right) >0$, then, by \eq{theta_inequality},
\begin{align*}
\max_{s_{Y}\in S_{Y}}\varphi\left(\sigma_{X}\left[\varnothing\right] ,s_{Y}\right)\leqslant -\lambda\inf_{h\in\mathcal{H}_{Y}\setminus\left\{\varnothing\right\}}\Theta\left(h\right) < 0 ,
\end{align*}
which gives $\max_{s_{Y}\in S_{Y}}\varphi\left(\tau_{X}^{-},s_{Y}\right) <0$. For $n\geqslant N_{\varepsilon}$, \eq{theta_inequality} and \eq{extrema_epsilon} give
\begin{align}
\min_{s_{Y}\in S_{Y}}\varphi\left(\tau_{X,n}^{+},s_{Y}\right)
&> \left(1-\lambda\right)\sup_{h\in\mathcal{H}_{Y}\setminus\left\{\varnothing\right\}}\Theta\left(h\right) -\frac{1}{2\lambda}\varepsilon , \label{eq:case_ii}
\end{align}
so $\min_{s_{Y}\in S_{Y}}\varphi\left(\tau_{X}^{+},s_{Y}\right) > 0$ in the limit, since $\varepsilon >0$ was arbitrary.

\item[\emph{(iii)}] If neither \emph{(i)} nor \emph{(ii)} applies, then we deduce that $\tau_{X}^{+}\in\Phi_{X}^{+}$ and $\tau_{X}^{-}\in\Phi_{X}^{-}$ by taking the limits of \eq{case_i} and \eq{case_ii} as $n\rightarrow\infty$ and noting that $\varepsilon >0$ was arbitrary.

\end{enumerate}

By Corollary~\ref{cor:corollary_zero}, we obtain a two-point reactive learning strategy that is $\left(\varphi ,\lambda\right)$-autocratic.
\end{proof}

Theorem~\ref{th:mainresultDiscounted} shows that if $\sigma_{X}$ is $\left(\varphi ,\lambda\right)$-autocratic, then there exist $\tau_{X}^{\pm}\in\Delta\left(S_{X}\right)$ satisfying \eq{lemma_inequality} and $\max_{s_{Y}\in S_{Y}}\varphi\left(\tau_{X}^{-},s_{Y}\right)\leqslant 0\leqslant\min_{s_{Y}\in S_{Y}}\varphi\left(\tau_{X}^{+},s_{Y}\right)$. If the latter condition is an equality, then \eq{lemma_inequality} implies that $\max_{s_{Y}\in S_{Y}}\varphi\left(\tau_{X}^{+},s_{Y}\right)\leqslant\min_{s_{Y}\in S_{Y}}\varphi\left(\tau_{X}^{-},s_{Y}\right)$, and thus $\varphi\left(\tau_{X}^{+},s_{Y}\right) =\varphi\left(\tau_{X}^{-},s_{Y}\right) =0$ for all $s_{Y}\in S_{Y}$. Therefore, in this case, \eq{lemma_inequality} holds with $\lambda =0$, and the simple mixed action $\tau_{X}^{+}$, played in every round, is $\left(\varphi ,\lambda\right)$-autocratic for all $\lambda\geqslant 0$. (This statement is also true for $\tau_{X}^{-}$, if it is distinct from $\tau_{X}^{+}$.) This case captures the unconditional autocratic strategies.

On the other hand, if $\max_{s_{Y}\in S_{Y}}\varphi\left(\tau_{X}^{-},s_{Y}\right) <\min_{s_{Y}\in S_{Y}}\varphi\left(\tau_{X}^{+},s_{Y}\right)$, then \eq{lemma_inequality} gives
\begin{align*}
\lambda &\geqslant \max\left\{\frac{\max_{s_{Y}\in S_{Y}}\varphi\left(\tau_{X}^{+},s_{Y}\right) -\min_{s_{Y}\in S_{Y}}\varphi\left(\tau_{X}^{+},s_{Y}\right)}{\max_{s_{Y}\in S_{Y}}\varphi\left(\tau_{X}^{+},s_{Y}\right) -\max_{s_{Y}\in S_{Y}}\varphi\left(\tau_{X}^{-},s_{Y}\right)},\frac{\max_{s_{Y}\in S_{Y}}\varphi\left(\tau_{X}^{-},s_{Y}\right) -\min_{s_{Y}\in S_{Y}}\varphi\left(\tau_{X}^{-},s_{Y}\right)}{\min_{s_{Y}\in S_{Y}}\varphi\left(\tau_{X}^{+},s_{Y}\right) -\min_{s_{Y}\in S_{Y}}\varphi\left(\tau_{X}^{-},s_{Y}\right)}\right\} \nonumber \\
&= 1-\frac{\min_{s_{Y}\in S_{Y}}\varphi\left(\tau_{X}^{+},s_{Y}\right) -\max_{s_{Y}\in S_{Y}}\varphi\left(\tau_{X}^{-},s_{Y}\right)}{\max\left\{ \substack{\max_{s_{Y}\in S_{Y}}\varphi\left(\tau_{X}^{+},s_{Y}\right) -\max_{s_{Y}\in S_{Y}}\varphi\left(\tau_{X}^{-},s_{Y}\right) , \\ \min_{s_{Y}\in S_{Y}}\varphi\left(\tau_{X}^{+},s_{Y}\right) -\min_{s_{Y}\in S_{Y}}\varphi\left(\tau_{X}^{-},s_{Y}\right)} \right\}} .
\end{align*}
By Lemma~\ref{lem:inequalities} and Theorem~\ref{th:mainresultDiscounted}, we immediately have the following result:

\begin{proposition}\label{prop:lambda_min}
There exists a $\left(\varphi ,\lambda\right)$-autocratic strategy if and only if $\Phi_{X}^{+},\Phi_{X}^{-}\neq\left\{\right\}$ and either \emph{(i)} $\Phi_{X}^{+}\cap\Phi_{X}^{-}\neq\left\{\right\}$ or \emph{(ii)} $\Phi_{X}^{+}\cap\Phi_{X}^{-}=\left\{\right\}$ and $\lambda\geqslant\lambda_{\min}$, where
\begin{align}\label{eq:lambda_min}
\lambda_{\min} &\coloneqq 1-\sup_{\left(\tau_{X}^{+},\tau_{X}^{-}\right)\in\Phi_{X}^{+}\times\Phi_{X}^{-}}\frac{\min_{s_{Y}\in S_{Y}}\varphi\left(\tau_{X}^{+},s_{Y}\right) -\max_{s_{Y}\in S_{Y}}\varphi\left(\tau_{X}^{-},s_{Y}\right)}{\max\left\{ \substack{\max_{s_{Y}\in S_{Y}}\varphi\left(\tau_{X}^{+},s_{Y}\right) -\max_{s_{Y}\in S_{Y}}\varphi\left(\tau_{X}^{-},s_{Y}\right) , \\ \min_{s_{Y}\in S_{Y}}\varphi\left(\tau_{X}^{+},s_{Y}\right) -\min_{s_{Y}\in S_{Y}}\varphi\left(\tau_{X}^{-},s_{Y}\right)} \right\}} .
\end{align}
(Note that we use the notation $\left\{\right\}$ to denote the empty set to avoid confusion with the null history, $\varnothing$.)
\end{proposition}

\begin{remark}
The minimum discount factor $\lambda_{\min}$ has a natural interpretation: when $\lambda <\lambda_{\min}$, the game is too short on average (with expected duration $1/\left(1-\lambda\right)$ rounds) for the player to enforce the constraint. Enforcement requires the player to ``correct'' deviations from the target payoff relationship over time. If the game terminates too quickly, there is insufficient opportunity for these corrections to bring the expected payoff to the target. Note that $\lambda_{\min}$ depends on the ``distance'' between the extreme values of $\varphi$ at $\tau_{X}^{+}$ and $\tau_{X}^{-}$: larger swings between $\varphi\left(\tau_{X}^{+},\cdot\right)$ and $\varphi\left(\tau_{X}^{-},\cdot\right)$ require more patience (larger $\lambda$) to balance out through the discounting mechanism.
\end{remark}

\begin{remark}\label{rem:remark3}
With some slight modifications in its proof, Theorem~\ref{th:mainresultDiscounted} can be extended to the setting where $S_{X}$ and $S_{Y}$ are compact subsets of $\mathbb{R}^{n}$ and $\mathbb{R}^{m}$, respectively, and $\varphi$ is continuous, by invoking the Banach-Alaoglu theorem. More generally, with minimal changes, it can be shown that for $S_{X}$ and $S_{Y}$ of any structure, and for any bounded function $\varphi$, reduction to short-memory autocratic strategies is still possible under arbitrarily small extensions of the game length. Rigorously speaking, for any $\varepsilon >0$ such that $\lambda +\varepsilon<1$, every $\left(\varphi ,\lambda\right)$-autocratic strategy can be replaced by a two-point reactive learning $\left(\varphi ,\lambda +\varepsilon\right)$-autocratic strategy.
\end{remark}

\subsection{Additive functions and reactive strategies}

\begin{definition}
A function $\varphi :S_{X}\times S_{Y}\rightarrow\mathbb{R}$ is additive if there exist functions $\phi_{X}:S_{X}\rightarrow\mathbb{R}$ and $\phi_{Y}:S_{Y}\rightarrow\mathbb{R}$ with $\varphi\left(s_{X},s_{Y}\right) =\phi_{X}\left(s_{X}\right) +\phi_{Y}\left(s_{Y}\right)$ for all $s_{X}\in S_{X}$ and $s_{Y}\in S_{Y}$.
\end{definition}

For additive functions, we can strengthen our results considerably. Theorem \ref{th:mainresultDiscounted} guarantees that enforcement can be achieved with two-point reactive learning strategies. However, the additive structure, $\varphi\left(s_{X},s_{Y}\right) = \phi_{X}\left(s_{X}\right) + \phi_{Y}\left(s_{Y}\right)$, allows us to take $\psi =\phi_{X}$ in the next-round correction condition. Here, we show that two-point reactive learning strategies can be reduced to even simpler reactive strategies that condition solely on the opponent's last action, without tracking $X$'s own actions.

\begin{theorem}\label{thm:additive}
If $\varphi\equiv 0$ is $\lambda$-enforceable by $X$ and $\varphi$ is additive, then $\varphi\equiv 0$ is $\lambda$-enforceable by $X$ using a reactive strategy $\sigma_{X}:S_{Y}\rightarrow\Delta\left(S_{X}\right)$.
\end{theorem}

\begin{proof}
Suppose that the opponent plays some fixed $s_{Y}\in S_{Y}$ unconditionally in every round. Let $\left\{\sigma_{X}^{0}\right\}\cup\left\{\sigma_{X}^{t}\left[s_{Y}\right]\right\}_{t=1}^{\infty}\subseteq\Delta\left(S_{X}\right)$ be a sequence for which $\varphi\left(\sigma_{X}^{0},s_{Y}\right) +\sum_{t=1}^{\infty}\lambda^{t}\varphi\left(\sigma_{X}^{t}\left[s_{Y}\right] ,s_{Y}\right) =0$ (making the dependence on $s_{Y}$ explicit for all rounds $t\geqslant 1$). Such a sequence exists by the assumption that $\varphi\equiv 0$ is $\lambda$-enforceable by some behavioral strategy, and we simply let $\sigma_{X}^{0}$ be the initial mixed action and $\sigma_{X}^{t}\left[s_{Y}\right]$ be the marginal mixed action for $X$ at time $t$ after $Y$ used action $s_{Y}$ in all previous rounds. Consider the distribution,
\begin{align*}
\sigma_{X}\left[s_{Y}\right]\left(\cdot\right) &\coloneqq \left(1-\lambda\right)\sum_{t=1}^{\infty}\lambda^{t-1}\sigma_{X}^{t}\left[s_{Y}\right]\left(\cdot\right)
\end{align*}
defined on the space of $X$'s actions, $S_{X}$. The sequence of measures, $\left\{\widetilde{\sigma}_{X}^{k}\left[s_{Y}\right]\right\}_{k\geqslant 1}$, defined by
\begin{align*}
\widetilde{\sigma}_{X}^{k}\left[s_{Y}\right]\left(\cdot\right) &\coloneqq \left(1-\lambda\right)\sum_{t=1}^{k}\lambda^{t-1}\sigma_{X}^{t}\left[s_{Y}\right]\left(\cdot\right) ,
\end{align*}
then converges in total variation to $\sigma_{X}\left[s_{Y}\right]$ because
\begin{align*}
\sigma_{X}\left[s_{Y}\right]\left(\cdot\right) - \widetilde{\sigma}_{X}^{k}\left[s_{Y}\right]\left(\cdot\right) &= \left(1-\lambda\right)\sum_{t=k+1}^{\infty}\lambda^{t-1}\sigma_{X}^{t}\left[s_{Y}\right]\left(\cdot\right) \nonumber \\
&\leqslant \left(1-\lambda\right)\sum_{t=k+1}^{\infty}\lambda^{t-1} \nonumber \\
&= \lambda^{k} ,
\end{align*}
and thus $\left\|\sigma_{X}\left[s_{Y}\right] -\widetilde{\sigma}_{X}^{k}\left[s_{Y}\right]\right\|_{\textrm{TV}}\leqslant\lambda^{k}\rightarrow 0$ as $k\rightarrow\infty$. It follows, then, that for every $s_{Y}\in S_{Y}$,
\begin{align*}
0 &= \left(1-\lambda\right)\varphi\left(\sigma_{X}^{0},s_{Y}\right) + \left(1-\lambda\right)\sum_{t=1}^{\infty}\lambda^{t}\varphi\left(\sigma_{X}^{t}\left[s_{Y}\right] ,s_{Y}\right) \nonumber \\
&= \left(1-\lambda\right)\varphi\left(\sigma_{X}^{0},s_{Y}\right) + \lambda\lim_{t\rightarrow\infty}\varphi\left(\widetilde{\sigma}_{X}^{t}\left[s_{Y}\right] ,s_{Y}\right) \nonumber \\
&= \left(1-\lambda\right)\varphi\left(\sigma_{X}^{0},s_{Y}\right) + \lambda\varphi\left(\sigma_{X}\left[s_{Y}\right] ,s_{Y}\right) ,
\end{align*}
where the last equation is due to the boundedness of $\varphi$ and the fact that the sequence $\left\{\widetilde{\sigma}_{X}^{k}\left[s_{Y}\right]\right\}_{k\geqslant 1}$ converges to $\sigma_{X}\left[s_{Y}\right]$ in total variation. Since $\varphi\left(s_{X},s_{Y}\right) =\phi_{X}\left(s_{X}\right) +\phi_{Y}\left(s_{Y}\right)$, we have
\begin{align*}
\phi_{Y}\left(s_{Y}\right) &= -\lambda\mathbb{E}_{s_{X}\sim\sigma_{X}\left[s_{Y}\right]}\left[\phi_{X}\left(s_{X}\right)\right] - \left(1-\lambda\right)\mathbb{E}_{s_{X}\sim\sigma_{X}^{0}}\left[\phi_{X}\left(s_{X}\right)\right] ,
\end{align*}
which is the pointwise next-round correction condition for the reactive strategy $\left(\sigma_{X}^{0},\sigma_{X}\left[s_{Y}\right]\right)$, giving the desired result by Theorem~\ref{thm:mcavoy-hauert} (or Proposition~\ref{prop:gnrcc_iff} more generally).
\end{proof}

\begin{remark}\label{rem:remark4}
Although our focus is on finite action spaces, this theorem readily extends to measurable spaces. Only minor additional justification is needed, such as the fact that for each $s_{Y}\in S_{Y}$, $\sigma_{X}\left[s_{Y}\right]$ is a probability measure (which follows from the Vitali-Hahn-Saks Theorem \citep[see][]{doob:S:1994}).
\end{remark}

\begin{remark}
An important feature of the reactive strategy constructed in Theorem \ref{thm:additive} is that it uses the same initial mixed action, $\sigma_{X}^{0}$, as the original (possibly longer-memory) autocratic strategy. This is crucial for the proof, which relies on taking a weighted average of the mixed actions along the trajectory induced by the original strategy when the opponent plays the same action, $s_{Y}$, repeatedly. The initial action anchors this averaging process, ensuring that the mean converges to the desired enforcement property.
\end{remark}

In fact, we can say slightly more about reactive strategies in this context:
\begin{lemma}\label{lem:inequalities_additive}
Suppose that $\lambda\in\left[0,1\right)$ and that $\tau_{X}^{\pm}\in\Delta\left(S_{X}\right)$ satisfy \eq{lemma_inequality}. If $\varphi$ is additive, then there exist $p_{0}\in\left[0,1\right]$ and a function $p^{\ast}:S_{Y}\rightarrow\left[0,1\right]$ such that, with the two-point reactive strategy
\begin{align*}
\sigma_{X}^{0} &= p_{0}\tau_{X}^{+}+\left(1-p_{0}\right)\tau_{X}^{-} ; \\
\sigma_{X}^{\ast}\left[s_{Y}\right] &= p^{\ast}\left[s_{Y}\right]\tau_{X}^{+}+\left(1-p^{\ast}\left[s_{Y}\right]\right)\tau_{X}^{-} ,
\end{align*}
$X$ can enforce $\varphi\equiv K$ for any $K\in\left[\phi_{X}\left(\tau_{X}^{-}\right) +\max_{s_{Y}\in S_{Y}}\phi_{Y}\left(s_{Y}\right) ,\phi_{X}\left(\tau_{X}^{+}\right) +\min_{s_{Y}\in S_{Y}}\phi_{Y}\left(s_{Y}\right)\right]$.
\end{lemma}

\begin{proof}
We assume $\lambda >0$ since, for $\lambda =0$, the result is immediate. Suppose that $\varphi\left(s_{X},s_{Y}\right)=\phi_{X}\left(s_{X}\right) +\phi_{Y}\left(s_{Y}\right)$ for all $s_{X}\in S_{X}$ and $s_{Y}\in S_{Y}$, and let $\psi =\phi_{X}$. The statement of the lemma is vacuous unless the range of allowable $K$ values is non-empty; thus, if $\psi\left(\tau_{X}^{+}\right) =\psi\left(\tau_{X}^{-}\right)$, then $\phi_{Y}$ must be constant, and $X$ can enforce $\varphi\equiv\psi\left(\tau_{X}^{-}\right) +\max_{s_{Y}\in S_{Y}}\phi_{Y}\left(s_{Y}\right)$ simply by playing $\tau_{X}^{+}$ in every round. Therefore, we may assume that $\psi\left(\tau_{X}^{+}\right)\neq\psi\left(\tau_{X}^{-}\right)$. The generalized next-round correction condition is equivalent to
\begin{align*}
p^{\ast}\left[p,s_{Y}\right] &= \frac{K-p\varphi\left(\tau_{X}^{+},s_{Y}\right) -\left(1-p\right)\varphi\left(\tau_{X}^{-},s_{Y}\right)}{\lambda\left(\psi\left(\tau_{X}^{+}\right) -\psi\left(\tau_{X}^{-}\right)\right)} + \frac{p-\left(1-\lambda\right) p_{0}}{\lambda} \nonumber \\
&= \frac{K-p\left(\phi_{X}\left(\tau_{X}^{+}\right) -\phi_{X}\left(\tau_{X}^{-}\right)\right) -\phi_{X}\left(\tau_{X}^{-}\right) -\phi_{Y}\left(s_{Y}\right)}{\lambda\left(\psi\left(\tau_{X}^{+}\right) -\psi\left(\tau_{X}^{-}\right)\right)} + \frac{p-\left(1-\lambda\right) p_{0}}{\lambda} \nonumber \\
&= \frac{K -\phi_{X}\left(\tau_{X}^{-}\right) -\phi_{Y}\left(s_{Y}\right)}{\lambda\left(\phi_{X}\left(\tau_{X}^{+}\right) -\phi_{X}\left(\tau_{X}^{-}\right)\right)} - \frac{\left(1-\lambda\right) p_{0}}{\lambda} ,
\end{align*}
which is independent of $p$. To ensure $p^{\ast}\left[s_{Y}\right]\in\left[0,1\right]$ for all $s_{Y}\in S_{Y}$ (dropping $p$), we require
\begin{align*}
\frac{K -\phi_{X}\left(\tau_{X}^{-}\right) -\min_{s_{Y}\in S_{Y}}\phi_{Y}\left(s_{Y}\right)}{\left(1-\lambda\right)\left(\phi_{X}\left(\tau_{X}^{+}\right) -\phi_{X}\left(\tau_{X}^{-}\right)\right)} -\frac{\lambda}{1-\lambda} &\leqslant p_{0} \leqslant \frac{K -\phi_{X}\left(\tau_{X}^{-}\right) -\max_{s_{Y}\in S_{Y}}\phi_{Y}\left(s_{Y}\right)}{\left(1-\lambda\right)\left(\phi_{X}\left(\tau_{X}^{+}\right) -\phi_{X}\left(\tau_{X}^{-}\right)\right)} .
\end{align*}
Since $K\in\left[\phi_{X}\left(\tau_{X}^{-}\right) +\max_{s_{Y}\in S_{Y}}\phi_{Y}\left(s_{Y}\right) ,\phi_{X}\left(\tau_{X}^{+}\right) +\min_{s_{Y}\in S_{Y}}\phi_{Y}\left(s_{Y}\right)\right]$, we have the inequalities
\begin{align*}
\frac{K -\phi_{X}\left(\tau_{X}^{-}\right) -\min_{s_{Y}\in S_{Y}}\phi_{Y}\left(s_{Y}\right)}{\left(1-\lambda\right)\left(\phi_{X}\left(\tau_{X}^{+}\right) -\phi_{X}\left(\tau_{X}^{-}\right)\right)} -\frac{\lambda}{1-\lambda} &\leqslant 1 ; \\
\frac{K -\phi_{X}\left(\tau_{X}^{-}\right) -\max_{s_{Y}\in S_{Y}}\phi_{Y}\left(s_{Y}\right)}{\left(1-\lambda\right)\left(\phi_{X}\left(\tau_{X}^{+}\right) -\phi_{X}\left(\tau_{X}^{-}\right)\right)} &\geqslant 0 .
\end{align*}
By \eq{lemma_inequality}, we have $\max_{s_{Y}\in S_{Y}}\phi_{Y}\left(s_{Y}\right)-\min_{s_{Y}\in S_{Y}}\phi_{Y}\left(s_{Y}\right)\leqslant \lambda\left(\phi_{X}\left(\tau_{X}^{+}\right) -\phi_{X}\left(\tau_{X}^{-}\right)\right)$, which gives
\begin{align*}
\frac{K -\phi_{X}\left(\tau_{X}^{-}\right) -\min_{s_{Y}\in S_{Y}}\phi_{Y}\left(s_{Y}\right)}{\left(1-\lambda\right)\left(\phi_{X}\left(\tau_{X}^{+}\right) -\phi_{X}\left(\tau_{X}^{-}\right)\right)} -\frac{\lambda}{1-\lambda} &\leqslant \frac{K -\phi_{X}\left(\tau_{X}^{-}\right) -\max_{s_{Y}\in S_{Y}}\phi_{Y}\left(s_{Y}\right)}{\left(1-\lambda\right)\left(\phi_{X}\left(\tau_{X}^{+}\right) -\phi_{X}\left(\tau_{X}^{-}\right)\right)} .
\end{align*}
Therefore, the range of acceptable values of $p_{0}$ is the interval
\begin{align}\label{eq:p_0undiscounted}
\Bigg[ \max &\left\{\frac{K -\phi_{X}\left(\tau_{X}^{-}\right) -\min_{s_{Y}\in S_{Y}}\phi_{Y}\left(s_{Y}\right)}{\left(1-\lambda\right)\left(\phi_{X}\left(\tau_{X}^{+}\right) -\phi_{X}\left(\tau_{X}^{-}\right)\right)} -\frac{\lambda}{1-\lambda},0\right\} , \nonumber \\
&\min\left\{\frac{K -\phi_{X}\left(\tau_{X}^{-}\right) -\max_{s_{Y}\in S_{Y}}\phi_{Y}\left(s_{Y}\right)}{\left(1-\lambda\right)\left(\phi_{X}\left(\tau_{X}^{+}\right) -\phi_{X}\left(\tau_{X}^{-}\right)\right)},1\right\} \Bigg] .
\end{align}
Any $p_{0}$ in this range translates into a valid response function, $p^{\ast}$.
\end{proof}

\begin{corollary}\label{cor:reactive}
If $\varphi\equiv 0$ is $\lambda$-enforceable by $X$ and $\varphi$ is additive, then $\varphi\equiv 0$ is $\lambda$-enforceable by $X$ using a two-point reactive strategy.
\end{corollary}

We further obtain a simplification of Proposition~\ref{prop:lambda_min}:

\begin{proposition}
Suppose $\varphi$ is additive. Then, there exists a $\left(\varphi ,\lambda\right)$-autocratic strategy if and only if $\Phi_{X}^{+},\Phi_{X}^{-}\neq\left\{\right\}$ and either \emph{(i)} $\phi_{X}$ is constant or \emph{(ii)} $\phi_{X}$ is non-constant and $\lambda\geqslant\lambda_{\min}$, where
\begin{align*}
\lambda_{\min} &= \frac{\max_{s_{Y}\in S_{Y}}\phi_{Y}\left(s_{Y}\right) -\min_{s_{Y}\in S_{Y}}\phi_{Y}\left(s_{Y}\right)}{\max_{s_{X}\in S_{X}}\phi_{X}\left(s_{X}\right) -\min_{s_{X}\in S_{X}}\phi_{X}\left(s_{X}\right)} .
\end{align*}
\end{proposition}

\subsection{Infinite-horizon (undiscounted) games}\label{sec:undiscounted}
Up until this point, our focus has been on discounted games, which terminate in finitely many rounds with probability one ($\lambda <1$). Although this case is the most realistic from a modeling perspective, we know from classical results that some linear payoff relationships are enforceable only in undiscounted, infinite-horizon games. An example is tit-for-tat in the repeated prisoner's dilemma, which ensures that $\pi_{X}=\pi_{Y}$ as $\lambda\rightarrow 1$ (and only in this infinite-horizon limit). In this section, we consider this limit more generally.

We assume that the undiscounted expectation of $\varphi$ is Ces\`{a}ro summable, meaning the limit
\begin{align}\label{eq:cesaro}
\lim_{T\rightarrow\infty}\frac{1}{T+1}\sum_{t=0}^{T}\mathbb{E}_{\sigma_{X},\sigma_{Y}}\left[\varphi\left(s_{X}^{t},s_{Y}^{t}\right)\right]
\end{align}
exists. Since both action spaces are finite, this limit exists whenever both behavioral strategies $\sigma_{X}$ and $\sigma_{Y}$ are finite-memory. Definitions~\ref{def:autocratic} and~\ref{def:enforceable} are naturally extended to the infinite-horizon regime: We say that a strategy $\sigma_{X}:\mathcal{H}\rightarrow\Delta\left(S_{X}\right)$ is $\left(\varphi ,1\right)$-autocratic if \eq{cesaro} vanishes for all opponent strategies $\sigma_{Y}:\mathcal{H}\rightarrow\Delta\left(S_{Y}\right)$. Accordingly, we say $\varphi\equiv 0$ is $1$-enforceable if there exists a $\left(\varphi ,1\right)$-autocratic strategy. Moreover, when we say that $\varphi\equiv 0$ is enforceable, we now implicitly assume the extended discount range, allowing $\lambda\in\left[0,1\right]$. We say that $\lambda_{\min}=1$ if $\varphi\equiv 0$ is $1$-enforceable but not $\lambda$-enforceable for any $\lambda <1$.

We begin with the undiscounted analog of Theorem~\ref{th:mainresultDiscounted}:
\begin{theorem}\label{th:mainresultUndiscounted}
Suppose that $\sigma_{X}:\mathcal{H}\rightarrow\Delta\left(S_{X}\right)$ is a $\left(\varphi ,1\right)$-autocratic strategy of arbitrary memory. Then, there exists a two-point reactive learning strategy that is also $\left(\varphi ,1\right)$-autocratic.
\end{theorem}

The proof requires two technical lemmas that establish the existence of appropriate two-point support and the undiscounted analog of the generalized next-round correction condition.

\begin{lemma}\label{lem:undiscounted_support}
If $\varphi\equiv 0$ is enforceable with $\lambda_{\min}=1$, then there exist $\tau_{X}^{\pm}\in\Delta\left(S_{X}\right)$ such that
\begin{align}\label{eq:zeroinequalities}
\max_{s_{Y}\in S_{Y}}\varphi\left(\tau_{X}^{+},s_{Y}\right) >\min_{s_{Y}\in S_{Y}}\varphi\left(\tau_{X}^{+},s_{Y}\right) =0=\max_{s_{Y}\in S_{Y}}\varphi\left(\tau_{X}^{-},s_{Y}\right) >\min_{s_{Y}\in S_{Y}}\varphi\left(\tau_{X}^{-},s_{Y}\right) .
\end{align}
\end{lemma}
\begin{proof}
Suppose $\varphi\equiv 0$ is enforceable in expectation by some behavioral strategy $\sigma_{X}$ with initial action $\sigma_{X}^{0}$ and discount factor $\lambda=1$. We first show that the sets $\Phi_{X}^{+}$ and $\Phi_{X}^{-}$ defined in \eq{SeparationSets} are non-empty. Suppose by contradiction that this fails. Then either $\min_{s_{Y}\in S_{Y}}\varphi\left(\tau_{X},s_{Y}\right)<0$ for all $\tau_{X}\in\Delta\left(S_{X}\right)$, or $\max_{s_{Y}\in S_{Y}}\varphi\left(\tau_{X},s_{Y}\right)>0$ for all $\tau_{X}\in\Delta\left(S_{X}\right)$. Assume the first case holds (the second is analogous). Let $a\coloneqq \max_{\tau_{X}\in\Delta\left(S_{X}\right)}\min_{s_{Y}\in S_{Y}}\varphi\left(\tau_{X},s_{Y}\right)<0$ and suppose that $Y$ plays $s_{Y}^{t}\coloneqq \argmin_{s_{Y}\in S_{Y}}\varphi\left(\tau_{X}^{t},s_{Y}\right)$ at each round $t$, where $\tau_{X}^{t}$ is the mixed action generated by $\left(\sigma_{X}^{0},\sigma_{X}\right)$ at time $t$. This implies $\left(T+1\right)^{-1}\sum_{t=0}^{T}\varphi\left(\tau_{X}^{t},s_{Y}^{t}\right)\leqslant a$ for every $T\geqslant 0$. Taking limits, we obtain $\lim_{T\rightarrow\infty}\left(T+1\right)^{-1}\sum_{t=0}^{T}\varphi\left(\tau_{X}^{t},s_{Y}^{t}\right)\leqslant a<0$, contradicting the assumption that $\sigma_{X}$ is $\left(\varphi,1\right)$-autocratic. As a result, $\Phi_{X}^{+},\Phi_{X}^{-}\neq\left\{\right\}$.

Next, for $\tau_{X}^{+}\in\Phi_{X}^{+}$ and $\tau_{X}^{-}\in\Phi_{X}^{-}$, we know that $\min_{s_{Y}\in S_{Y}}\varphi\left(\tau_{X}^{+},s_{Y}\right)\geqslant 0\geqslant\max_{s_{Y}\in S_{Y}}\varphi\left(\tau_{X}^{-},s_{Y}\right)$. Suppose that at least one of these two inequalities is strict. Since no unconditional autocratic strategies exist (because $\lambda_{\min}=1$), by Proposition~\ref{prop:lambda_min} there exists $\lambda^{\ast}<1$ and a $\left(\varphi,\lambda^{\ast}\right)$-autocratic strategy, but this finding is a contradiction because $\lambda_{\min}=1$. Therefore, it must be the case that $\min_{s_{Y}\in S_{Y}}\varphi\left(\tau_{X}^{+},s_{Y}\right) =0=\max_{s_{Y}\in S_{Y}}\varphi\left(\tau_{X}^{-},s_{Y}\right)$. The remaining two strict inequalities in \eq{zeroinequalities} also follow from the assumption that $\lambda_{\min}=1$, since if either of these were not strict, the corresponding mixed action could be used to produce an unconditional autocratic strategy.
\end{proof}

\begin{lemma}\label{lem:undiscounted_gnrcc}
Suppose that $\Phi_{X}^{+},\Phi_{X}^{-}\neq\left\{\right\}$. Then, there exist a set $\mathcal{M}_{X}\subseteq\Delta\left(S_{X}\right)$, a (two-point) response function $\sigma_{X}^{\ast}:\mathcal{M}_{X}\times S_{Y}\rightarrow\mathcal{M}_{X}$, and a bounded enforcement potential $\Psi:\mathcal{M}_{X}\rightarrow\mathbb{R}$ such that
\begin{align}\label{eq:undiscountedgeneralized}
\varphi\left(\tau_{X},s_{Y}\right) =\Psi\left(\tau_{X}\right) -\Psi\left(\sigma_{X}^{\ast}\left[\tau_{X},s_{Y}\right]\right)
\end{align}
for all $\tau_{X}\in\mathcal{M}_{X}$ and $s_{Y}\in S_{Y}$.
\end{lemma}

\begin{proof}
The proof parallels that of Lemma~\ref{lem:inequalities}. Let $\tau_{X}^{+}\in\Phi_{X}^{+}$ and $\tau_{X}^{-}\in\Phi_{X}^{-}$. Consider the response function $p^{\ast}$ from \eq{pfunction} with $K=0$ and $\lambda =1$, let $\psi\left(\tau_{X}^{-}\right) =0$, and choose $\psi\left(\tau_{X}^{+}\right)$ satisfying
\begin{align}
\psi\left(\tau_{X}^{+}\right) > \max\left\{\max_{s_{Y}\in S_{Y}}\varphi\left(\tau_{X}^{+},s_{Y}\right) ,-\min_{s_{Y}\in S_{Y}}\varphi\left(\tau_{X}^{-},s_{Y}\right)\right\} . \label{eq:psi_plus_lowerbound}
\end{align}
Then, we deduce that $\psi\left(\tau_{X}^{+}\right) -\psi\left(\tau_{X}^{-}\right) >0$, which, with the bound \eq{psi_plus_lowerbound}, ensures $\textrm{Im}\left(p^{\ast}\right)\subseteq\left[0,1\right]$. Defining $\mathcal{M}_{X}$ and $\Psi$ as in Lemma~\ref{lem:inequalities} (with $p_{0}$ now free in $\left[0,1\right]$), we obtain \eq{undiscountedgeneralized}. We conclude by noting that this procedure produces a bounded enforcement potential, $\Psi$.
\end{proof}

\eq{undiscountedgeneralized} serves as the undiscounted generalized next-round correction condition, representing the limiting case of \eq{mainIntegral} as $\lambda\rightarrow 1$. We now prove Theorem~\ref{th:mainresultUndiscounted}.

\begin{proof}[Proof of Theorem~\ref{th:mainresultUndiscounted}]
Assume, without loss of generality, that no unconditional autocratic strategies exist. We consider two cases. In the first case, if there exists a $\left(\varphi,\lambda\right)$-autocratic strategy for some $\lambda\in\left[0,1\right)$, then by the proof of Theorem~\ref{th:mainresultDiscounted}, the separation sets $\Phi_{X}^{+}$ and $\Phi_{X}^{-}$ are non-empty. In the second case, if no $\left(\varphi,\lambda\right)$-autocratic strategy exists for any $\lambda\in\left[0,1\right)$, then by Lemma~\ref{lem:undiscounted_support} there exist $\tau_{X}^{\pm}\in\Delta\left(S_{X}\right)$ satisfying \eq{zeroinequalities}, which also implies $\Phi_{X}^{+},\Phi_{X}^{-}\neq\left\{\right\}$. In either case, Lemma~\ref{lem:undiscounted_gnrcc} ensures that \eq{undiscountedgeneralized} holds. Analogous to the proof of Proposition~\ref{prop:gnrcc_iff}, for any behavioral strategy $\sigma_{Y}$, we obtain a sequence of distributions $\left\{\nu_{X}^{t}\right\}_{t=0}^{\infty}$ over $\Delta\left(S_{X}\right)\times S_{Y}$ by marginalizing over the history up to time $t$; this is the time-$t$ law of $\left(\tau_{X},s_{Y}\right)$ (i.e., $X$'s mixed action and $Y$'s realized pure action). Telescoping the undiscounted generalized next-round correction condition yields
\begin{align}\label{eq:TelescopicUndiscounted}
\frac{1}{T+1}\sum_{t=0}^{T}\mathbb{E}_{\left(\tau_{X},s_{Y}\right)\sim\nu_{X}^{t}}\left[\varphi\left(\tau_{X},s_{Y}\right)\right] =\frac{1}{T+1}\left(\mathbb{E}_{\left(\tau_{X},s_{Y}\right)\sim\nu_{X}^{0}}\left[\Psi\left(\tau_{X}\right)\right]-\mathbb{E}_{\left(\tau_{X},s_{Y}\right)\sim\nu_{X}^{T+1}}\left[\Psi\left(\tau_{X}\right)\right]\right)
\end{align}
for every $T\geqslant 0$. By boundedness of $\Psi$ and the dominated convergence theorem, the right-hand side of \eq{TelescopicUndiscounted} converges to $0$ as $T\rightarrow\infty$.
\end{proof}

\begin{remark}
It is important to highlight that, unlike in the discounted case, the initial action $\sigma_{X}^{0}$ plays no part in the enforceability of $\varphi$; any ``suboptimal'' initial choice will eventually be corrected as the stage game repeats itself infinitely many times. It plays a major role, however, in the conditioning structure of the response function, $\sigma_{X}^{\ast}$. If $p^{\ast}\left[p,s_{Y}\right]\in\left\{0,1\right\}$, then $p\in\left\{0,1\right\}$ and $s_{Y}\in\left\{\argmax_{s_{Y}\in S_{Y}}\varphi\left(\tau_{X}^{-},s_{Y}\right) ,\argmin_{s_{Y}\in S_{Y}}\varphi\left(\tau_{X}^{+},s_{Y}\right)\right\}$. Inductively, we deduce that if $p_{0}\in\left(0,1\right)$, then $\textrm{Im}\left(p^{\ast}\right)\subseteq\left(0,1\right)$. This essentially means that the Markov chain generated by $\left(\sigma_{X}^{0},\sigma_{X}^{\ast}\right)$ (and any behavioral strategy of $Y$) is ergodic. However, if $X$ desires to enforce $\varphi\equiv 0$ via a simple reactive learning strategy, she can do so by choosing $p_{0}=0$ or $p_{0}=1$.
\end{remark}

By Lemmas~\ref{lem:undiscounted_support}, \ref{lem:undiscounted_gnrcc} and Theorem~\ref{th:mainresultUndiscounted}, we obtain a version of Proposition~\ref{prop:lambda_min} in the case where a payoff constraint is only enforceable in the undiscounted setting.

\begin{proposition}\label{prop:separationundiscounted}
A function $\varphi:S_{X}\times S_{Y}\rightarrow\mathbb{R}$ is enforceable with $\lambda_{\min}=1$ if and only if $\Phi_{X}^{+},\Phi_{X}^{-}\neq\left\{\right\}$ and
\begin{align*}
\Phi_{X}^{+}&=\left\{\tau_{X}\in\Delta\left(S_{X}\right)\mid\max_{s_{Y}\in S_{Y}}\varphi\left(\tau_{X},s_{Y}\right)>\min_{s_{Y}\in S_{Y}}\varphi\left(\tau_{X},s_{Y}\right)=0\right\};\\ \Phi_{X}^{-}&=\left\{\tau_{X}\in\Delta\left(S_{X}\right)\mid\min_{s_{Y}\in S_{Y}}\varphi\left(\tau_{X},s_{Y}\right)<\max_{s_{Y}\in S_{Y}}\varphi\left(\tau_{X},s_{Y}\right)=0\right\}.
\end{align*}
\end{proposition}

The following result extends \cite[Proposition 6]{hilbe:GEB:2015}. In simple terms, a player's control over the long-run payoff outcome is not altered if the game is arbitrarily extended.

\begin{proposition}\label{prop:gamelength}
Fix $\varphi:S_{X}\times S_{Y}\to\mathbb{R}$ and $\lambda\in\left[0,1\right]$. If there exists a $\left(\varphi ,\lambda\right)$-autocratic strategy, then there exists a (two-point reactive learning) $\left(\varphi ,\lambda^{\ast}\right)$-autocratic strategy for any $\lambda^{\ast}\in\left[\lambda ,1\right]$.
\end{proposition}

\begin{proof}
If $\lambda =1$, then we have nothing to show, so we may assume that $\lambda\in\left[0,1\right)$. From Proposition~\ref{prop:lambda_min}, we know that $\Phi_{X}^{+},\Phi_{X}^{-}\neq\left\{\right\}$ and either \emph{(i)} $\Phi_{X}^{+}\cap\Phi_{X}^{-}\neq\left\{\right\}$ or \emph{(ii)} $\Phi_{X}^{+}\cap\Phi_{X}^{-}=\left\{\right\}$ and $\lambda\geqslant\lambda_{\min}$, where $\lambda_{\min}$ is given by \eq{lambda_min}. In case \emph{(i)}, for any $\tau_{X}\in\Phi_{X}^{+}\cap\Phi_{X}^{-}$, we have $\varphi\left(\tau_{X},s_{Y}\right) =0$ for all $s_{Y}\in S_{Y}$. Thus, playing $\tau_{X}$ in every round constitutes an autocratic strategy for any $\lambda^{\ast}\in\left[0,1\right]$. In case \emph{(ii)}, we have $\lambda^{\ast}\geqslant\lambda\geqslant\lambda_{\min}$, and for $\lambda^{\ast}<1$, Proposition~\ref{prop:lambda_min} gives a (two-point reactive learning) $\left(\varphi ,\lambda^{\ast}\right)$-autocratic strategy. For $\lambda^{\ast}=1$, Lemma~\ref{lem:undiscounted_gnrcc} can be used since $\Phi_{X}^{+},\Phi_{X}^{-}\neq\left\{\right\}$, and then the proof of Theorem~\ref{th:mainresultUndiscounted} shows that the resulting reactive learning strategy is $\left(\varphi ,1\right)$-autocratic.
\end{proof}

\section{Properties of payoff relationships and autocratic strategies}\label{sec:properties}
Having established that any enforceable payoff relationship can be implemented using two-point reactive learning strategies (Theorem~\ref{th:mainresultDiscounted}), we now investigate the computational and structural properties of autocratic strategies. We show that verifying enforceability and computing optimal strategies can be accomplished in polynomial time using linear programming. We also characterize the discounted range needed when symmetric incentives are considered, and we further establish restrictions for a ``selfish'' or ``punishing'' enforcer.

\subsection{Enforceable relationships and the interval of enforceability}
While the generalized next-round correction condition provides a characterization of autocratic strategies, verifying it directly can be challenging in practice, as it requires finding an appropriate map $\Psi$ on the reachable set, $\mathcal{M}_{X}$. In this section, we provide a simple, computationally tractable criterion for determining enforceability. The key observation is that enforceability depends only on whether an ``interval of enforceability'' contains zero, which can be checked by solving two saddle point problems: $\min_{\tau_{X} \in \Delta\left(S_{X}\right)} \max_{s_{Y} \in S_{Y}} \varphi\left(\tau_{X},s_{Y}\right)$ and $\max_{\tau_{X} \in \Delta\left(S_{X}\right)} \min_{s_{Y} \in S_{Y}} \varphi\left(\tau_{X},s_{Y}\right)$.

\begin{definition}
For a function $\varphi :S_{X}\times S_{Y}\rightarrow\mathbb{R}$, the interval of enforceability is
\begin{align*}
J\left(\varphi\right) &\coloneqq  \left[\min_{\tau_{X}\in \Delta\left(S_{X}\right)}\max_{s_{Y}\in S_{Y}}\varphi\left(\tau_{X},s_{Y}\right),\max_{\tau_{X}\in \Delta\left(S_{X}\right)}\min_{s_{Y}\in S_{Y}}\varphi\left(\tau_{X},s_{Y}\right) 
\right] .
\end{align*}
Clearly, $J\left(\varphi\right)\neq\left\{\right\}$ if and only if $\min_{\tau_{X}\in \Delta\left(S_{X}\right)}\max_{s_{Y}\in S_{Y}}\varphi\left(\tau_{X},s_{Y}\right)\leqslant\max_{\tau_{X}\in \Delta\left(S_{X}\right)}\min_{s_{Y}\in S_{Y}}\varphi\left(\tau_{X},s_{Y}\right)$.
\end{definition}

The interval $J\left(\varphi\right)$ encodes geometric information about the separability of payoff values. Since the functions $\max_{s_{Y}\in S_{Y}}\varphi\left(\cdot ,s_{Y}\right)$ and $\min_{s_{Y}\in S_{Y}}\varphi\left(\cdot,s_{Y}\right)$ are continuous over the compact set $\Delta\left(S_{X}\right)$, the extrema defining $J\left(\varphi\right)$ are always attained. The following characterization connects enforceability to the condition $0 \in J\left(\varphi\right)$, which has a natural interpretation: player $X$ can enforce $\varphi \equiv 0$ if and only if she can find mixed actions that ``sandwich'' zero between the worst and best values of $\varphi$ across opponent responses. Its proof follows by combining Propositions~\ref{prop:lambda_min} and~\ref{prop:separationundiscounted}.

\begin{corollary}\label{cor:enforceability}
Consider a function
$\varphi:S_{X}\times S_{Y}\rightarrow\mathbb{R}$. The following are equivalent:
\begin{enumerate}
\item[\emph{(a)}] $\varphi\equiv 0$ is enforceable;
\item[\emph{(b)}] $\Phi_{X}^{+},\Phi_{X}^{-}\neq\left\{\right\}$;
\item[\emph{(c)}] $0\in J\left(\varphi\right)$.
\end{enumerate}
\end{corollary}

\subsection{Computational tractability of enforceable relationships}
Suppose that $S_{X}=\left\{U,D\right\}$, $S_{Y}=\left\{L,R\right\}$, and that
\begin{align*}
\varphi &= \bordermatrix{%
& L & R \cr
U &\ 4 & \ 1 \cr
D &\ -1 & \ 0 \cr
} .
\end{align*}
Let $\tau_{X}^{+}$ be the mixed action that plays $U$ and $D$ uniformly at random, and let $\tau_{X}^{-}$ be the pure action $D$. Then, $\varphi\left(\tau_{X}^{+},L\right) =3/2$, $\varphi\left(\tau_{X}^{+},R\right) =1/2$, $\varphi\left(\tau_{X}^{-},L\right) =-1$, and $\varphi\left(\tau_{X}^{-},R\right) =0$, from which we see that $\tau_{X}^{+}\in\Phi_{X}^{+}$ and $\tau_{X}^{-}\in\Phi_{X}^{-}$, and the two inequalities of \eq{lemma_inequality} hold with $\lambda =2/3$ (in fact, they are both equalities). However, if we set $\tau_{X}^{+}$ to be $U$ and $\tau_{X}^{-}$ to be $D$, then the first of these two inequalities fails to hold. In fact, the minimum $\lambda$ for which they hold when restricted to pure actions is $\lambda =3/4$. Therefore, randomizing between two mixed actions, $\tau_{X}^{+}$ and $\tau_{X}^{-}$, might be preferable to randomizing between two pure actions, for impatient players.

However, under somewhat restrictive conditions, we can guarantee that $\lambda_{\min}$ is attained for pure actions, $\tau_{X}^{\pm}\in\Delta\left(S_{X}\right)$.
\begin{lemma}\label{lem:two_action_dominance}
Suppose $S_{X}=\left\{U,D\right\}$, $S_{Y}=\left\{L,R\right\}$, and let $\varphi :S_{X}\times S_{Y}\rightarrow\mathbb{R}$ satisfy $\varphi\left(U,L\right)\geqslant\varphi\left(U,R\right)$ and $\varphi\left(D,L\right)\geqslant\varphi\left(D,R\right)$. If $\varphi\equiv 0$ is enforceable, then $\lambda_{\min}$ is attained at pure actions.
\end{lemma}

\begin{proof}
Suppose that $\tau_{X}^{+}\in\Phi_{X}^{+}$ and $\tau_{X}^{-}\in\Phi_{X}^{-}$ are any two mixed actions satisfying the inequalities of \eq{lemma_inequality} for some $\lambda\in\left[0,1\right)$. Without a loss of generality, we may assume that $\delta_{U}\in\Phi_{X}^{+}$ and $\delta_{D}\in\Phi_{X}^{-}$ because $\varphi\left(\tau_{X}^{+},R\right)\geqslant 0$ means that there must be a pure action $s_{X}^{+}\in\left\{U,D\right\}$ in the support of $\tau_{X}^{+}$ such that $\varphi\left(s_{X}^{+},R\right)\geqslant 0$, and similarly for $\tau_{X}^{-}$ (but with $L$ replacing $R$). Since $S_{X}$ has only two options, $\tau_{X}^{+}$ and $\tau_{X}^{-}$ can be represented by the probabilities of playing $U$, denoted $p$ and $q$, respectively. Moreover, because $\max_{s_{Y}\in S_{Y}}\varphi\left(\tau_{X}^{-},s_{Y}\right) =\varphi\left(\tau_{X}^{-},L\right)$ is a non-decreasing function of $q$, we see that \eq{lemma_inequality1} holds when $\tau_{X}^{-}$ is replaced by $\delta_{D}$. This first inequality, \eq{lemma_inequality1}, then says
\begin{align*}
\varphi\left(\tau_{X}^{+},R\right) &\geqslant \left(1-\lambda\right)\varphi\left(\tau_{X}^{+},L\right) +\lambda\varphi\left(D,L\right) .
\end{align*}
This inequality is linear in $p$, so it must also hold when $p=0$ or when $p=1$. If it holds when $p=0$, then $\varphi\left(D,R\right)\geqslant\varphi\left(D,L\right)$. Since it must be true that $\varphi\left(D,L\right) =\varphi\left(D,R\right)$ or $\varphi\left(D,L\right) >\varphi\left(D,R\right)$, we see at once that the inequality holds for $p=1$ and thus for the pair $\left(\delta_{U},\delta_{D}\right)$. That \eq{lemma_inequality2} also holds for the pair $\left(\delta_{U},\delta_{D}\right)$ is analogous, this time increasing $p$ to $1$ first and then $q$ to $0$.
\end{proof}

To see that Lemma~\ref{lem:two_action_dominance} does not extend to larger action spaces for $X$, consider the function
\begin{align*}
\varphi &= \bordermatrix{%
& L & R \cr
U &\ -2 & \ -\frac{2}{5} \cr
M &\ 2 & \ 5 \cr
D &\ 1 & \ 2 \cr
} .
\end{align*}
For pure action candidates, we must have $s_{X}^{+}\in\left\{M,D\right\}$ and $s_{X}^{-}=U$. Suppose that $\lambda =1/2$, and let $\tau_{X}^{+}=\left(3/10\right)\delta_{M}+\left(7/10\right)\delta_{D}$ and $\tau_{X}^{-}=U$. We then calculate that
\begin{align*}
\frac{13}{10} = \min_{s_{Y}\in S_{Y}}\varphi\left(\tau_{X}^{+},s_{Y}\right) &\geqslant \left(1-\lambda\right)\max_{s_{Y}\in S_{Y}}\varphi\left(\tau_{X}^{+},s_{Y}\right) +\lambda\max_{s_{Y}\in S_{Y}}\varphi\left(\tau_{X}^{-},s_{Y}\right) = \frac{25}{20} ; \\
-\frac{2}{5} = \max_{s_{Y}\in S_{Y}}\varphi\left(\tau_{X}^{-},s_{Y}\right) &\leqslant \lambda\min_{s_{Y}\in S_{Y}}\varphi\left(\tau_{X}^{+},s_{Y}\right) +\left(1-\lambda\right)\min_{s_{Y}\in S_{Y}}\varphi\left(\tau_{X}^{-},s_{Y}\right) = -\frac{7}{20} .
\end{align*}
However, if we consider the pure-action candidates for $s_{X}^{+}$ and $s_{X}^{-}$ for $\lambda =1/2$, then, with $s_{X}^{+}=M$ and $s_{X}^{-}=U$, the first inequality is
\begin{align*}
2 = \min_{s_{Y}\in S_{Y}}\varphi\left(s_{X}^{+},s_{Y}\right) &\geqslant \left(1-\lambda\right)\max_{s_{Y}\in S_{Y}}\varphi\left(s_{X}^{+},s_{Y}\right) +\lambda\max_{s_{Y}\in S_{Y}}\varphi\left(s_{X}^{-},s_{Y}\right) = \frac{23}{10} ,
\end{align*}
which does not hold. If $s_{X}^{+}=D$ and $s_{X}^{-}=U$, then the second inequality is
\begin{align*}
-\frac{2}{5} = \max_{s_{Y}\in S_{Y}}\varphi\left(s_{X}^{-},s_{Y}\right) &\leqslant \lambda\min_{s_{Y}\in S_{Y}}\varphi\left(s_{X}^{+},s_{Y}\right) +\left(1-\lambda\right)\min_{s_{Y}\in S_{Y}}\varphi\left(s_{X}^{-},s_{Y}\right) = -\frac{1}{2} ,
\end{align*}
which also does not hold.

These examples demonstrate that mixed actions can achieve better (smaller) values of $\lambda_{\min}$ than pure actions, even in small games. This motivates the need for efficient algorithms that search over the full space $\Delta\left(S_{X}
\right)$ rather than just pure actions. We now show that despite this complexity, both the verification of enforceability and the computation of optimal autocratic strategies are tractable via linear programming.

To specify a two-point reactive learning strategy explicitly, we need only provide a small number of parameters. For a strategy $\left(\sigma_{X}^{0},\sigma_{X}^{\ast}\right)$ defined by transition probabilities $\left(p_{0},p^{\ast}\right)$ (as in Lemmas~\ref{lem:inequalities}~and~\ref{lem:undiscounted_gnrcc}), we define its base to be the finite set of probabilities that fully characterize the strategy.

\begin{definition}
Consider a two-point reactive learning strategy $\left(\sigma_{X}^{0},\sigma_{X}^{\ast}\right)$ defined by $\left(p_{0},p^{\ast}\right)$ of either Lemma~\ref{lem:inequalities} or Lemma~\ref{lem:undiscounted_gnrcc}. The set $\left\{p_{0}\right\}\cup\left\{p^{\ast}\left[0,s_{Y}\right] ,p^{\ast}\left[1,s_{Y}\right]\right\}_{s_{Y}\in S_{Y}}$ is called a ``base'' of $\left(\sigma_{X}^{0},\sigma_{X}^{\ast}\right)$.
\end{definition}

The notion of a base arises from the fact that
\begin{align}\label{eq:convexproperty}
p^{\ast}\left[q,s_{Y}\right]=qp^{\ast}\left[1,s_{Y}\right] +\left(1-q\right) p^{\ast}\left[0,s_{Y}\right]
\end{align}
for all $q\in\left[0,1\right]$ and $s_{Y}\in S_{Y}$.

\begin{proposition}\label{prop:tractability}
Consider a function $\varphi:S_{X}\times S_{Y}\rightarrow\mathbb{R}$. The problem of identifying whether $\varphi\equiv 0$ is enforceable can be solved in polynomial time. In addition, if $\varphi\equiv 0$ is enforceable, then the minimum discount factor $\lambda_{\min}$ and a base of some two-point reactive learning $\left(\varphi ,\lambda_{\min}\right)$-autocratic strategy can both be computed in polynomial time.
\end{proposition}

\begin{proof}
Let $S_{X}\coloneqq\left\{s_{X,1},\dots ,s_{X,m}\right\}$ and $S_{Y}\coloneqq\left\{s_{Y,1},\dots ,s_{Y,n}\right\}$. We can think of $\varphi$ as an $m\times n$ matrix with entries $\phi_{ij}=\varphi\left(s_{X,i},s_{Y,j}\right)$. Recall that $\varphi\equiv 0$ is enforceable if and only if $\Phi_{X}^{+},\Phi_{X}^{-}\neq\left\{\right\}$. $\Phi_{X}^{+}\neq\left\{\right\}$ is equivalent to $\phi^{\top}x\geqslant 0$ for $x\in\mathbb{R}_{+}^{m}$ and $\sum_{i=1}^{m}x_{i}=1$, while $\Phi_{X}^{-}\neq\left\{\right\}$ is equivalent to $\phi^{\top}x\leqslant 0$ for $x\in\mathbb{R}_{+}^{m}$ and $\sum_{i=1}^{m}x_{i}=1$, both of which can be solved in polynomial time. In addition, $\varphi\equiv 0$ is $0$-enforceable if and only if the linear system $\phi^{\top}x=0$ with $x\in\mathbb{R}_{+}^{m}$ and $\sum_{i=1}^{m}x_{i}=1$ is feasible, which is also a computationally tractable task.

Suppose now that $\varphi\equiv 0$ is enforceable. If $\varphi$ is $0$-enforceable, then we can compute the unconditional autocratic strategy (mixed action) in polynomial time as above. Therefore, we may assume that $\varphi\equiv 0$ is not $0$-enforceable. As $\Phi_{X}^{+},\Phi_{X}^{-}\neq\left\{\right\}$, the supremum problem in \eq{lambda_min} is well-defined and its value is attained. In fact, by standard max-min techniques and the Charnes-Cooper transformation, the supremum problem can be equivalently split into the following linear programming forms:
\begin{center}
\setlength{\arraycolsep}{15pt}
\begin{equation*}{\left(P_{\lambda_{\min}}^{1}\right)\hspace{1cm}}
\begin{array}{rrr}
\text{maximize} & z^{+} - w^{-} & \\[4mm]
\text{subject to} & z^{+}\leqslant\sum_{i=1}^{m}x_{i}^{+}\phi_{ij} & j=1,\dots , n\\[1.5mm]
& w^{+}\geqslant\sum_{i=1}^{m}x_{i}^{+}\phi_{ij} & j=1,\dots ,n\\[1.5mm]
& z^{-}\leqslant\sum_{i=1}^{m}x_{i}^{-}\phi_{ij} & j=1,\dots ,n\\[1.5mm]
& w^{-}\geqslant\sum_{i=1}^{m}x_{i}^{-}\phi_{ij} & j=1,\dots ,n\\[1.5mm]
& x_{i}^{\pm}\geqslant 0 & i=1,\dots ,m\\[1.5mm]
& \sum_{i=1}^{m}x_{i}^{\pm}=t\\[1.5mm]
& w^{+}-w^{-} = 1\\[1.5mm]
& z^{+}-z^{-}\leqslant 1\\[1.5mm]
& z^{+}\geqslant 0 \\[1.5mm]
& w^{-}\leqslant 0\\[1.5mm]
& t\geqslant 0
\end{array}
\end{equation*}
\\
\begin{equation*}{\left(P_{\lambda_{\min}}^{2}\right)\hspace{1cm}}
\begin{array}{rrr}
\text{maximize} & z^{+} - w^{-} & \\[4mm]
\text{subject to} & z^{+}\leqslant\sum_{i=1}^{m}x_{i}^{+}\phi_{ij} & j=1,\dots ,n\\[1.5mm]
& w^{+}\geqslant\sum_{i=1}^{m}x_{i}^{+}\phi_{ij} & j=1,\dots ,n\\[1.5mm]
& z^{-}\leqslant\sum_{i=1}^{m}x_{i}^{-}\phi_{ij} & j=1,\dots ,n\\[1.5mm]
& w^{-}\geqslant\sum_{i=1}^{m}x_{i}^{-}\phi_{ij} & j=1,\dots ,n\\[1.5mm]
& x_{i}^{\pm}\geqslant 0 & i=1,\dots ,m\\[1.5mm]
& \sum_{i=1}^{m}x_{i}^{\pm}=t\\[1.5mm]
& w^{+}-w^{-}\leqslant 1\\[1.5mm]
& z^{+}-z^{-}= 1\\[1.5mm]
& z^{+}\geqslant 0 \\[1.5mm]
& w^{-}\leqslant 0\\[1.5mm]
& t\geqslant 0
\end{array}
\end{equation*}
\end{center}
Here, $z^{\pm}$ and $w^{\pm}$ represent the scaled min and max of $\varphi$, respectively, over $s_{Y}\in S_{Y}$, at $\tau_{X}^{\pm}$, and $t$ is the Charnes-Cooper scaling. The two programs correspond to the two pieces of the piecewise-linear denominator in \eq{lambda_min}. Normalizing the first piece to one gives $\left(P_{\lambda_{\min}}^{1}\right)$ and normalizing the second gives $\left(P_{\lambda_{\min}}^{2}\right)$. Because the auxiliary variables $w^{+}$ and $z^{-}$ are bounded by the data on one side only, neither normalization restricts the set of feasible pairs, and the two programs, in fact, reduce to the following single linear program:
\begin{center}
\setlength{\arraycolsep}{15pt}
\begin{equation*}{\left(P_{\lambda_{\min}}\right)\hspace{1cm}}
\begin{array}{rrr}
\text{maximize} & z - w & \\[4mm]
\text{subject to} & z\leqslant\sum_{i=1}^{m}x_{i}^{+}\phi_{ij} & j=1,\dots , n\\[1.5mm]
& w\geqslant -1+\sum_{i=1}^{m}x_{i}^{+}\phi_{ij} & j=1,\dots ,n\\[1.5mm]
& z\leqslant 1+\sum_{i=1}^{m}x_{i}^{-}\phi_{ij} & j=1,\dots ,n\\[1.5mm]
& w\geqslant\sum_{i=1}^{m}x_{i}^{-}\phi_{ij} & j=1,\dots ,n\\[1.5mm]
& x_{i}^{\pm}\geqslant 0 & i=1,\dots ,m\\[1.5mm]
& \sum_{i=1}^{m}x_{i}^{\pm}=t\\[1.5mm]
& z\geqslant 0 \\[1.5mm]
& w\leqslant 0\\[1.5mm]
& t\geqslant 0
\end{array}
\end{equation*}
\end{center}
Here, $z$ and $w$ play the roles of $z^{+}$ and $w^{-}$. Note that because $\lambda_{\min}\in\left(0,1\right]$, we can always find an optimal solution for $\left(P_{\lambda_{\min}}\right)$ with $t>0$. We can then set $\tau_{X}^{\pm}\coloneqq x^{\pm}/t$ to retrieve the desired mixed actions in $\Phi_{X}^{+}$ and $\Phi_{X}^{-}$, and we set $\lambda_{\min}=1-\left(z-w\right)$. The pair $\tau_{X}^{\pm}$ then leads to a two-point reactive learning $\left(\varphi ,\lambda_{\min}\right)$-autocratic strategy $\left(\sigma_{X}^{0},\sigma_{X}^{\ast}\right)$, as seen from Lemma~\ref{lem:inequalities}.

Finally, in order to identify a base of $\left(\sigma_{X}^{0},\sigma_{X}^{\ast}\right)$, we need only compute at most $2n+1$ values, as indicated by \eq{p_0discounted}, \eq{p_0undiscounted}, and \eq{convexproperty}.
\end{proof}

\subsection{Equalizing and self-equalizing strategies for generic payoff functions}\label{sec:equalizers}
$X$ can set her own payoff to some constant, $K$, if $u_{X}-K\equiv 0$ is enforceable. Assuming that the aim of player $X$ is setting her own score, we deduce from Corollary~\ref{cor:enforceability} that she can set her payoff equal to $K$ if and only if
\begin{align}\label{eq:selfequalizing}
K\in\left[\min_{\tau_{X}\in\Delta\left(S_{X}\right)}\max_{s_{Y}\in S_{Y}}u_{X}\left(\tau_{X},s_{Y}\right),\max_{\tau_{X}\in \Delta\left(S_{X}\right)}\min_{s_{Y}\in S_{Y}}u_{X}\left(\tau_{X},s_{Y}\right)
\right] .
\end{align}
Note that the maximum payoff value that can be set by $X$ is $\max_{\tau_{X}\in\Delta\left(S_{X}\right)}\min_{s_{Y}\in S_{Y}}u_{X}\left(\tau_{X},s_{Y}\right)$. This results in a simple relationship between an agent's control of their payoff in a repeated game and that of their optimal return in a one-shot game: an agent can achieve their (average) security level when the enforceability interval \eq{selfequalizing} is well-defined.

Accordingly, suppose player $X$ has an incentive to set her opponent's expected payoff. This case corresponds to enforcing $u_{Y}-K\equiv 0$. In similar fashion to the above, it can be shown that if she desires to ``punish'' her opponent in the worst possible way, the best she can hope for is setting $Y$'s payoff to $\min_{\tau_{X}\in \Delta\left(S_{X}\right)}\max_{s_{Y}\in S_{Y}}u_{Y}\left(\tau_{X},s_{Y}\right)$. Since all functions are extended from pure actions to mixed actions in a linear fashion, the inner maximum is unchanged when $Y$ mixes. By the minimax theorem, this quantity is equal to $\max_{\tau_{Y}\in\Delta\left(S_{Y}\right)}\min_{\tau_{X}\in \Delta\left(S_{X}\right)}u_{Y}\left(\tau_{X},\tau_{Y}\right)$. This is precisely the security level of the co-player.

\subsection{Symmetric functions}
Next, we prove that an enforceable symmetric relationship is either enforceable by an unconditional strategy or enforceable with $\lambda_{\min}=1$. In other words, if $X$'s aim is to control a symmetric payoff region, then she either needs an unconditional plan or an ``infinite'' time horizon.

\begin{proposition}\label{prop:symmetricproperty}
Suppose that $S_{X}=S_{Y}$ and that $\varphi :S_{X}\times S_{Y}\rightarrow\mathbb{R}$ is symmetric, meaning $\varphi\left(s_{X},s_{Y}\right) =\varphi\left(s_{Y},s_{X}\right)$ for all $s_{X}\in S_{X}$ and $s_{Y}\in S_{Y}$. Then, $\varphi\equiv 0$ is enforceable if and only if there exists an unconditional autocratic strategy for $\varphi$ or $\max_{\tau_{X}\in \Delta\left(S_{X}\right)}\min_{\tau_{Y}\in\Delta\left(S_{Y}\right)}\varphi\left(\tau_{X},\tau_{Y}\right) =0$.
\end{proposition}

\begin{proof}
Suppose that $\varphi$ is not $0$-enforceable. We know $\max_{\tau_{X}\in \Delta\left(S_{X}\right)}\min_{\tau_{Y}\in\Delta\left(S_{Y}\right)}\varphi\left(\tau_{X},\tau_{Y}\right) =0$ implies that $J\left(\varphi\right) =\left\{0\right\}$ by the minimax theorem and the symmetry of $\varphi$, thus $\varphi\equiv 0$ is enforceable with $\lambda_{\min}=1$ by Proposition \ref{prop:separationundiscounted}. Conversely, let $\varphi\equiv 0$ be enforceable but not by an unconditional strategy. Then, it is $\lambda$-enforceable with $\lambda\in\left(0,1\right]$ and $J\left(\varphi\right)$ is well-defined and it contains $0$ due to Corollary \ref{cor:enforceability}. However, the symmetry of $\varphi$ and the minimax theorem imply that $J\left(\varphi\right)$ has to contain only one element, i.e., the assumption of Proposition~\ref{prop:separationundiscounted} holds, and the maximin equality follows.
\end{proof}

It is important to note that Proposition~\ref{prop:symmetricproperty} also holds for skew-symmetric relationships (i.e., for functions $\varphi$ satisfying $\varphi\left(s_{X},s_{Y}\right) =-\varphi\left(s_{Y},s_{X}\right)$). Indeed, such a $\varphi$ itself gives a symmetric, zero-sum game whose value is thus $0$, giving the maximin condition of Proposition~\ref{prop:symmetricproperty} automatically. We therefore deduce that TFT, which corresponds to $\varphi\left(s_{X},s_{Y}\right) =u_{Y}\left(s_{X},s_{Y}\right) -u_{X}\left(s_{X},s_{Y}\right)$, can be an autocratic strategy and enforce fair payoff relationships for symmetric games (such as the repeated prisoner's dilemma) only in the undiscounted setting. In Section~\ref{sec:applications}, we show that the payoff equality can be enforced with $\lambda <1$ in asymmetric games, however.

\subsection{``Zero-sum'' strategies for generic payoff functions}
Here, we say that an autocratic strategy for $X$ is ``zero-sum'' if it enforces $\pi_{X}=-\pi_{Y}$. Corollary \ref{cor:enforceability} directly yields the following:

\begin{proposition}\label{prop:zerosum}
There exists a zero-sum autocratic strategy for player $X$ if and only if there exist $\tau_{X}^{\pm}\in\Delta\left(S_{X}\right)$ such that $u_{X}\left(\tau_{X}^{+},s_{Y}\right)\geqslant -u_{Y}\left(\tau_{X}^{+},s_{Y}\right)$ and $u_{X}\left(\tau_{X}^{-},s_{Y}\right)\leqslant -u_{Y}\left(\tau_{X}^{-},s_{Y}\right)$ for all $s_{Y}\in S_{Y}$.
\end{proposition}

\section{Applications and examples}\label{sec:applications}
We now apply our theoretical framework to four games: the prisoner's dilemma, a nonlinear donation game, an asymmetric donation game, and the hawk-dove game. We focus on linear payoff relationships relative to a reference line, such that the desired relationship is either exactly the reference line or else intersects it in a unique point. For example, the line $\kappa -\pi_{X}=\chi\left(\kappa -\pi_{Y}\right)$ intersects the reference line $\pi_{X}=\pi_{Y}$ at $\left(\kappa ,\kappa\right)$ provided $\chi\neq 1$. We conclude with an example of how one can get two-dimensional feasible regions from autocratic strategies by choosing $\varphi$ to be non-affine, meaning not of the form $\alpha u_{X}+\beta u_{Y}+\gamma$ for any $\alpha ,\beta ,\gamma\in\mathbb{R}$.

\subsection{Prisoner's dilemma}
In the prisoner's dilemma, the payoff matrix of \eq{payoff_matrix} satisfies $T>R>P>S$. We impose the standard constraint of $2P<S+T<2R$ as well, which means that the mean payoff for alternating cooperation and defection lies between the payoffs for mutual defection and mutual cooperation. This condition is not strictly necessary, but we use it to simplify the exposition. In what follows, we consider the function $\varphi\left(s_{X},s_{Y}\right) =\kappa -u_{X}\left(s_{X},s_{Y}\right) -\chi\left(\kappa -u_{Y}\left(s_{X},s_{Y}\right)\right)$, i.e.
\begin{align*}
\varphi\left(C,C\right) &= \left(\chi -1\right)\left(R-\kappa\right) ; \\
\varphi\left(C,D\right) &= \left(\kappa -S\right) +\chi\left(T-\kappa\right) ; \\
\varphi\left(D,C\right) &= -\left(T-\kappa\right) - \chi\left(\kappa -S\right) ; \\
\varphi\left(D,D\right) &= -\left(\chi -1\right)\left(\kappa -P\right) .
\end{align*}
The case $\kappa =P$ and $\chi >1$ corresponds to ``extortionate'' strategies \citep{press:PNAS:2012}, while ``generous'' strategies have $\kappa =R$ and $\chi >1$ \citep{stewart:PNAS:2012,stewart:PNAS:2013,hilbe:JTB:2015}. So-called ``equalizers'' (Section~\ref{sec:equalizers}) in this parameterization correspond to $\chi\rightarrow\infty$.

\subsubsection{Unconditional autocratic strategies}
In this game, an unconditional autocratic strategy is a value $p\in\left[0,1\right]$ such that $\varphi\left(p,C\right) =\varphi\left(p,D\right) =0$. For every $p\in\left[0,1\right]$, we can simply solve for $\kappa$ and $\chi$ in these equations to obtain
\begin{align*}
\kappa &= P + p\left(S + T - 2P\right) + p^{2}\left(R - S - T + P\right) ; \\
\chi &= \frac{T - P + p\left(R - S - T + P\right)}{S - P + p\left(R - S - T + P\right)} .
\end{align*}
For this pair $\left(\kappa ,\chi\right)$, playing $C$ with probability $p$ in every round is $\left(\varphi ,0\right)$-autocratic.

\subsubsection{Conditional autocratic strategies}
We now assume that we have already classified unconditional autocratic strategies and we are in the situation in which $\varphi\equiv 0$ is not $0$-enforceable. By Lemma~\ref{lem:two_action_dominance}, we can restrict the search for $\left(\tau_{X}^{+},\tau_{X}^{-}\right)$ in Proposition~\ref{prop:lambda_min} to pairs of pure actions provided that the differences
\begin{align*}
\varphi\left(C,C\right) - \varphi\left(C,D\right) &= -\chi\left(T-R\right) -\left(R-S\right) ; \\
\varphi\left(D,C\right) - \varphi\left(D,D\right) &= -\chi\left(P-S\right) -\left(T-P\right)
\end{align*}
both have the same sign. This condition gives two branches for $\chi$, namely
\begin{subequations}
\begin{align}
\chi &\geqslant \max\left\{ -\frac{R-S}{T-R} , -\frac{T-P}{P-S} \right\} \label{eq:pd_chi_branch1} ; \\
\textrm{and }\chi &\leqslant \min\left\{ -\frac{R-S}{T-R} , -\frac{T-P}{P-S} \right\} . \label{eq:pd_chi_branch2}
\end{align}
\end{subequations}

In the first branch (\eq{pd_chi_branch1}), we have $\varphi\left(C,C\right)\leqslant\varphi\left(C,D\right)$ and $\varphi\left(D,C\right)\leqslant\varphi\left(D,D\right)$. Therefore, a necessary condition for both $\Phi_{X}^{+}\neq\left\{\right\}$ and $\Phi_{X}^{-}\neq\left\{\right\}$ (in the context of Proposition~\ref{prop:lambda_min}) is that $\varphi\left(C,C\right)$ and $\varphi\left(D,D\right)$ have opposite signs. For $\chi\neq 1$, this condition is equivalent to $\kappa\in\left[P,R\right]$. (Note that for $\chi =1$, the value of $\kappa$ is irrelevant.) There are then two possibilities \citep[see also][]{hilbe:GEB:2015}:
\begin{subequations}
\begin{align}
\varphi\left(D,D\right)\leqslant 0\leqslant\varphi\left(C,C\right) &\iff \chi \geqslant 1 ; \label{eq:branch1_1} \\
\varphi\left(C,D\right)\leqslant 0\leqslant\varphi\left(D,C\right) &\iff \chi\leqslant\min\left\{-\frac{\kappa -S}{T-\kappa},-\frac{T-\kappa}{\kappa -S}\right\} . \label{eq:branch1_2}
\end{align}
\end{subequations}
For $\chi\geqslant 1$ (\eq{branch1_1}), Proposition~\ref{prop:lambda_min} gives a $\left(\varphi ,\lambda\right)$-autocratic strategy if and only if $\lambda\geqslant\lambda_{\min}$, where
\begin{align*}
\lambda_{\min} &= 1-\frac{\min_{s_{Y}\in S_{Y}}\varphi\left(C,s_{Y}\right) -\max_{s_{Y}\in S_{Y}}\varphi\left(D,s_{Y}\right)}{\max\left\{ \substack{\max_{s_{Y}\in S_{Y}}\varphi\left(C,s_{Y}\right) -\max_{s_{Y}\in S_{Y}}\varphi\left(D,s_{Y}\right) , \\ \min_{s_{Y}\in S_{Y}}\varphi\left(C,s_{Y}\right) -\min_{s_{Y}\in S_{Y}}\varphi\left(D,s_{Y}\right)} \right\}} \nonumber \\
&= 1-\frac{\varphi\left(C,C\right) -\varphi\left(D,D\right)}{\max\left\{ \varphi\left(C,D\right) -\varphi\left(D,D\right) , \varphi\left(C,C\right) -\varphi\left(D,C\right) \right\}} .
\end{align*}
Since $\chi\geqslant 1$, we have $\varphi\left(C,D\right) -\varphi\left(D,D\right)\geqslant\varphi\left(C,C\right) -\varphi\left(D,C\right)$ if and only if $R+P\leqslant S+T$, thus
\begin{align*}
\lambda_{\min} &= 
\begin{cases}
\displaystyle 1-\frac{\left(\chi -1\right)\left(R-P\right)}{-S + \chi T -\left(\chi -1\right) P} & \displaystyle R+P\leqslant S+T ,\\[1.5em]
\displaystyle 1-\frac{\left(\chi -1\right)\left(R-P\right)}{\left(\chi -1\right) R - \chi S + T} & \displaystyle R+P>S+T .
\end{cases}
\end{align*}
Similarly, in the case of \eq{branch1_2}, we obtain
\begin{align*}
\lambda_{\min} &= 1-\frac{\min_{s_{Y}\in S_{Y}}\varphi\left(D,s_{Y}\right) -\max_{s_{Y}\in S_{Y}}\varphi\left(C,s_{Y}\right)}{\max\left\{ \substack{\max_{s_{Y}\in S_{Y}}\varphi\left(D,s_{Y}\right) -\max_{s_{Y}\in S_{Y}}\varphi\left(C,s_{Y}\right) , \\ \min_{s_{Y}\in S_{Y}}\varphi\left(D,s_{Y}\right) -\min_{s_{Y}\in S_{Y}}\varphi\left(C,s_{Y}\right)} \right\}} \nonumber \\
&= 1-\frac{\varphi\left(D,C\right) -\varphi\left(C,D\right)}{\max\left\{ \varphi\left(D,D\right) -\varphi\left(C,D\right) , \varphi\left(D,C\right) -\varphi\left(C,C\right) \right\}} \nonumber \\
&= 
\begin{cases}
\displaystyle 1-\frac{-\left(\chi +1\right)\left(T-S\right)}{S-\chi T+\left(\chi -1\right) P} & \displaystyle R+P\leqslant S+T, \\[1.5em]
\displaystyle 1-\frac{-\left(\chi +1\right)\left(T-S\right)}{-\left(\chi -1\right) R+\chi S-T} & \displaystyle R+P>S+T .
\end{cases}
\end{align*}

In the second branch (\eq{pd_chi_branch2}), we have $\varphi\left(C,C\right)\geqslant\varphi\left(C,D\right)$ and $\varphi\left(D,C\right)\geqslant\varphi\left(D,D\right)$. As in the first branch, we require $\varphi\left(C,C\right)$ and $\varphi\left(D,D\right)$ to have opposite signs, forcing $\kappa\in\left[P,R\right]$, which implies that $\varphi\left(C,C\right)\leqslant 0$ and $\varphi\left(D,D\right)\geqslant 0$ on this branch. Here, we see that
\begin{align*}
\lambda_{\min} &= 1-\frac{\varphi\left(D,D\right) -\varphi\left(C,C\right)}{\max\left\{\varphi\left(D,C\right) -\varphi\left(C,C\right) , \varphi\left(D,D\right) -\varphi\left(C,D\right) \right\}} \nonumber \\
&= 1 - \frac{\left(1-\chi\right)\left(R-P\right)}{\max\left\{\left(1-\chi\right) R - T + \chi S , S - \chi T -\left(1-\chi\right) P \right\}} \nonumber \\
&= \begin{cases}
\displaystyle 1-\frac{\left(1-\chi\right)\left(R-P\right)}{S - \chi T -\left(1-\chi\right) P} & \displaystyle R+P\leqslant S+T ,\\[1.5em]
\displaystyle 1-\frac{\left(1-\chi\right)\left(R-P\right)}{\left(1-\chi\right) R - T + \chi S} & \displaystyle R+P>S+T .
\end{cases}
\end{align*}

For all other parameter combinations, $\tau_{X}^{\pm}$ in Proposition~\ref{prop:lambda_min} can take values in strictly mixed actions, and in this case, we resort to Proposition~\ref{prop:tractability} to calculate $\lambda_{\min}$. For the standard (non-additive) payoff parameters of a prisoner's dilemma ($R=3$, $S=0$, $T=5$, and $P=1$), \fig{pd_heatmap} summarizes the enforceable lines.

\begin{figure}
    \centering
    \includegraphics[width=\textwidth]{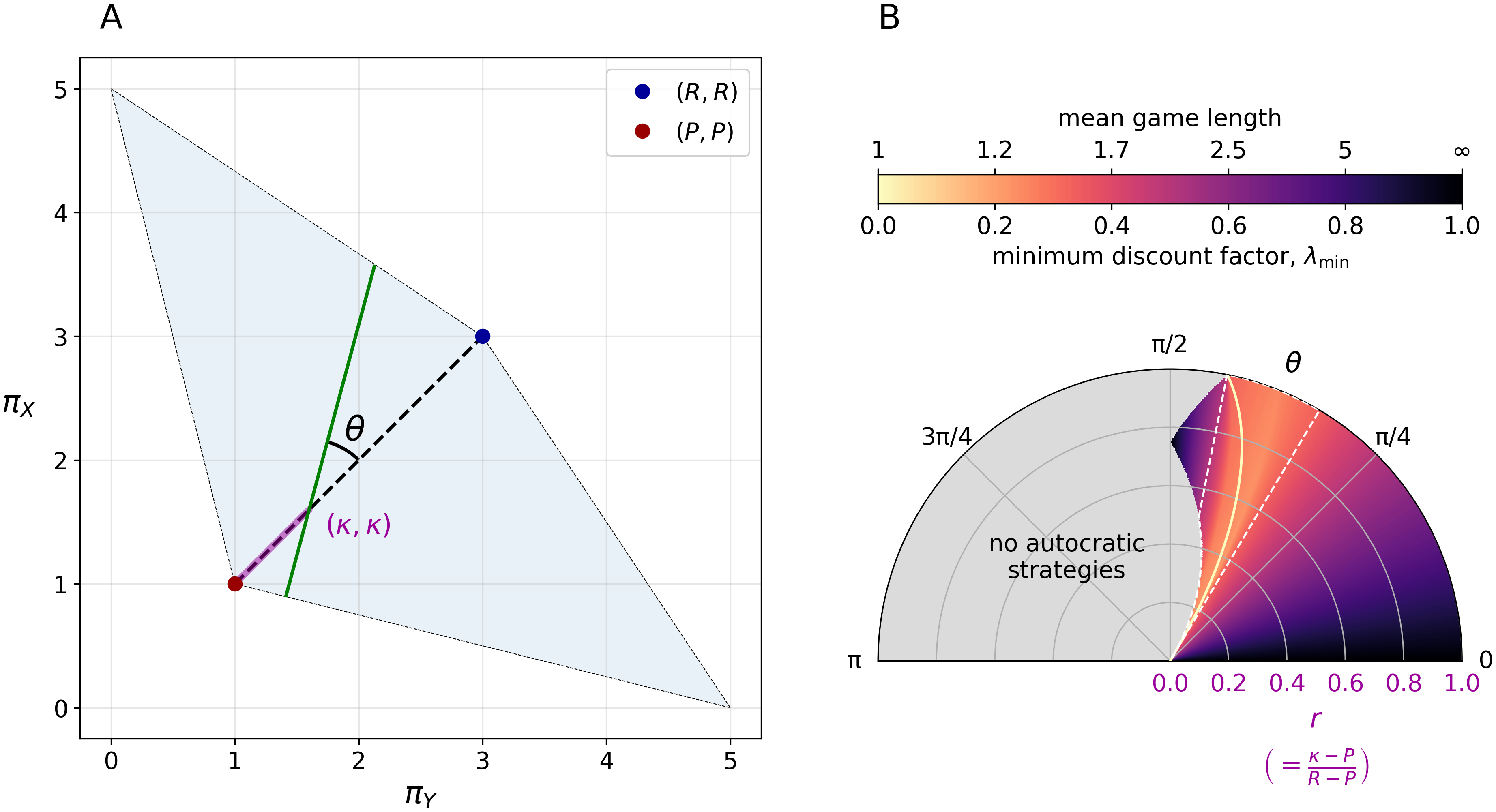}
    \caption{Heat map on enforceability of the linear payoff relationship $\varphi =\kappa -u_{X}\left(s_{X},s_{Y}\right) -\chi\left(\kappa -u_{Y}\left(s_{X},s_{Y}\right)\right)\equiv 0$, for the repeated prisoner's dilemma with $\left(R,S,T,P\right) =\left(3,0,5,1\right)$. $\theta$ is the angle between the payoff relationship (green) and the reference line (dashed), and $r$ represents the fraction (purple) of the reference line made up by the intersection point. The region enclosed by the white dashed line is the set of $\left(r,\theta\right)$ for which at least one of $\tau_{X}^{+}$ and $\tau_{X}^{-}$ is non-pure in the optimizer for $\lambda_{\min}$ (\eq{lambda_min}), and the solid line represents the lines that can be enforced by unconditional strategies. For this particular game, we have $\kappa =P+r\left(R-P\right)$ and $\chi =\tan\left(\theta +\pi /4\right)$.}
    \label{fig:pd_heatmap}
\end{figure}

\subsection{Nonlinear donation game}
The donation game has a relatively simple structure and has been studied extensively in the context of ZD strategies \citep{press:PNAS:2012,hilbe:PNAS:2013,mcavoy:PNAS:2016}. A slightly more realistic social dilemma occurs when the agents have multiple choices when it comes to paying up a cost and gaining a benefit. For simplicity, consider the three-action extension of the donation game, with payoff matrix
\begin{align*}
\bordermatrix{%
& C_{1} & C_{2} & D \cr
C_{1} &\ b_{1}-c_{1},\ b_{1}-c_{1} & \ b_{2}-c_{1},\ b_{1}-c_{2} & \ -c_{1},\ b_{1} \cr
C_{2} &\ b_{1}-c_{2},\ b_{2}-c_{1} & \ b_{2}-c_{2},\ b_{2}-c_{2} & \ -c_{2},\ b_{2} \cr
D &\ b_{1},\ -c_{1} & \ b_{2},\ -c_{2} & \ 0,\ 0 \cr
} .
\end{align*}
Here, we assume that $0<c_{1}<c_{2}$, $0<b_{1}<b_{2}$, $c_{1}<b_{1}$, and $c_{2}<b_{2}$, but that $b_{2}-c_{2}<b_{1}-c_{1}$. In other words, playing $C_{2}$ is more costly, but more beneficial than playing $C_{1}$. Yet, when both players choose the same level, mutual ``$C_{1}$'' strictly Pareto-dominates mutual ``$C_{2}$.'' In this game, ``nonlinear'' refers to the schedule of benefits and costs, since their relative rankings imply that the benefit is not simply a linear function of the cost.

\begin{figure}
    \centering
    \includegraphics[width=\textwidth]{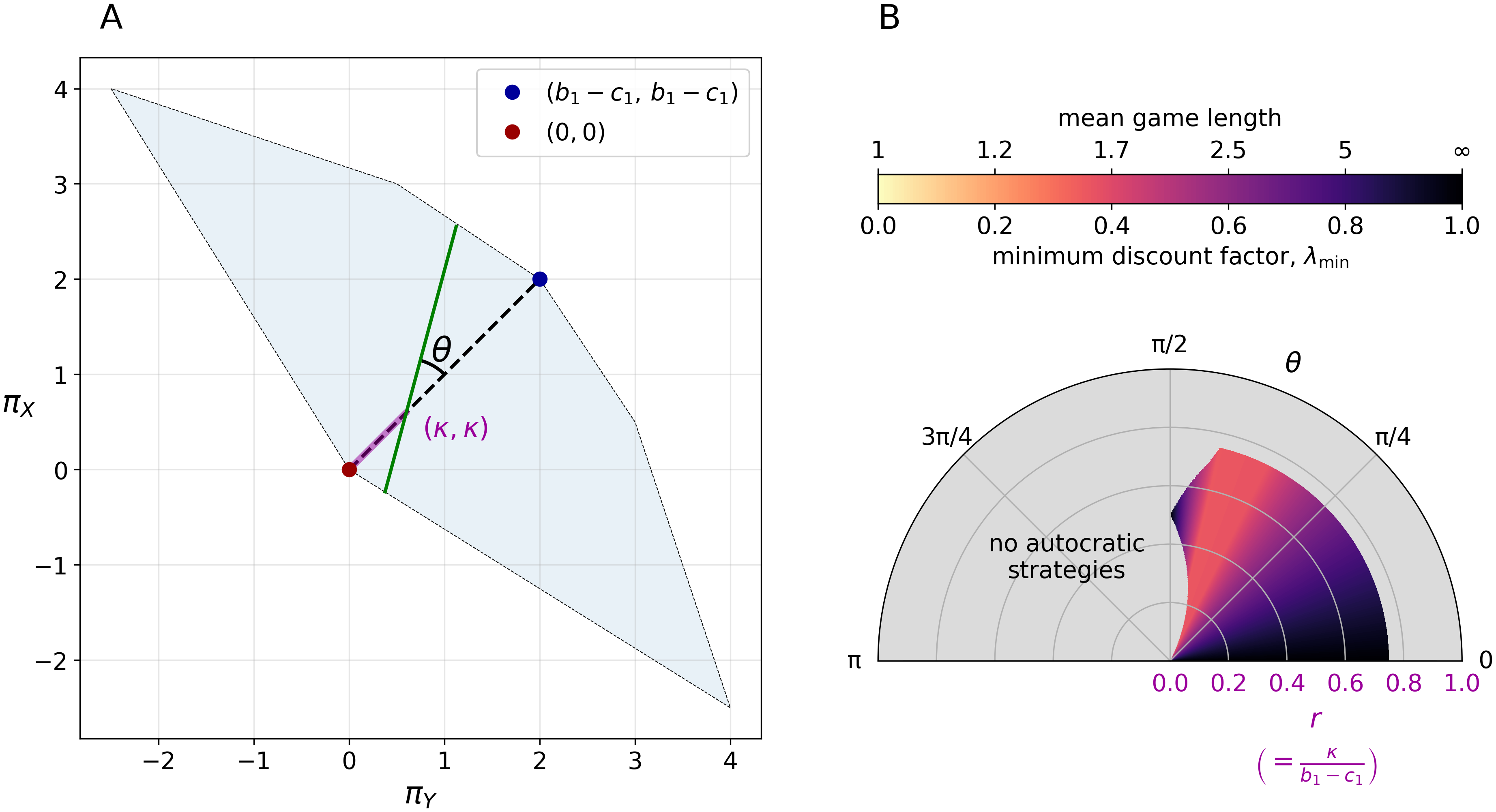}
    \caption{Enforceability of linear payoff relationships in the three-action nonlinear donation game with parameters $b_{1}=3$, $c_{1}=1$, $b_{2}=4$, $c_{2}=2.5$ (satisfying $b_{2}-c_{2}<b_{1}-c_{1}$). The heat map shows the minimum discount factor $\lambda_{\min}$ required to enforce $\varphi =\kappa -u_{X} -\chi\left(\kappa -u_{Y}\right)\equiv 0$ across different values of $\kappa$ and $\chi$. $\theta$ is the angle between the payoff relationship (green) and the reference line (dashed), and $r$ represents the fraction (purple) of the reference line made up by the intersection point. Due to the game's additive structure, any enforceable relationship can be implemented using a two-point reactive strategy, significantly simplifying the strategy space compared to general memory-one approaches. For this game, we have $\kappa =r\left(b_{1}-c_{1}\right)$ and $\chi =\tan\left(\theta +\pi /4\right)$.}
    \label{fig:ndg_heatmap}
\end{figure}

The game's payoff structure is additive: for each action profile, we can write $u_{X}\left(s_{X},s_{Y}\right) =\phi_{X}\left(s_{X}\right) +\phi_{Y}\left(s_{Y}\right)$. By symmetry, $u_{Y}$ is also additive. Consequently, for linear payoff relationships $\varphi =\kappa -u_{X} -\chi\left(\kappa -u_{Y}\right)$, the function $\varphi$ inherits this additive structure.

Since this game is additive and $\varphi$ is linear, Theorem~\ref{thm:additive} tells us that any enforceable payoff relationship can be implemented using a reactive strategy that conditions solely on the opponent's most recent action. Moreover, by Lemma~\ref{lem:inequalities_additive}, enforcement can be achieved using a two-point reactive strategy that mixes between two fixed distributions based on the opponent's last action.

The presence of three actions introduces additional strategic complexity compared to the standard two-action donation game. While mutual cooperation at level $C_{1}$ yields the Pareto-efficient outcome $\left(b_{1}-c_{1},b_{1}-c_{1}\right)$, players face a tension between contributing at the higher level $C_{2}$ (which benefits the opponent more) and defecting entirely. This structure creates richer possibilities for autocratic strategies, as player $X$ can condition her response not just on whether $Y$ cooperated, but also on the level of cooperation chosen. \fig{ndg_heatmap} illustrates that the enforceability landscape extends naturally from the two-action case, with the minimum discount factor varying smoothly across the parameter space $\left(\kappa ,\chi\right)$.

\subsection{Fairness and equality in an asymmetric donation game}
In symmetric games like the standard prisoner's dilemma, ``fair'' strategies enforcing $\varphi =u_{Y}-u_{X}\equiv 0$ are well-studied \citep{press:PNAS:2012}, with TFT being the canonical example. And we have seen that payoff equality can be enforced in expectation if and only if $\lambda\rightarrow 1$ (Proposition~\ref{prop:symmetricproperty}). However, the distinction between equality and ``fairness'' becomes crucial in asymmetric social dilemmas, where two agents might have different abilities or resources. Consider the asymmetric donation game with payoff matrix
\begin{align*}
\bordermatrix{%
& C & D \cr
C &\ b_{Y}-c_{X},\ b_{X}-c_{Y} & \ -c_{X},\ b_{X} \cr
D &\ b_{Y},\ -c_{Y} & \ 0,\ 0 \cr
} ,
\end{align*}
where $b_{X}>c_{X}>0$ and $b_{Y}>c_{Y}>0$. In other words, as cooperators, $X$ pays $c_{X}$ to donate $b_{X}$ and $Y$ pays $c_{Y}$ to donate $b_{Y}$. As defectors, they both pay nothing and donate nothing.

In this game, the natural reference line is not $\pi_{X}=\pi_{Y}$ but rather the line through $\left(0,0\right)$ and $\left(b_{X}-c_{Y},b_{Y}-c_{X}\right)$ since this represents the line between payoffs for mutual defection and payoffs for mutual cooperation. A strategy is ``fair'' if it enforces a payoff relationship along this reference line, reflecting proportional sharing that accounts for the asymmetric costs and benefits. In contrast, a strategy emphasizing equality enforces $\pi_{X}=\pi_{Y}$ regardless of the players' differing contributions.

Like the standard donation game, this asymmetric variant is additive. By Theorem~\ref{thm:additive}, any enforceable payoff relationship can be implemented using a reactive strategy. The key distinction from symmetric games emerges when comparing fairness and equality: while TFT enforces equality in the symmetric donation game (where $b_{X}=b_{Y}$ and $c_{X}=c_{Y}$), in the asymmetric case these two objectives diverge. Fair strategies that respect the natural reference line require $\lambda\rightarrow 1$, making this constraint practically infeasible in finite-horizon interactions. By contrast, enforcing equality can be achieved with smaller discount factors (in this case, when the natural reference line lies below the line of equality), as shown in \fig{adg_heatmap}.

\begin{figure}
    \centering
    \includegraphics[width=\textwidth]{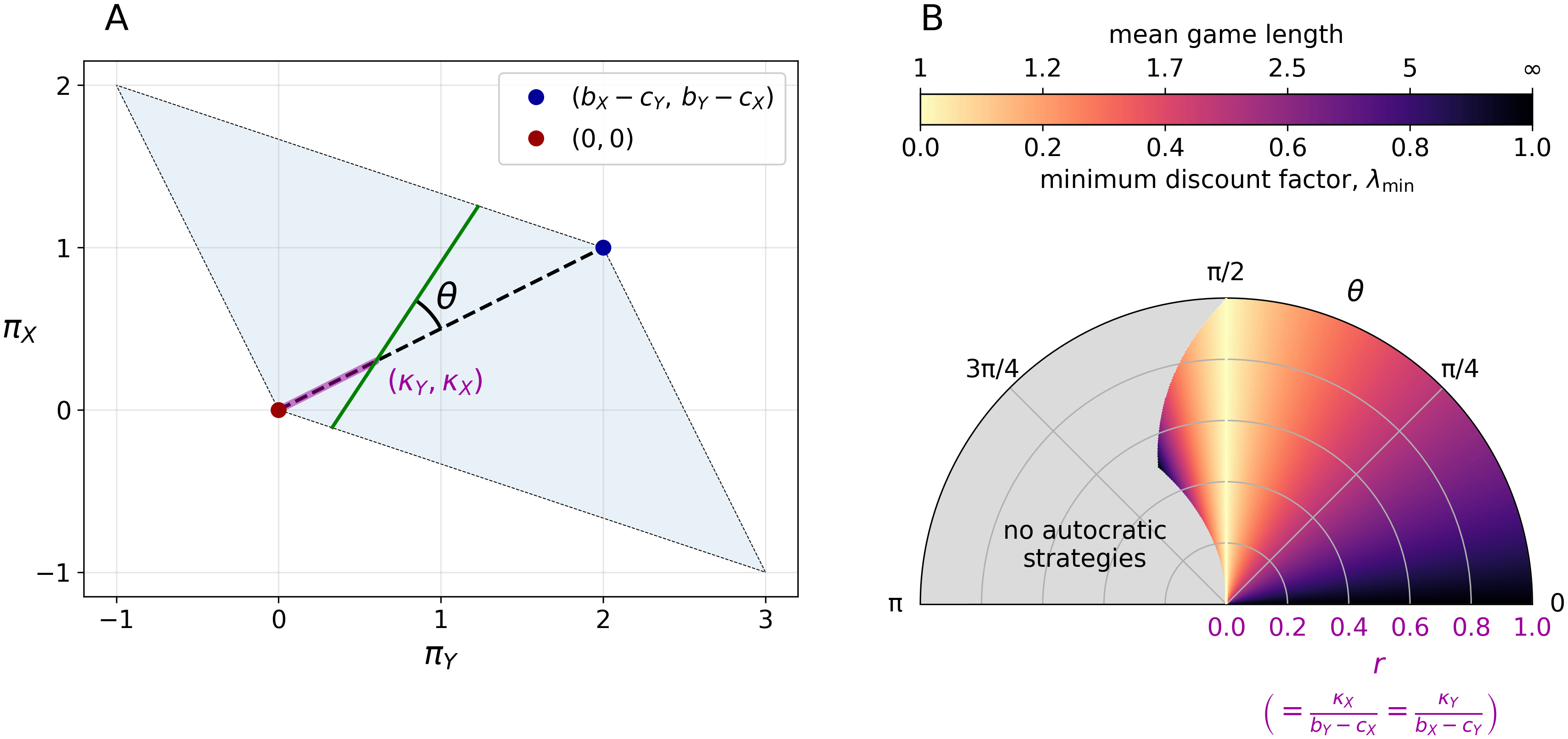}
    \caption{Enforceability in the asymmetric donation game with $b_{X}=3$, $c_{X}=1$, $b_{Y}=2$, and $c_{Y}=1$. The heat map displays the minimum discount factor $\lambda_{\min}$ for enforcing $\varphi =\kappa_{X}-u_{X}-\chi\left(\kappa_{Y}-u_{Y}\right)\equiv 0$. $\theta$ is the angle between the payoff relationship (green) and the reference line (dashed), and $r$ represents the fraction (purple) of the reference line made up by the intersection point. The natural reference line (indicated by the dashed line in the payoff space) connects mutual defection $\left(0,0\right)$ to mutual cooperation $\left(b_{X}-c_{Y},b_{Y}-c_{X}\right)$, reflecting the asymmetric costs and benefits. ``Fair'' strategies enforcing proportional sharing along the reference line require $\lambda\rightarrow 1$, while enforcing equality ($\pi_{X}=\pi_{Y}$, corresponding to $\chi =1$) is more readily achievable. For this game, we have $\kappa_{X}=r\left(b_{Y}-c_{X}\right)$ (and $\kappa_{Y}=r\left(b_{X}-c_{Y}\right)$) and $\chi =\tan\left(\theta +\arctan\frac{b_{Y}-c_{X}}{b_{X}-c_{Y}}\right)$.}
    \label{fig:adg_heatmap}
\end{figure}

\subsection{Hawk-dove game: a ``weak'' social dilemma}
Next, we consider a weak social dilemma \citep{hauert:JTB:2006a} known as the ``hawk-dove'' \citep{maynardsmith:Nature:1973} (or ``snowdrift'' \citep{sugden1986economics,hauert:Nature:2004}) game. The toy narrative behind this game is that there is an aggressive type (hawk) and a peaceful type (dove), and there is a common resource of value $V$ for which two individuals compete. When two doves meet, they share the resource, each receiving $V/2$. When a hawk meets a dove, the hawk takes the resource and leaves nothing for the dove. However, when two hawks meet, they fight and incur a cost that effectively lowers the value of the resource by $C>V$, resulting in $\left(V-C\right) /2$ to each.

\fig{hd_heatmap} shows which linear relationships are enforceable in this game, using parameters $V=2$ and $C=4$ as examples. Here, like in the non-additive prisoner's dilemma, we observe regions for which $\lambda_{\min}$ can be attained only when at least one of the two mixed actions in a two-point reactive learning strategy is non-pure.

\begin{figure}
    \centering
    \includegraphics[width=\textwidth]{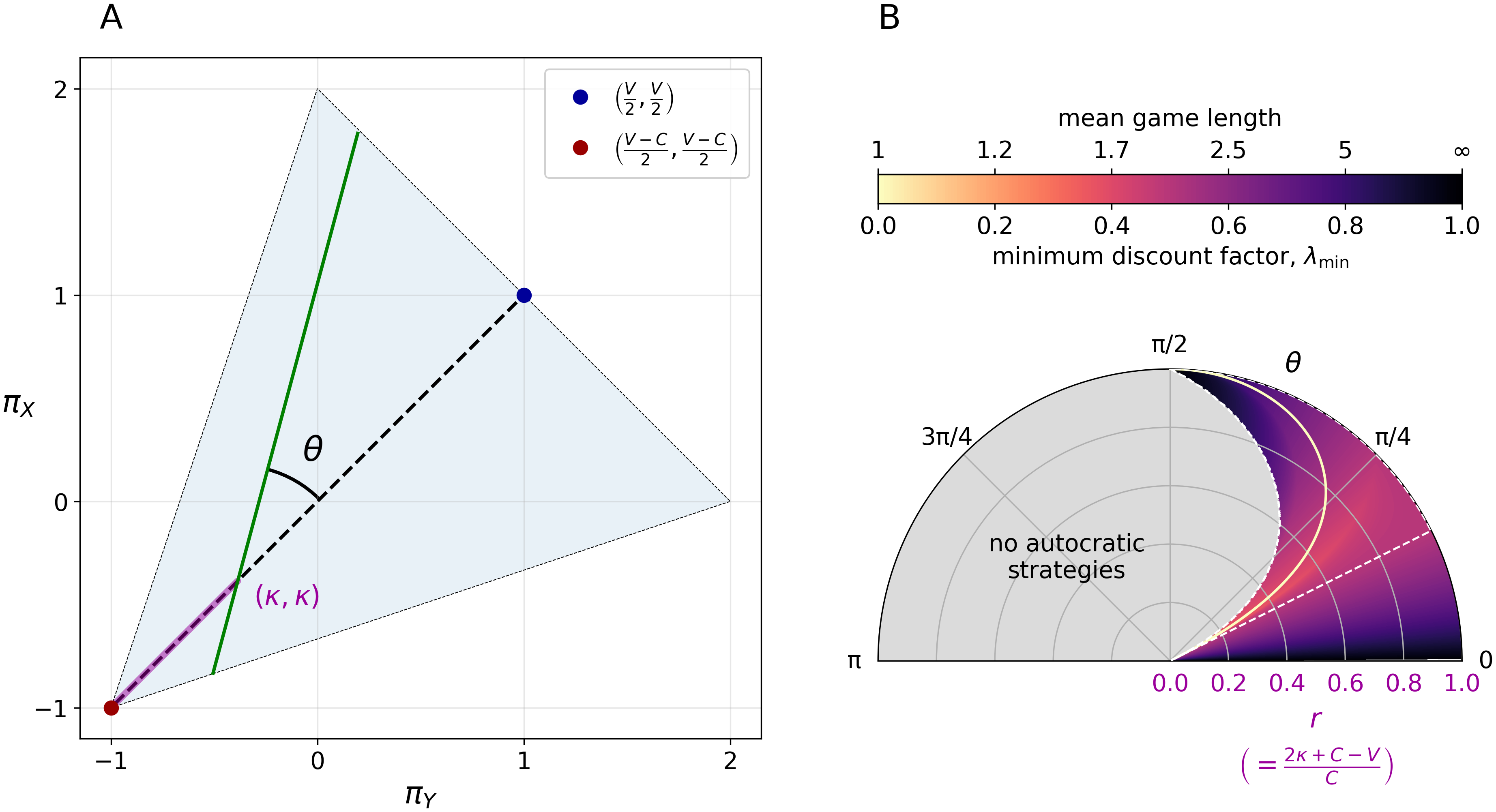}
    \caption{Heat map on enforceability of the linear payoff relationship $\varphi =\kappa -u_{X}\left(s_{X},s_{Y}\right) -\chi\left(\kappa -u_{Y}\left(s_{X},s_{Y}\right)\right)\equiv 0$, for the repeated hawk-dove game with resource value $V=2$ and conflict cost $C=4$. $\theta$ is the angle between the payoff relationship (green) and the reference line (dashed), and $r$ represents the fraction (purple) of the reference line made up by the intersection point. The region enclosed by the white dashed line is the set of $\left(r,\theta\right)$ for which at least one of $\tau_{X}^{+}$ and $\tau_{X}^{-}$ is non-pure in the optimizer for $\lambda_{\min}$ (\eq{lambda_min}), and the solid line represents the lines that can be enforced by unconditional strategies. For this game, we have $\kappa =\left(V-C\right) /2+r\left(C/2\right)$ and $\chi =\tan\left(\theta +\pi /4\right)$.}
    \label{fig:hd_heatmap}
\end{figure}

\subsection{Non-affine payoff relationships}
Although our theory works for arbitrary $\varphi$, the examples so far have focused on $\varphi$ arising as linear combinations of the stage payoff functions, $u_{X}$ and $u_{Y}$. We conclude our examples with the simplest possible such $\varphi$ that does not conform to this structure. Suppose that $\varphi_{\circ}\coloneqq\alpha u_{X}+\beta u_{Y}+\gamma$ for some $\alpha ,\beta ,\gamma\in\mathbb{R}$, which we can decompose as $\varphi_{\circ}=\varphi_{\circ}^{+}-\varphi_{\circ}^{-}$, where $\varphi_{\circ}^{+}\coloneqq\max\left\{0,\varphi_{\circ}\right\}$ and $\varphi_{\circ}^{-}\coloneqq\max\left\{0,-\varphi_{\circ}\right\}$. From these components, we let $\varphi\coloneqq\varphi_{\circ}^{+}-\ell\varphi_{\circ}^{-}$, where $\ell >0$. If $X$ uses a strategy that enforces $\varphi\equiv 0$ against all opponents, then
\begin{align*}
\mathbb{E}_{\sigma_{X},\sigma_{Y}}\left[\left(1-\lambda\right)\sum_{t=0}^{\infty}\lambda^{t}\varphi_{\circ}\left(s_{X}^{t},s_{Y}^{t}\right)\right] = \left(\ell -1\right)\mathbb{E}_{\sigma_{X},\sigma_{Y}}\left[\left(1-\lambda\right)\sum_{t=0}^{\infty}\lambda^{t}\varphi_{\circ}^{-}\left(s_{X}^{t},s_{Y}^{t}\right)\right] .
\end{align*}
In the case $\ell =1$, this strategy degenerates to an autocratic strategy enforcing $\alpha\pi_{X}+\beta\pi_{Y}+\gamma =0$. For $\ell >1$, we get $\alpha\pi_{X}+\beta\pi_{Y}+\gamma\geqslant 0$ for all $Y$. For $\ell <1$, we get $\alpha\pi_{X}+\beta\pi_{Y}+\gamma\leqslant 0$ for all $Y$.

Generically, for $\ell\neq 1$, such autocratic strategies enforce non-affine relationships in that the regions they enforce have full dimension. \fig{nonlinear} illustrates the effects of $\ell$ on a linear relationship of slope $2$ (in the $\left(\pi_{Y},\pi_{X}\right)$-plane) passing through the point of mutual cooperation in the standard repeated prisoner's dilemma. What is interesting about these strategies is that they allow autocratic strategies to be used for shaping strategic incentives outside of the standard linear, ``knife-edge'' case. Using such a strategy, $X$ cannot necessarily restrict $\varphi_{\circ}^{-}$ to a codimension one subspace, but since $\varphi_{\circ}^{-}$ is non-negative, $X$'s control over $\varphi$ allows her to unilaterally bound $\varphi_{\circ}$. We note that despite the differences in strategy payoff-geometry in the three panels of \fig{nonlinear}, all three induce the same incentives to a selfish opponent, namely that greedy optimization induces an outcome of mutual cooperation for both players.

\begin{figure}
\centering
\includegraphics[width=\textwidth]{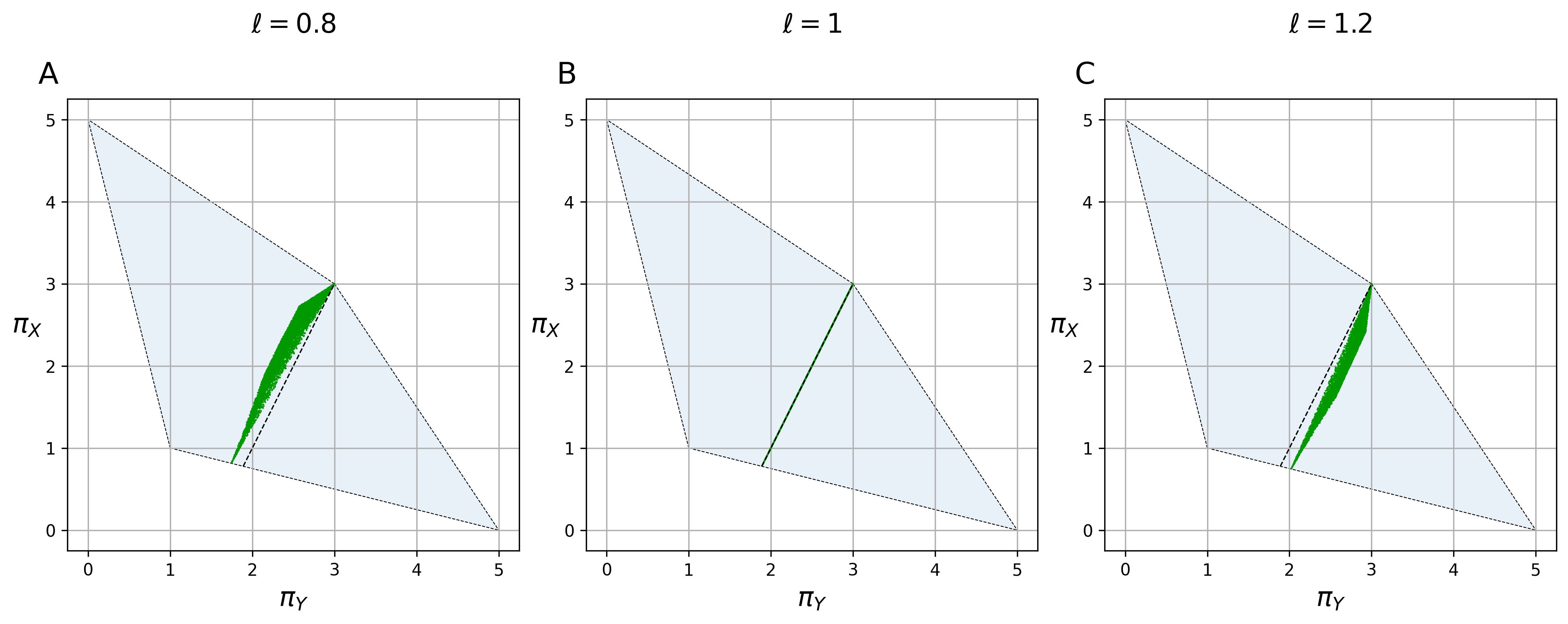}
\caption{Autocratic strategies enforcing non-affine payoff relationships. For the repeated prisoner's dilemma with $\left(R,S,T,P\right) =\left(3,0,5,1\right)$ and $\lambda =0.95$, we consider the relationship $3-\pi_{X}=2\left(3-\pi_{Y}\right)$ on expected payoffs (black dashed line). By splitting the corresponding function $\varphi_{\circ}\coloneqq 3-u_{X}-2\left(3-u_{Y}\right)$ into a difference $\varphi_{\circ}^{+}-\varphi_{\circ}^{-}$, where $\varphi_{\circ}^{+}\coloneqq\max\left\{0,\varphi_{\circ}\right\}$ and $\varphi_{\circ}^{-}\coloneqq\max\left\{0,-\varphi_{\circ}\right\}$, we can ask when $X$ can unilaterally ensure that the piecewise linear function $\varphi\coloneqq\varphi_{\circ}^{+}-\ell\varphi_{\circ}^{-}$ vanishes, on average, against all opponents. Solving for these strategies for three different values of $\ell$ ($0.8$ in A; $1$ in B; and $1.2$ in C) gives the payoff regions depicted here, with the middle value of $\ell =1$ corresponding to the degenerate case of a truly linear payoff relationship (i.e., a ``generous'' \citep{stewart:PNAS:2012,stewart:PNAS:2013,hilbe:JTB:2015} ZD strategy).}
\label{fig:nonlinear}
\end{figure}

\section{Discussion}\label{sec:discussion}
Our results establish a complete characterization of enforceable payoff relationships in discounted repeated games. The main theoretical contribution answers a question about the role of strategic complexity: extending memory beyond a simple reactive learning structure provides no additional power for enforcing payoff constraints. Any payoff relationship that can be enforced by a strategy of arbitrary memory can be implemented using a two-point reactive learning strategy. This universality result demonstrates that the space of reactive learning strategies captures all possible enforceable constraints, regardless of the opponent's strategic sophistication.

The universality of reactive learning strategies within the class of autocratic strategies is perhaps surprising given the apparent richness of the space of behavioral strategies with arbitrary memory. On the other hand, since reactive learning strategies can be thought of as longer-memory behavioral strategies with a ``right-invariance'' property, they give a mechanism for ``rolling up'' longer histories of simple information into shorter memories of richer information.

One consequence is that a search for autocratic strategies need only consider the computationally tractable space of reactive learning strategies, rather than the vast and intractable space of behavioral strategies. Our constructive characterization provides explicit formulas for both verifying enforceability and computing the minimum discount factor required. For arbitrary functions, $\varphi$, these computations can be performed in polynomial time using linear programming. The computational tractability of our characterization stands in stark contrast to the general intractability of analyzing finite-memory strategies. For memory-$m$ strategies in games with $k$ actions per player, the strategy space has dimension $\left(k-1\right) k^{2m}$, making exhaustive analysis infeasible even for modest values of $m$ and $k$ \citep{hauert1997effects}. The reduction to two-point reactive learning strategies eliminates this exponential dependence on memory length, reducing the problem to a fixed-dimensional space determined only by the number of actions in the stage game.

These findings resolve several open questions in the literature on ZD strategies. First, we provide a definitive answer to the question raised by \citet{hilbe:PLOSONE:2013} regarding the existence of autocratic strategies in discounted games. Our necessary and sufficient conditions establish precisely when a payoff relationship is enforceable, and our formulas for the minimum discount factor extend the memory-one analysis of \citet{hilbe:GEB:2015} to strategies of arbitrary memory. Second, we demonstrate that the initial action plays a crucial role in the construction of autocratic strategies, contrasting with undiscounted settings where initial conditions are often irrelevant. The structural properties we establish reveal that the space of autocratic strategies has favorable geometric and algorithmic properties. The dichotomy for symmetric relationships shows that fairness constraints are fundamentally incompatible with discounting: symmetric relationships either are enforceable with unconditional strategies or require the limiting case of an infinite horizon. This explains why strategies like tit-for-tat, which enforce equal payoffs in the prisoner's dilemma, lose this property in any discounted setting.

Although we have framed our results primarily for two-player games, the framework extends naturally to multiplayer settings where coalitions of players coordinate to enforce payoff constraints on the larger group. A coalition $I\subseteq\left\{1,\dots ,N\right\}$ can use a correlated strategy to enforce a linear relationship $\sum_{i=1}^{N}\alpha_{i}\pi_{i}+\alpha_{0}=0$ on the expected payoffs of all players. Our characterization carries over to this setting: any enforceable coalitional constraint can be implemented using a reactive learning strategy for the coalition that conditions on the previous actions of players outside the coalition and the coalition's own previous mixed action. In fact, one need only replace $X$ by $I$ (the coalition) and $Y$ by $-I$ (the anti-coalition) in Theorem~\ref{th:mainresultDiscounted}. The only subtlety is that a strategy for $I$ allows for correlations, in the sense that it is a map $\sigma_{I}:\mathcal{H}\rightarrow\Delta\left(S_{I}\right)$ rather than from $\mathcal{H}$ to $\prod_{i\in I}\Delta\left(S_{i}\right)$. This, we believe, is reasonable, when a subset of a larger group strategically coordinates to control outcomes for others. There has been only limited work in the area of multiplayer payoff enforcement \citep{hilbe:JTB:2015,pan:SR:2015,govaert2019zero,chen:JTB:2022}.

This coalitional perspective is especially relevant for applications in multi-agent reinforcement learning and algorithmic game theory \citep{su2025emergence,bernheim1987coalition,sandholm1996multiagent}. ZD strategies allow a coalition to effectively reshape the incentive landscape faced by learning agents outside the coalition. If external agents are optimizing their policies using gradient-based methods or other adaptive algorithms, the coalition can enforce constraints that guide the learning dynamics toward desirable equilibria or prevent convergence to undesirable outcomes. Doing so is relevant in, for example, specific domains such as climate agreements \citep{nordhaus2021dynamic} and algorithmic collusion \citep{calvano2020artificial}. Our results provide concrete tools for designing such coalitional strategies and understanding their limitations.

Autocratic strategies are closely linked to evolutionary game theory, where populations of agents adapt their strategies over time through selection, mutation, or learning \citep{weibull:MIT:1995,sigmund1999evolutionary,stewart:PNAS:2013,chen:JTB:2014}. In fact, this is the setting in which linear payoff constraints were first recognized \citep{press:PNAS:2012,stewart:PNAS:2012}. An autocratic strategy that enforces a favorable payoff relationship can maintain a consistent advantage over a wide range of opponents, potentially allowing it to persist in evolutionary competition despite not being a Nash equilibrium of the stage game. The robustness of reactive learning strategies (in the sense of their ability to enforce constraints against arbitrary opponents) suggests they may be especially resilient to invasion attempts and environmental perturbations. In learning dynamics, the presence of an autocratic player fundamentally alters the optimization landscape for other agents. If player $Y$ is adapting through reinforcement learning or evolutionary search, player $X$'s autocratic strategy constrains the payoffs $Y$ can achieve along any learning trajectory. Understanding these constraints is useful for predicting the dynamics of multi-agent learning systems and for designing interventions that guide learning toward socially beneficial equilibria.

Several important questions remain open. First, while we characterize enforceable relationships for arbitrary payoff functions, our explicit formulas for minimum discount factors apply most directly to finite action spaces. Extensions to continuous action spaces or games with state-dependent payoffs would require additional technical machinery, though we expect the fundamental principles to carry over (Remarks \ref{rem:remark3} and \ref{rem:remark4}). Second, our framework assumes that players observe actions perfectly and condition their strategies on these observations. In settings with imperfect monitoring or private information, enforcement becomes more subtle, and the relationship between memory and enforceability may change \citep{mamiya2020zero}. Future work could explore the robustness of autocratic strategies to noise and misperception, characterize the set of enforceable relationships when both players simultaneously attempt to enforce constraints, and investigate the evolutionary stability of populations containing autocratic strategists. While much of the existing literature focuses on linear payoff relationships, our framework naturally extends to arbitrary non-affine constraints. The extension to non-affine constraints is particularly relevant for settings where players care about relative performance, inequity aversion, or other behavioral considerations that induce nonlinear preferences over payoff profiles \citep{fehr1999theory}. Our results suggest that even in these complex preference structures, the fundamental limits on enforceability are determined by the same geometric separation conditions that govern linear relationships.

\section*{Acknowledgments}
We thank Christian Hilbe for helpful comments and suggestions. This research was supported in part by grants from the NSF (DMS-2235451) and Simons Foundation (MPS-NITMB-00005320) to the NSF-Simons National Institute for Theory and Mathematics in Biology (NITMB).

\end{document}